%% file: fft.tex
\documentclass[twoside,11pt]{article}

\usepackage{jsat}

\usepackage{fancyvrb}

\DefineVerbatimEnvironment{code}{Verbatim}{fontsize=\small}
\usepackage{url}
\usepackage{graphicx}
\usepackage{amsmath}
\usepackage{amssymb}
\usepackage{amsthm}
\usepackage{amsfonts}
\usepackage{amscd}
\usepackage{tabularx}
\usepackage{epsfig}
\usepackage{color}
\usepackage[ruled,linesnumbered,noend]{algorithm2e}
\usepackage{listings}
\usepackage{capt-of}
\usepackage{nicefrac}
\usepackage{float}

\theoremstyle{definition}
\newtheorem{definition}{Definition}[section]
\newtheorem{prop}{Property}[section]
\newtheorem{example}{Example}[section]
\newtheorem{theorem}{Theorem}[section]

\usepackage[english]{babel}
\usepackage[T1]{fontenc}
\usepackage[ansinew]{inputenc}
\usepackage{lmodern}	% font definition
\usepackage{tikz}
\usetikzlibrary{decorations.pathmorphing} % noisy shapes
\usetikzlibrary{fit,backgrounds,shapes}					% fitting shapes to coordinates

\newcommand{\wnf}{\psi_n}
\newcommand{\mdf}{\text{ (mod }\nicefrac{n}{4}\text{)}}
\renewcommand\Re{\operatorname{\mathfrak{Re}}}
\renewcommand\Im{\operatorname{\mathfrak{Im}}}

\begin{document}

\title{Generating and Searching Families of FFT Algorithms\thanks{August 30, 2011, preprint. To appear in the Journal of Satisfiability, Boolean Modeling and Computation.}}

\author{\name Steve Haynal \email steve@softerhardware.com \\
	 \addr SofterHardware
       \AND
       \name Heidi Haynal \email heidi.haynal@wallawalla.edu \\
	 \addr Department of Mathematics and Computer Science \\
	 Walla Walla University}
       
\maketitle

\begin{abstract}

A fundamental question of longstanding theoretical interest is to prove the lowest exact count of real additions and multiplications required to compute a power-of-two discrete Fourier transform (DFT). For 35 years the split-radix algorithm held the record by requiring just $4n \log_2n-6n+8$ arithmetic operations on real numbers for a size-$n$ DFT, and was widely believed to be the best possible. Recent work by Van Buskirk and Lundy demonstrated improvements to the split-radix operation count by using multiplier coefficients or ``twiddle factors'' that are not  $n^{th}$ roots of unity for a size-$n$ DFT.

This paper presents a Boolean Satisfiability-based proof of the lowest operation count for certain classes of DFT algorithms. First, we present a novel way to choose new yet valid twiddle factors for the nodes in flowgraphs generated by common power-of-two fast Fourier transform algorithms, FFTs. With this new technique, we can generate a large family of FFTs realizable by a fixed flowgraph. This solution space of FFTs is cast as a Boolean Satisfiability problem, and a modern Satisfiability Modulo Theory solver is applied to search for FFTs requiring the fewest arithmetic operations. Surprisingly, we find that there are FFTs requiring fewer operations than the split-radix even when all twiddle factors are $n^{th}$ roots of unity.

\end{abstract}


%\published{March 2011}{August 2011}{November 2004}
%\published{Preprint submitted to the Journal on Satisfiability, Boolean Modeling and Computation, March 2011}{}{}

\section{Introduction}

In 1965 Cooley and Tukey\cite{cooley1965algorithm} started a revolution in digital signal processing when they introduced their fast Fourier transform algorithm (FFT). Their FFT required only $O(n \log n)$ addition and multiplication floating-point operations on real numbers, or FLOPs, rather than the $O(n^2)$ FLOPs required to directly compute a discrete Fourier transform. Although discovered previously\cite{1866-gauss-3}\cite{rockmore2002fft}, it was Cooley and Tukey's timing, which coincided with the beginning of widespread use and availability of digital computers, that led to its success. The FFT and related algorithms have now found a wide range of application, including electroacoustic music, audio signal processing, medical imaging, image processing, pattern recognition, computational chemistry, error correcting codes, spectral methods for PDEs and harmonic analysis\cite{1988-brigham}\cite{rockmore2002fft}.

After the FFT's introduction, there was considerable work on further lowering the FLOP count. This was of particular interest since addition and especially multiplication were expensive with the computer hardware available at that time. One result that stands out is the work done by Yavne\cite{1968-yavne} in developing an initial split-radix\cite{duhamel1990fast}\cite{sorensen2003computing} algorithm with a $4n \log_2n-6n+8$ FLOP count for a size-$n$ FFT where $n$ is some power of two, $n=2^m$. Other important work minimized the number of multiplications but not the total arithmetic complexity\cite{heideman2003number}\cite{heideman1988multiplicative}\cite{duhamel2002algorithms}\cite{winograd1980arithmetic}. In 2004 Lundy and Van Buskirk\cite{lundy2007new} demonstrated improvements to the split-radix operation count by using using constant complex-value multiplier coefficients or \textbf{twiddle factors} that are not $n^{th}$ roots of unity for a size-$n$ DFT. Frigo and Johnson\cite{johnson2007modified} generalized Van Buskirk's pioneering work in the context of optimizing the conjugate-pair split-radix algorithm\cite{kamarconjugate}. Bernstein\cite{bernstein2007tangent} then described Johnson's algorithm, which is distinct from Van Buskirk's, in terms of algebraic twisting and named it the tangent FFT. In this paper, we refer to Johnson's algorithm and Bernstein's variation of it, differing only in decimation in time versus frequency, as the \textbf{tangent FFT}. Van Buskirk's algorithm and the tangent FFT exhibits a modest $(\sim5.6\%)$ reduction in FLOP count when compared to the split-radix, requiring roughly $\frac{34}{9} n \log_2 n$ operations rather than the previous $4n \log_2n-6n+8$.

This paper presents a proof of the lowest FLOP count for certain classes of DFTs. It is beyond the scope of this paper to consider all possible DFTs in our proof. Instead, we focus on the common power-of-two complex FFTs and the flowgraphs\cite{cnx:m16352} implied by them. This scope still includes a rich set of FFTs as our experiments confirm what others have seen\cite{duhamel1990fast}; common power-of-two complex FFTs (radix-2, radix-4, decimation-in-time or decimation-in-frequency split-radix, conjugate split-radix, classic or any twisted) all exhibit the same flowgraph structure (they are graph isomorphisms) but have different twiddle factors assigned to the flowgraph nodes. Furthermore, we restrict our scope to FFTs where twiddle factors are $n^{th}$ root of unity. This excludes Van Buskirk's algorithm and the tangent FFT, but we still can show that other algorithms with lower FLOP count than the traditional split-radix exist.

In 1962, a few years before Cooley and Tukey introduced their FFT, Davis \emph{et al.} developed a machine program for theorem proving\cite{davis1962machine}, now referred to as the Davis--Putnam--Logemann--Loveland or DPLL algorithm, which is still at the core of modern Boolean Satisfiability or SAT solvers. In the past decade, several advances and refinements to DPLL have made it practical for larger problems\cite{silva1997grasp}\cite{moskewicz2001chaff}\cite{een2006translating}. New conflict-driven clause-learning (CDCL) SAT solvers, which incorporate these recent advances, are now commonly used in industry to verify hardware and software correctness. Current SAT research benefits from industrial sponsorship and an active community, which organizes conferences and competitions, creates challenge problems, and defines problem formats\cite{SATLive}\cite{SATCompetitions}\cite{ranise2006smt}. Recently, Satisfiability Module Theories (SMT) generalize SAT beyond binary variables to incorporate higher-level theories such as bitvectors, lists and arrays\cite{nieuwenhuis2006solving}\cite{ranise2006smt}. SMT solvers range from those that simply reduce a higher-level theory to Boolean logic for a SAT solver, to those that extend the core decision procedure to accommodate higher-level theories. 

In this paper, we apply a modern SMT solver to find a lowest FLOP count algorithm for the class of FFTs considered. First, we present a novel way to choose new yet valid twiddle factors for the nodes in a FFT flowgraph. This technique is more general and leads to a richer solution space than the twisting\cite{bernstein2007tangent}\cite{2008-mateer}, an algebraic way to correctly change the value of twiddle factors. This solution space of FFTs is cast as a SAT problem using quantifier-free formulas over the theory of fixed-size bitvectors, specified in SMT-LIB format\cite{ranise2006smt}, and searched with existing SMT solvers\cite{brummayer2009boolector}\cite{yices}\cite{mathsat4}\cite{een2006translating}. After applying partitioning techniques, we are able to find 6616 FLOP count algorithms for size-$256$ FFTs, and 15128 FLOP count algorithms for size-$512$ FFTs. These numbers are lower than traditional split-radix, 6664 for size-$256$ and 15368 for size-$512$, but not as low as achieved by Van Buskirk's algorithm\cite{lundy2007new} or the tangent FFT\cite{johnson2007modified}\cite{bernstein2007tangent}, 6552 for size-$256$ and 15048 for size-$512$, due to our constraint that complex twiddle factors must have modulus one, an absolute value of one.

Although we supply code for a witness size-$256$ FFT requiring fewer operations than a traditional split-radix\cite{fftexamples}, we are not addressing algorithm design in this paper. An objective to minimize FLOP count is primarily academic given the capabilities of modern computing hardware. We use it only as a well-defined and widely-understood objective to introduce and demonstrate the power of our formulation and search. We believe the ideas presented in this paper can be used to do FFT algorithm design where twiddle factors are not all $n^{th}$ roots of unity, specific hardware is targeted, or other objectives such as overall performance or accuracy are pursued, but these are the topics of a future paper.

This paper continues with an introduction to the DFT, with emphasis on defining concepts central and unique to this paper. In Section \ref{sec:firsthalf}, we present a FFT flowgraph representation for generating a family of FFTs. This formulation of the solution space is tailored so that it can be easily cast as a SMT problem. Section \ref{sec:secondhalf} introduces a first SMT problem formulation and then develops symmetry reduction and partitioning ideas which allows us to solve larger problems. Finally, we conclude with discussion of our experiments and results, application to FFT algorithm design, and future work. 

\subsection{Definitions}
\label{sec:definitions}
The DFT (discrete Fourier transform) is a specific kind of Fourier transform whose input is a sequence of numbers instead of a function.  The sequence of numbers is often obtained by sampling a continuous function.  Throughout this paper, let $n=2^m$, let $i^2=-1$, and let $\omega_n$ represent the complex $n^{th}$ root of unity $e^{-i\frac{2\pi}{n}}$. The $n$-tuple of complex numbers $(a_0,\,a_1,\,a_2,\,\dotsc,a_{n-1})$ is transformed by the DFT into another $n$-tuple of complex numbers $X(k)$ according to the formula 
\begin{equation*}
  X(k) = \sum_{j=0}^{n-1} a_j \omega_n^{jk}.
\end{equation*}
It is well-known that the complex size-$n$ DFT is a linear operator on $\mathbb{C}^n$ and can be represented as multiplication by an $n \times n$ Vandermonde matrix. For our purposes, it is better to identify the entries of the $n$-tuple with the coefficients of the polynomial
\begin{equation*}
  f(x) = a_0 + a_1x + a_2x^2 + \dotsb + a_{n-1}x^{n-1} \in \mathbb{C}[x].
\end{equation*}
Then computing the DFT for a given $n$-tuple is equivalent to evaluating the polynomial $f$ at each of the $n^{th}$ roots of unity $\omega_n^k$, for $k = 0,\,1,\,2,\,\dotsc, n-1$. That is, $X(k) = f(\omega_n^k)$. So each output value of the DFT is a weighted sum of the $a_j$, where the weight of $a_j$ in $X(k)$ is $\omega_n^{jk}$. 

When an FFT is used to compute a size-$n$ DFT, with twiddle factors of modulus one, the product of all the twiddle factors applied to $a_j$ in the computation of $X(k)$ equals the weight $\omega_n^{jk}$. We'd like to keep track of the accumulated weight on any given $a_j$ through all of the intermediate FFT results. To do this, we employ the polynomial view introduced by Fiduccia in \cite{fiduccia1972polynomial} and elaborated by Bernstein in \cite{bernstein2007tangent} and Burrus in \cite{cnx:m16327}. Associating the input to a polynomial of degree $n-1$ with coefficients $a_j$ means that an intermediate FFT result is associated to a polynomial of lower degree whose coefficients are weighted sums of the $a_j$. For example, when $n=8$ and
\begin{equation*}
f(x) = a_0 + a_1x + a_2x^2 + a_3x^3 + a_4x^4 + a_5x^5 + a_6x^6 + a_7x^7, 
\end{equation*}
two of the intermediate results of the radix-2 FFT are
\begin{equation*}
 f\text{ mod} (x^2 - \omega_8^0) = (a_0 + a_2 + a_4 + a_6) + (a_1 + a_3 + a_5 + a_7)x
\end{equation*}
and
\begin{align*}
  f\text{ mod} (x^2 - \omega_8^4) &= (a_0 - a_2 + a_4 - a_6) + (a_1 - a_3 + a_5 - a_7)x\\
                                  &= (a_0 + a_2\omega_8^4 + a_4 + a_6\omega_8^4) + (a_1 + a_3\omega_8^4 + a_5 + a_7\omega_8^4)x.
\end{align*}
Each of the coefficients in the two linear polynomials above is represented by a node in a flowgraph, which we describe in the next section. First, we define some characteristics of these coefficients.

\begin{definition}
The \textbf{\emph{base}} of a coefficient is the $a_j$ of lowest index that appears in the weighted sum comprising that coefficient.
\label{defn:base}
\end{definition}

\begin{definition}
The \textbf{\emph{stride}} of a coefficient is the integer difference between the indices of any two successive $a_j$ in the weighted sum, when the terms of the sum are written with the indices in strictly increasing order. When the coefficient consists of a single $a_j$, the stride is defined as $n$, the size of the DFT.
\label{defn:stride}
\end{definition}

In the polynomials above, the constant terms have base $a_0$, the linear terms have base $a_1$, and all four coefficients have stride 2.  If the example above is continued to determine $f\text{ mod}(x^8-1)$, each coefficient $a_j$ will have base $a_j$ and stride 8. For any $k$, the output value $X(k)$ has base $a_0$ and stride 1.

\begin{definition}
The \textbf{\emph{weight stride}}, $\boldsymbol{W_s}$, of a coefficient is the integer difference (mod $n$) of the powers put on $\omega_n$ to form the weights on any two successive $a_j$ in the coefficient, when the terms of the coefficient are written with the indices in strictly increasing order. 
\label{defn:weightstride}
\end{definition}

The $W_s$ of each coefficient in $f\text{ mod} (x^2 - \omega_8^0)$ above is 0. The $W_s$ of each coefficient in $f\text{mod } (x^2 - \omega_8^4)$ above is 4. For $f\text{ mod}(x^8-1)$ in the example above, there is no \emph{combination} of the $a_j$ comprising any coefficient, so the $W_s$ of each of the eight original coefficients is defined to be zero.  The $W_s$ for $X(k) = a_0\omega_n^0 + a_1 \omega_n^k + a_2 \omega_n^{2k} + \dotsb + a_{n-1}^{(n-1)k}$ is $k$. 

\begin{definition}
The \textbf{\emph{weight on base}}, $\boldsymbol{W_b}$, of a coefficient in an intermediate FFT result is the integer power (mod $n$) to which $\omega_n$ has been raised to form the accumulated weight on the base of the coefficient.
\label{defn:weightonbase}
\end{definition}

The $W_b$ of each of the four coefficients in the example, indeed of any coefficient from the radix-2 FFT, is zero.  To find an example of coefficients with nonzero $W_b$ among the common FFTs, we'll consider the size-8 twisted FFT. Given $f(x)$ as in the example, the remainder $f$ mod$(x^4+1)$ determines the remainder $f(\omega_8 x)$ mod$(x^4-1)$ as described in \cite{bernstein2007tangent} and \cite{2008-mateer}. It follows that one of the intermediate results of the size-8 twisted FFT is
\begin{align*}
  f(\omega_8 x)\text{ mod}(x^2 - 1) &= (a_0 - a_4 + \omega_8^2[a_2 - a_6]) + (\omega_8^1[a_1 - a_5] + \omega_8^3[a_3 - a_7])x\\
                                  &= (a_0 + a_2\omega_8^2 + a_4\omega_8^4 + a_6\omega_8^6) + (a_1\omega_8^1 + a_3\omega_8^3 + a_5\omega_8^5 + a_7\omega_8^7)x.
\end{align*}
So we see that the coefficient of the linear term has base $a_1$ and $W_b = 1$. This $W_b$ as well as the other definitions from this section are visually summarized in Figure \ref{fig:definitions}.

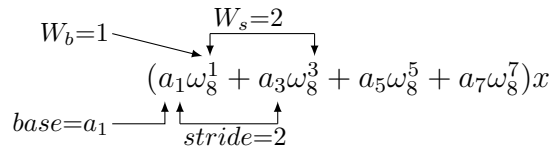
\begin{figure}[h]
\centering
\input{definitions}
\caption{Definitions}
\label{fig:definitions}
\end{figure}

\section{A Flowgraph Representation for \\ Generating a Family of FFT Algorithms}
\label{sec:firsthalf}

Signal flowgraphs are a widely used formalism to represent and analyze FFTs\cite{cnx:m16352}\cite{1988-brigham}. In this section we show how the concepts defined in Section \ref{sec:definitions} occur in flowgraphs of common power-of-two FFTs. In particular, we will show that each node in a given flowgraph can be labeled with a 3-tuple, $(stride,base,W_s)$, which is an invariant for a family of FFT algorithms that can be realized by that given flowgraph. This invariant can then be used to generate FFT instances realizable by that flowgraph. 

To facilitate the discussion, we show two example flowgraphs. The first, shown in Figure \ref{fig:radix2fft}, is Gauss' original FFT\cite{1866-gauss-3}\cite{bernstein2007tangent}. The second, shown in Figure \ref{fig:conjugatesplitradixfft}, represents a size-16 conjugate split-radix as discussed in \cite{johnson2007modified}\cite{kamarconjugate}.

\subsection{Edges and Nodes}
Each directed edge represents the transfer of a complex number, either into or out of a node. In an algorithmic implementation of the flowgraph, each edge is indeed a single concrete complex number, but for the purposes of our flowgraph analysis, this complex number should be thought of symbolically as a weighted sum of the $a_j$, where the weight on any $a_j$ is some $\omega^*_n$. 

The input operands of the FFT, labeled $a_0...a_{n-1}$, are shown at the top and the output values, labeled $X(0)...X(n-1)$, are shown at the bottom. Unlike traditional FFT flowgraphs, we use $a_j$ instead of $x(j)$ for input operands and show data flow top-to-bottom instead of left-to-right to facilitate discussion relating this flowgraph to the polynomial evaluation perspective of the FFT.

Each node represents complex addition and/or multiplication operations applied to the input operands to generate the output values. Figure \ref{fig:node} shows the internal behavior of a node. For nodes with two input edges, the two input operands are added to produce the single complex result \emph{id}, when viewed concretely. We prefer to view \emph{id} symbolically, as a weighted sum of the $a_j$, where the weight on any $a_j$ is some $\omega^*_n$.  Next, \emph{id} is separately multiplied by two \textbf{twiddle factors} to produce the left and right output values. In Figure \ref{fig:node} the left and right twiddle factors are shown as $\omega^{ltfp}_n$ and $\omega^{rtfp}_n$ respectively but in the concise node representation as seen in Figure \ref{fig:radix2fft} only the integers $ltfp$ (left twiddle factor power) and $rtfp$ (right twiddle factor power) are shown in the bottom row of each node. 

\begin{figure}[t]
\centering
\input{nodeinternals}
\caption{Node Internal Behavior}
\label{fig:node}
\end{figure}
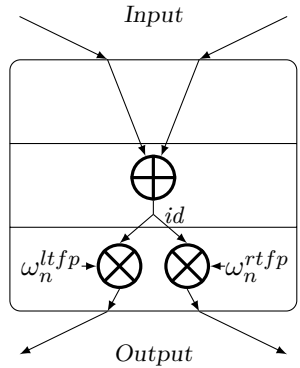

We count the cost of multiplication by some twiddle factor $\omega^{tfp}_n$ in the traditional way, where the cost of multiplication by $1$,$-1$,$i$ or $-i$ is free, multiplication by $\sqrt{i}$, $-\sqrt{i}$, $\sqrt{-i}$ or $-\sqrt{-i}$ is 4 floating point operations (FLOPs), and multiplication by any other $n^{th}$ root of unity is 6 FLOPs. In addition to the potential multiplication cost, each node always requires 2 FLOPs for the cost of the addition.  

There is one interesting multiplication cost exception when $\omega_n^{ltfp}$ and $\omega_n^{rtfp}$ are complex conjugates. In this case, $\Re{(\omega_n^{ltfp})} = \Im{(\omega_n^{rtfp})}$ and $\Im{(\omega_n^{ltfp})} = \Re{(\omega_n^{rtfp})}$ so that only 4 real multiplications and 4 real additions are needed to weight by both $\omega_n^{ltfp}$ and $\omega_n^{rtfp}$. In this case, we tally 6 FLOPs for the weight by $\omega_n^{ltfp}$ and only an additional 2 FLOPs for the weight by $\omega_n^{rtfp}$. For cases where $ltfp=rtfp$, we tally FLOPs for $\omega_n^{ltfp}$ only.   

Two separate multiplications by $\omega^{ltfp}_n$ and $\omega^{rtfp}_n$ are never seen in traditional FFTs as it typically leads to higher cost when counting total floating point operations. Instead, one multiplication is done and the result may or may not be negated at no additional cost to generate the second output. In this paper, we adopt the more general description containing two separate multiplications and will later show how constraints can be applied to prune the search space to solutions requiring only a single multiplication without detriment to the final global FLOP count.

\begin{figure}[t]
\centering
\tikzstyle{fftnode}=[rectangle split, rectangle split parts=3, rounded corners, text badly centered, text width=8mm, minimum width=8.1mm, draw=black,inner sep=0.5mm]
\tikzstyle{fftterminal}=[rectangle, text centered]
\tikzstyle{arc}=[->]
\tikzstyle{fftgraph}=[>=latex]
\tikzstyle{fftmatrix}=[row sep=10mm, column sep={13mm,between origins}]
\tikzstyle{background}=[rectangle,rounded corners,fill=gray!25,inner sep=1.0mm]
\tikzstyle{polytext}=[anchor=120,inner sep=0mm]
\tikzstyle{costtext}=[anchor=220,inner sep=0mm]
\tikzstyle{stride}=[inner sep=1.0mm]
\input{classic_size8}
\caption{Size-8 Radix-2 DIF FFT Flowgraph}
\label{fig:radix2fft}
\end{figure}
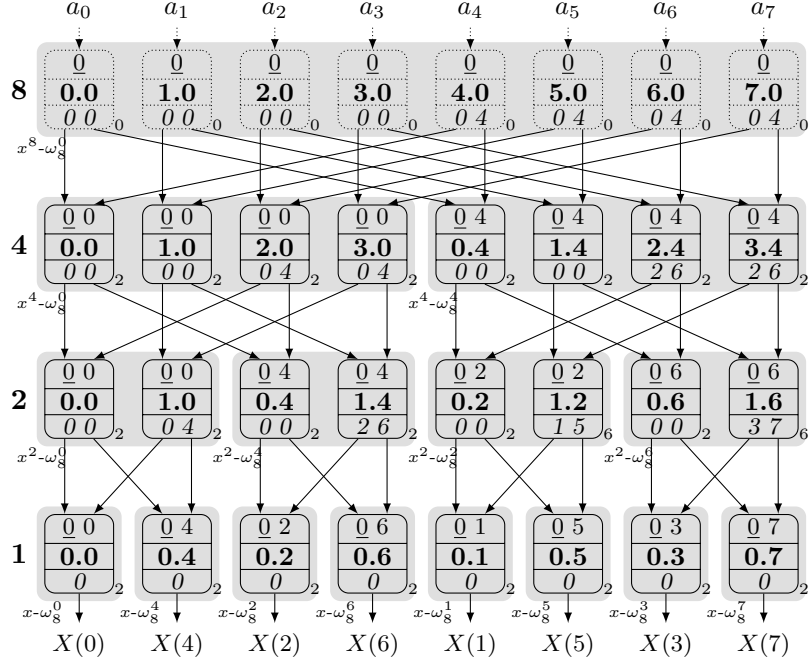

Dotted nodes in the top row of Figure \ref{fig:radix2fft} only have a single input operand and consequently there is no internal addition. In this case, the input operand is used directly as operand \emph{id} within the node. Nodes in the bottom row of Figure \ref{fig:radix2fft} only have a single output value and consequently there is just a single internal multiplication. In this case, only a single twiddle factor is specified. Again, traditional FFTs often suppress this final multiplication as it is typically a cost-free multiplication by 1 or -1. For the generality of our first formulation, we always include this final multiplication, but will later prove via SMT that it is not required when searching for lowest FLOP count FFTs.  

For the class of FFT flowgraphs we are considering, each node has at most two parents and two children. We adopt {\bf{dot notation}} when it is necessary to refer to attributes of a node's parents or children. We refer to a node's left or right parent as \emph{nd.lp} or \emph{nd.rp}, respectively, for a node \emph{nd}. Likewise, \emph{nd.lc} or \emph{nd.rc} refer to a node's left or right child, respectively. With this notation, a node's left parent's left parent twiddle factor can be referred to concisely as $nd.lp.lp.\omega^{tfp}_n$. Note that \emph{tfp} can be used here rather than \emph{ltfp} or \emph{rtfp} as it is clear from the graph context when tracing edges which twiddle factor applies.

\subsection{Flowgraph Properties}

Our flowgraph analysis requires that the following two properties be true, which are checked by computer traversal of the flowgraph.   

\begin{prop}
There is at most one path from any input operand $a_j$ or internal node $nd_p$ to any output value $X(k)$ or node $nd_q$. 
\label{prop:singlepath}
\end{prop}

\begin{definition}
 The subset of input operands $a_j$ that can reach a node contains that node's \textbf{original ancestors} and is denoted as $nd.A$ for a node $nd$. The subset of output values $X(k)$ that is reachable from a node contains that node's \textbf{terminal descendants} and is denoted as $nd.D$ for a node $nd$.
\label{defn:reachable} 
\end{definition}

\begin{prop}
For any node $nd$ in the flowgraph, when the elements of $nd.A$ are ordered such that indices are strictly increasing, the difference (mod $n$) of indices on successive original ancestors in the list is constant. Furthermore, original ancestors of a node's left parent interleave precisely with the original ancestors of the right parent, and the sets are always disjoint.
\label{prop:interleave}
\end{prop}

\begin{example}
For the node at the end of the third row in Figure \ref{fig:radix2fft} labeled with a bold 1.6,
\begin{equation*}
 nd.A = \{a_1,a_3,a_5,a_7\}
\end{equation*}
when ordered with strictly increasing indices. The integer difference (mod $n$) of successive original ancestor indices is always 2 for this example. For the node's left and right parents, 
\begin{align*}
 nd.A &= nd.lp.A \cup nd.rp.A \\
      &= \{a_1,a_5\} \cup \{a_3,a_7\} \\
      &= \{a_1,a_3,a_5,a_7\},
\end{align*}
which interleave precisely when combined.
\label{exmp:interleave}
\end{example}

Flowgraphs adhering to these two properties are expected given the divide-and-conquer nature of common power-of-two FFTs. We have built flowgraphs of various size-$n$ for radix-2, radix-4, decimation-in-time and decimation-in-frequency split-radix, conjugate split-radix \cite{johnson2007modified} as well as twisted\cite{bernstein2007tangent} FFTs, and have always found these properties to be true. For FFTs with some radix-4 content, this requires that when adding four numbers, the addition is factored into a binary addition tree that observes Property \ref{prop:interleave}, which is what is commonly done. For twisted FFTs, different twisting functions $\zeta$ lead to different permutations of $X(k)$, but these are isomorphisms of the same flowgraph structure. It is not the point of this paper to prove which FFT algorithms generate which flowgraphs. Instead, we observe that many common power-of-two FFT algorithms generate flowgraphs that have these properties, and we require adherence to develop our flowgraph-based ideas. 

\subsection{A Node's \emph{base} and \emph{stride}}

In the flowgraph shown in Figure \ref{fig:radix2fft}, the left number in the middle row of each node is the $base$ index for that node's $id$ and the number at the left of an entire row of nodes is the $stride$ for any node's $id$, when $id$ is viewed symbolically as a weighted sum. Figure \ref{fig:flowgraphkey} is a key for all flowgraph labels and Figure \ref{fig:node} identifies the internal edge $id$.

\begin{figure}[hb]
\centering
\input{nodekey}
\caption{Flow Graph Node Key}
\label{fig:flowgraphkey}
\end{figure}

\begin{definition}
A \textbf{node's \emph{base} label}, $nd.base$, is always the index of the $base$, as defined in Definition \ref{defn:base}, for the weighted sum $nd.id$ represented by the node. The number $nd.base$ is the minimum of $nd.lp.base$ or $nd.rp.base$. For the terminal case when $nd$ has a single input operand, $nd.base$ is equal to $j$ for the given input $a_j$.
\label{defn:graphbase} 
\end{definition}

To facilitate computation of $nd.base$, later computation of \emph{weight on base}, as well as impose regularity on the flowgraph, the following property is enforced in flowgraph diagrams and computer data structures.

\begin{prop}
For any node $nd$ in a flowgraph, the relation $(nd.lp.base<nd.rp.base)$ is always true. 
\label{prop:orderedbases}
\end{prop}

\begin{definition}
A \textbf{node's \emph{stride} label}, $nd.stride$ for a node $nd$, is always the $stride$, as defined in Definition \ref{defn:stride}, for the weighted sum $nd.id$ represented by the node. The number $nd.stride$ is the absolute difference (mod $n$) of $nd.lp.base$ and $nd.rp.base$. For the terminal case when $nd$ has a single input operand, $nd.stride$ is defined to be $n$ for a size-$n$ FFT. Property \ref{prop:interleave} ensures that strides are constant and hence a single stride label per node is sufficient.
\label{defn:graphstride}
\end{definition}

\begin{example}
From Example \ref{exmp:interleave}, we know that for the last node in the third row of Figure \ref{fig:radix2fft},
\begin{equation*}
 nd.A = \{a_1,a_3,a_5,a_7\}.
\end{equation*}
Ignoring values of applied weights in the flowgraph, the polynomial coefficient represented by this node must be of the form
\begin{equation*}
 nd.id = a_1\omega^*_8+a_3\omega^*_8+a_5\omega^*_8+a_7\omega^*_8.
\end{equation*}
From the discussion in Section \ref{sec:definitions} we can deduce that $nd.base=1$ and $nd.stride=2$. Also, we see that $nd.lp.base=1$ and $nd.rp.base=3$. By Definition \ref{defn:graphbase}, 
\begin{align*}
 nd.base &= \min\{nd.lp.base,nd.rp.base\}\\
	 &= \min\{1,3\}\\
         &= 1,
\end{align*}
and by Definition \ref{defn:graphstride},
\begin{align*}
 nd.stride &= nd.rp.base - nd.lp.base \text{ (mod }n\text{)}\\
	 &= 3 - 1 \text{ (mod }8\text{)}\\
         &= 2.
\end{align*}
This node's row in the flowgraph is labeled with 2, the stride. The first label in the middle row of the node itself is 1, the base.  
\end{example}

\subsection{Weight on \emph{base}}

The \emph{weight on base} for every node's input edge, as well as that node's weighted sum $id$, is recorded in the flowgraph. Even though $W_b$ is defined for a true polynomial coefficient in Definition \ref{defn:weightonbase}, we record the \emph{weight on base} for both the left and right input edge \emph{before} the addition since both are required later to determine a node's weight stride. As shown in Figure \ref{fig:flowgraphkey}, the top row of each node specifies $W_b$, the integer power (mod $n$) to which $\omega_n$ has been raised to form the accumulated weight on the $base$ of the weighted sum of $a_j$ represented by the left input edge. Likewise, $rW_b$ represents the same for the right input edge. From Property \ref{prop:orderedbases}, we know that after the addition the $base$ of $nd.id$ is the same as the $base$ of the left parent and that the addition does not alter weights. Thus, $W_b$ for $nd.id$ is equal to the \emph{weight on base} of the left input edge and there is no need for a separate $lW_b$.

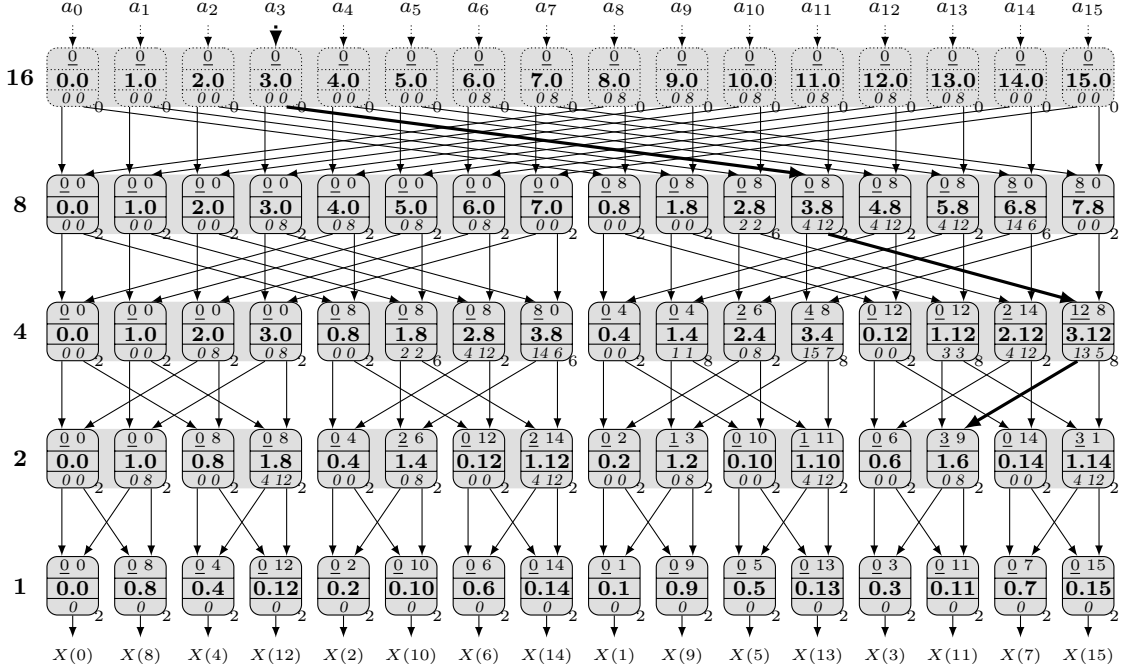
\begin{figure}[t]
\centering
\tikzstyle{fftnode}=[rectangle split, rectangle split parts=3, rounded corners, text badly centered, text width=6mm, minimum width=6mm, draw=black,inner sep=0.35mm]
\tikzstyle{fftterminal}=[rectangle, text centered]
\tikzstyle{arc}=[->]
\tikzstyle{tarc}=[->,very thick]
\tikzstyle{fftgraph}=[>=latex]
\tikzstyle{fftmatrix}=[row sep=9mm, column sep={9mm,between origins}]
\tikzstyle{background}=[rectangle,rounded corners,fill=gray!25,inner sep=0.0mm]
\tikzstyle{polytext}=[anchor=65,inner sep=0mm]
\tikzstyle{costtext}=[inner sep=1.0mm]
\tikzstyle{stride}=[inner sep=1.0mm]
\input{conjugatesplitradix_size16}
\caption{Size-16 Conjugate Split-Radix FFT}
\label{fig:conjugatesplitradixfft}
\end{figure}

\begin{definition}
The \textbf{\emph{weight on base}} of a node's left input edge is,
\begin{equation*}
nd.\omega^{W_b}_n=(nd.lparent.\omega^{tfp}_n)(nd.lparent.\omega^{W_b}_n),
\end{equation*}
and likewise for a node $nd$'s right input edge is,
\begin{equation*}
nd.\omega^{rW_b}_n=(nd.rparent.\omega^{tfp}_n)(nd.rparent.\omega^{W_b}_n).
\end{equation*}
Following from Definition \ref{defn:graphbase} and Property \ref{prop:orderedbases}, the \emph{weight on base} of $nd.id$ is equal to the \emph{weight on base} of the left input edge and both are referenced as $W_b$. For the terminal case when $nd$ has a single input, $W_b$ is defined to be zero. Finally, note that $W_b$ for all output values $X(k)$ is always zero as all $X(k)$ contain a constant term with $a_0$ that can only be weighted by $\omega^0_n$ in any correct DFT.
\label{defn:graphweightonbase} 
\end{definition}

\begin{example}
 Since \emph{weight on base} is the result of a series of multiplications by various roots of unity $\omega^*_n$, it can also be viewed as addition (mod $n$) of the powers on the roots of unity. This is illustrated by the path shown with bold edges from input operand $a_3$ to a node $nd$ with $nd.base=1$ and $nd.stride=2$ in Figure \ref{fig:conjugatesplitradixfft}. Then
\begin{equation*}
 nd.\omega^{rW_b}_{16}=\omega^{13}_{16}(\omega^{12}_{16}(\omega^0_{16}(a_3))),
\end{equation*}
which is $a_3$ multiplied by all twiddle factors along the path, and can be rewritten as
\begin{align*}
 rW_b &= 13 + 12 + 0 \text{ (mod }16\text{)}\\
	   &= 9.
\end{align*}
This $rW_b$, 9, is shown in the upper right corner of the node. Once this path reaches $nd$, we no longer keep track of the weight on $a_3$ as it is no longer the $base$ of the weighted sum $id$. However, it is still essential to keep track of this weight up to this point as it is used to compute $W_s$. 
\end{example}

\subsection{Weight Stride}

A node's \emph{weight stride} label, $W_s$, is shown at the right of each node's middle row, as seen in Figures \ref{fig:radix2fft}, \ref{fig:flowgraphkey} and \ref{fig:conjugatesplitradixfft}.

\begin{definition}
A node's \emph{weight stride} label is
\begin{equation*}
  nd.W_s = nd.rW_b - nd.W_b \text{ (mod }n\text{)}.
\end{equation*}
For the terminal case when $nd$ has a single input, $W_s$ is defined to be zero. The number $nd.W_s$ is always the $W_s$ as defined in Definition \ref{defn:weightstride} for the weighted sum $nd.id$ represented by the node.
\label{defin:graphweightstride}
\end{definition}

\begin{example}
Again consider the node $nd$ at the end of the bold path in Figure \ref{fig:conjugatesplitradixfft} where

\begin{align*}
 nd.W_s &= nd.rW_b - nd.W_b \text{ (mod }n\text{)}\\
 &= 9 - 3 \text{ (mod }16\text{)}\\
		      &= 6.
\end{align*}
This $W_s$, 6, appears as the last label in the middle row of this node. We can now reconstruct exactly the weighted sum of coefficient $nd.id$. For the node we are considering with $stride=2$, $base=1$, $W_s=6$ and $W_b=3$,
\begin{equation*}
 nd.id = a_1\omega^3_{16} + a_3\omega^9_{16} + a_5\omega^{15}_{16} + a_7\omega^{5}_{16} + a_9\omega^{11}_{16} + a_{11}\omega^{1}_{16} + a_{13}\omega^{7}_{16} + a_{15}\omega^{13}_{16}
\end{equation*}
\end{example}

Now that a node's $W_s$ is defined, we present a key observation that $W_s$ is invariant across all FFTs that can be mapped to the given flowgraph. This invariance is central in defining a family of FFTs that can then be searched for desirable members.

\begin{theorem}
 For a size-$n$ FFT flowgraph constructed by \emph{any} FFT algorithm such that Properties \ref{prop:singlepath} and \ref{prop:interleave} hold, every node's $W_s$ is an invariant. 
\label{thm:weightstride}
\end{theorem}

\begin{proof}
 Consider an arbitrary node $nd$ in the flowgraph. Next, consider an arbitrary FFT output value from this node's terminal descendants, $X(k) \in nd.D$. For the two input values $a_{(nd.base)}$ and $a_{(nd.base)+(nd.stride)}$ in the weighted sum $X(k)$, it follows from the definition of the DFT that these input values must be weighted as 
\begin{equation*}
a_{(nd.base)}\omega^{k(nd.base)\text{ (mod }n\text{)}}_n\text{ and }a_{(nd.base)+(nd.stride)}\omega^{k((nd.base)+(nd.stride))\text{ (mod }n\text{)}}_n.
\end{equation*}
Hence, the \emph{weight stride} between these input values is
\begin{align*}
 \text{\emph{weight stride}} &= k((nd.base)+(nd.stride)) - k(nd.base)\text{ (mod }n\text{)}\\
                             &= k(nd.stride) \text{ (mod }n\text{)}.
\end{align*}
Since, by Property \ref{prop:singlepath}, $nd$ is the only contributor of $a_{(nd.base)}$ and $a_{((nd.base) + (nd.stride))}$ to $X(k)$, and they are are bound together by the addition in $nd$ and never will be weighted again individually, we have
\begin{equation*}
 nd.W_s = k(nd.stride)\text{ (mod }n\text{)}.
\end{equation*}
Any other value for $nd.W_s$ would produce an incorrect FFT result.
\end{proof}

\begin{example}
Again consider the node at the end of the bold path in Figure \ref{fig:conjugatesplitradixfft} with $nd.W_s=6$, $nd.stride=2$ and $X(3),X(11) \in nd.D$. Then,
\begin{align*}
nd.W_s &= k(nd.stride)\text{ (mod }n\text{)} \\
 6 &= 3\times2 \text{ (mod }16\text{)}\\
   &= 11\times2 \text{ (mod }16\text{)}.
\end{align*}
\end{example}

\subsection{Canonical Node Labels}
Since the three node labels $base$, $stride$ and $W_s$ are either defined or proven to be unchanged by any applied weight $\omega_n$, we can now assign a canonical label to each node. For a size-$n$ FFT flowgraph, the set of canonical node labels defines a family of FFTs that can be realized by that flowgraph. Actual applied weights $\omega_n$, interpreted as $W_b$, distinguish members in the family of FFTs.

\begin{definition}
A node's \textbf{canonical label} is $nd(nd.stride,nd.base,nd.W_s)$.
\label{defn:canonicallabel}
\end{definition}

\begin{example}
Again consider the node at the end of the bold path in Figure \ref{fig:conjugatesplitradixfft}. This node is labeled $nd(2,1,6)$ and is the only node with that label in the flowgraph. The $nd.stride=2$ appears to the left of the row in which $nd(2,1,6)$ is found. The $nd.base=1$ and $nd.W_s=6$ appear in bold on the node itself.
\end{example}

\subsection{Correspondence to the Polynomial View}
\label{sec:polyview}

Although our main representation is a flowgraph, we have relied on the polynomial view of the FFT to facilitate our discussion. In particular, each edge of the flowgraph represents a weighted sum of $a_j$ and each $nd.id$ a coefficient of the original polynomial modulo one of its factors, also a weighted sum of $a_j$.  We highlight with gray background in Figures \ref{fig:radix2fft} and \ref{fig:conjugatesplitradixfft} the polynomial factor lattice as described by Bernstein in \cite{bernstein2007tangent} and shown in Figure \ref{fig:factortree}. The degree of each polynomial factor is the stride for all flowgraph nodes it contains. The power to which $\omega_n$ is raised to form the constant term in that polynomial is the $W_{stride}$ for all flowgraph nodes it contains. And finally, each node's $nd.base$ is the index of lowest degree among the $a_j$ used in the weighted sum represented by that node.

Recall our example for the case $n=8$. We associate the sampled data to coefficients of a degree 7 polynomial:
\begin{equation*}
 f(x) = a_0 + a_1x + a_2x^2 + a_3x^3 + a_4x^4 + a_5x^5 + a_6x^6 + a_7x^7.
\end{equation*}

This polynomial is an element of $\mathbb{C}[x]$, where we have a division algorithm that gives $f(x) \equiv r(x)\,\text{mod}\,p(x)$ when $ f(x)=q(x) p(x) + r(x)$, for $r(x)$ of lesser degree than $p(x)$. In particular, we have $f(x) \equiv f(x)\,\text{mod}(x^8 -1)$ because $f$ has degree strictly less than 8. The residue class of $f$ modulo $(x^8 -1)$ determines the residue class of $f$ modulo $(x^4 -1)$ and modulo $(x^4 +1)$, as the latter two polynomials are factors of $x^8-1$.  Given a complete factorization of $x^8-1$ into distinct, irreducible, linear factors as shown in Figure \ref{fig:factortree},

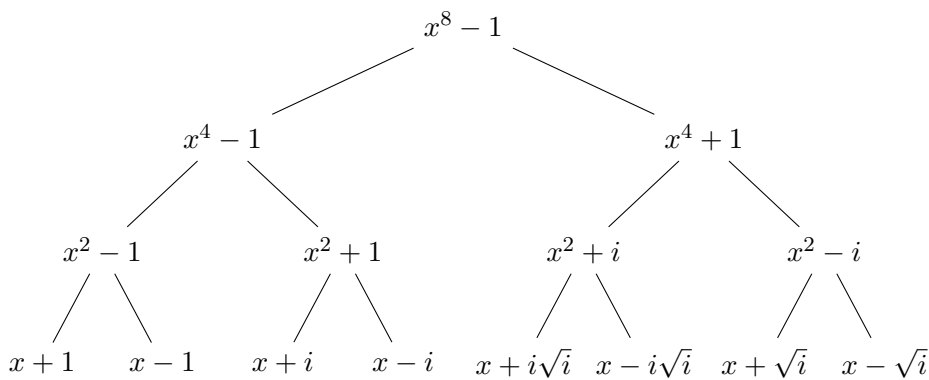
\begin{figure}[t]
    \centering
    \begin{tikzpicture}[
        level 1/.style={sibling distance=64mm},
        level 2/.style={sibling distance=32mm},
        level 3/.style={sibling distance=16mm},
        level 4/.style={sibling distance=8mm}  
    ]
    \node{$x^8-1$}
        child{node {$x^4-1$}
            child{node {$x^2-1$}
                child{node {$x+1$}
                  %child{node {$x-\omega_8^0$}}
               }
                child{node {$x-1$}
                  %child{node {$x-\omega_8^4$}}
               }
            }
            child{node {$x^2+1$}
                child{node {$x+i$}
                  %child{node {$x-\omega_8^2$}}
               }
                child{node {$x-i$}
                  %child{node {$x-\omega_8^6$}}
               }
            }           
        }
        child{node {$x^4+1$}
            child{node {$x^2+i$}
                child{node {$x+i\sqrt{i}$}
                  %child{node {$x-\omega_8^1$}}
               }
                child{node {$x-i\sqrt{i}$}
                  %child{node {$x-\omega_8^5$}}
               }
            }
            child{node {$x^2-i$}
                child{node {$x+\sqrt{i}$}
                  %child{node {$x-\omega_8^3$}}
               }
                child{node {$x-\sqrt{i}$}
                  %child{node {$x-\omega_8^7$}} 
               } 
            }
        };
    \end{tikzpicture}
    \caption{Factor lattice of $x^8 -1$}
    \label{fig:factortree}
\end{figure}

we have the following isomorphism of rings:
\begin{equation*}
 \mathbb{C}[x]/ (x^8-1) \cong \mathbb{C}[x]/ (x-1)   \times \mathbb{C}[x]/ (x +1)   \times \mathbb{C}[x]/ (x +i)  \times  \dotsb \times \mathbb{C}[x]/ (x +\sqrt{i})  \times  \mathbb{C}[x]/ (x -\sqrt{i}),
\end{equation*}
which follows from the Chinese Remainder Theorem, as seen in \cite{AbstractAlgebra}.

So an FFT algorithm is finding an element from the product ring that corresponds to the given $f(x) \in \mathbb{C}/(x^8-1)$. Our flowgraph highlights the path of the inputs through the lattice of factor rings and canonical homomorphisms in Figure \ref{fig:factorrings}. The intermediate polynomials are 

\begin{align*}
 f(x)\,\text{mod}(x^4-1) &= (a_0+a_4) + (a_1+a_5)x +(a_2+a_6)x^2 +(a_3+a_7)x^3\\
 f(x)\,\text{mod}(x^4+1) &= (a_0-a_4) + (a_1-a_5)x +(a_2-a_6)x^2 +(a_3-a_7)x^3\\
 f(x)\,\text{mod}(x^2-1) &= (a_0+a_4+a_2+a_6) + (a_1+a_5+a_3+a_7)x\\
 f(x)\,\text{mod}(x^2+1) &= (a_0+a_4-a_2-a_6) + (a_1+a_5-a_3-a_7)x\\
 f(x)\,\text{mod}(x^2-i) &= (a_0-a_4)+(a_2-a_6)i + [(a_1-a_5)+(a_3-a_7)i]x\\
 f(x)\,\text{mod}(x^2+i) &= (a_0-a_4)-(a_2-a_6)i + [(a_1-a_5)-(a_3-a_7)i]x.
\end{align*}

Since finding the residue of $f(x)\,\text{mod}\,(x^4+1)$ is equivalent to setting $x^4$ equal to $-1 = \omega_8^4$ in $f(x)$, we see in the coefficients of $f(x)\,\text{mod}\,(x^4+1)$ pairs of the original inputs whose indices differ by 4.  Viewing these coefficients as weighted sums of the $a_j$, where the $a_j$ are written with the indices increasing, we note that successive weights change by a factor of $\omega_8^4 = -1$.   Since finding the residue of $f(x)\,\text{mod}\,(x^2+i)$ is equivalent to setting $x^2$ equal to $-i = \omega_8^2$ in $f(x)$, we see in the coefficients of $f(x)\,\text{mod}\,(x^2+i)$ four of the original inputs whose indices differ by 2 when listed in increasing order. Viewing these coefficients as weighted sums of the $a_j$, where the $a_j$ are written with the indices increasing, we note that successive weights change by a factor of $\omega_8^2 = -i$.

\begin{figure}[b]
   \centering
    \begin{tikzpicture}[
        level 1/.style={sibling distance=10mm},
        level 2/.style={sibling distance=10mm},
        level 3/.style={sibling distance=10mm},
        level 4/.style={sibling distance=10mm}
        ]
    \node {$f(x) \in \mathbb{C}/(x^8-1)$}
        child{node {$\mathbb{C}/ (x^4-1) \times \mathbb{C}/ (x^4+1)$} edge from parent [->]
           child{node {$\mathbb{C}/ (x^2-1) \times \mathbb{C}/ (x^2+1) \times \mathbb{C}/ (x^2+i) \times \mathbb{C}/ (x^2-i)$} edge from parent [->]
                 child{node {$\mathbb{C}/ (x-1) \times \mathbb{C}/ (x+1)  \times \dotsb \dotsb \times \mathbb{C}/ (x+\sqrt{i}) \times \mathbb{C}/ (x-\sqrt{i})$} edge from parent [->]}}
         };
    \end{tikzpicture}
    \caption{Factor rings with canonical homomorphisms}
   \label{fig:factorrings}
\end{figure}

\begin{example}
In Figure \ref{fig:radix2fft}, the $stride=4$ row has two polynomials highlighted, the left labeled $x^4_8-\omega^0_8$ and the right labeled $x^4_8-\omega^4_8$. The right polynomial, $x^4_8-\omega^4_8$, has four nodes corresponding to the four terms of this new polynomial. The constant term has $base=0$, the linear term has $base=1$, and so on, until the last node with $base=3$ represents the coefficient of the $x^3$ term in the polynomial. Since these four nodes arise from $x^4_8-\omega^4_8$, all nodes in this polynomial have $stride=4$. Finally, since the constant term of the factor polynomial is $-\omega^4_8$, written as $\omega_n$ raised to a power instead of the usual $+1$, all nodes in this polynomial have a $W_{stride}=4$.
\end{example}

This correspondence exists for the original FFT attributed to Gauss\cite{1866-gauss-3}. Twisting as described in \cite{bernstein2007tangent} implies that we use a different factor lattice for $x^8-1$. But it is essential to remember that our flowgraph analysis to derive canonical labels is independent of any twists and will derive the same canonical labels for a size-$n$ flowgraph regardless of what twists are applied in the particular FFT used to generate the flowgraph.

In the case of twisting $x^4+1$ to $x^4-1$ in a size-8 FFT, we see that 
\begin{align*}
 f(x) &\equiv r(x)\,\text{mod}(x^4+1)\\
 \implies f(x) &=q(x)(x^4+1) + r(x)\\
 \implies f(\omega_8x) &= q(\omega_8x)(\,(\omega_8x)^4+1)+r(\omega_8x)\\
   &=q(\omega_8x)(\,\omega_8^4x^4+1)+r(\omega_8x)\\
   &= (-1)q(\omega_8x)(x^4-1)+r(\omega_8x)\\
 \implies f(\dot{x}) &\equiv r(\dot{x})\,\text{mod}(\dot{x}^4-1),
\end{align*}
where $\dot{x}=\omega_8x$.  Taking the polynomial view, the element $\dot{x}^4-1 \in \mathbb{C}[\dot{x}]$ has a factor tree isomorphic to that of $x^4-1 \in \mathbb{C}[x]$. Whereas the factor ring $\mathbb{C}[x]/(x^4+1)$ may be considered a 4-dimensional vector space over $\mathbb{C}$ with basis $\{1,\,x,\,x^2,\,x^3\}$, the new factor ring $\mathbb{C}[\dot{x}]/(\dot{x}^4-1)$ has basis $\{1,\,\omega_8x,\,(\omega_8x)^2,\,(\omega_8x)^3\}$.  So the \emph{stride} and $W_{stride}$ exhibited by each set of highlighted nodes in the flowgraph is preserved.

\subsection{Generating a Family Member FFT Algorithm}

Because $W_s$ is independent of any particular FFT's twiddle factors, we can use it as the basis for generating all members of a family of FFT algorithms represented by a given flowgraph. A valid FFT can be created by \emph{randomly} picking integer $W_b$ values for all nodes in the flowgraph. Given these choices for $W_b$, $W_s$ determines values for $rW_b$ for all nodes in the flowgraph. Next, $W_b$ and $rW_b$ determine values for all twiddle factors, and a unique assignment of twiddle factors distinguishes a member in the family of FFT algorithms. Before we present a more formal algorithm for this process, we must first define how twiddle factors can be determined from $W_b$.

\begin{definition}
Following from Definition \ref{defn:graphweightonbase}, a node's \textbf{twiddle factors}, $\omega^{ltfp}_n$ and $\omega^{rtfp}_n$, can be determined from $W_b$:
\begin{align*}
(nd.lp.\omega^{tfp}_n)(nd.lp.\omega^{W_b}_n) &= nd.\omega^{W_b}_n\\
nd.lp.\omega^{tfp}_n &= (nd.\omega^{W_b}_n)/(nd.lp.\omega^{W_b}_n).
\end{align*}
This can be expressed as (mod $n$) subtraction of powers:
\begin{equation*}
 nd.lp.tfp = nd.W_b - nd.lp.W_b \text{ (mod }n\text{)}.
\end{equation*}
The twiddle factor for a right parent is similarly defined as:
\begin{equation*}
 nd.rp.tfp = nd.rW_b - nd.rp.W_b \text{ (mod }n\text{)}.
\end{equation*}
\label{defn:graphtwiddlefactors} 
\end{definition}

\begin{example}
 Consider the node $nd(2,1,6)$ in Figure \ref{fig:conjugatesplitradixfft}. For this node, $nd.W_b=3$ and $nd.rW_b=9$ are specified. Also, $nd.lp.W_b=0$ and $nd.rp.W_b=12$ are specified. Hence,
\begin{align*}
 nd.lp.tfp &= nd.W_b - nd.lp.W_b \text{ (mod }n\text{)}\\
  &= 3 - 0 \text{ (mod }16\text{)}\\
  &= 3,
\end{align*}
which is the twiddle factor applied to that edge by $nd.lp$ as shown in the Figure. Likewise,
\begin{align*}
 nd.rp.tfp &= nd.rW_b - nd.rp.W_b \text{ (mod }n\text{)}\\
 &= 9 - 12 \text{ (mod }16\text{)}\\
 &= 13,
\end{align*}
which is the twiddle factor applied to that edge by $nd.rp$.
\end{example}

\begin{algorithm}[H]
\DontPrintSemicolon
\KwIn{Size-$n$ flowgraph with labeled invariants}
\KwOut{Size-$n$ flowgraph with twiddle factors assigned}
\ForEach {$nd \in flowgraph$}{
\eIf{$nd.stride \neq n$} {
  $nd.W_b \leftarrow rand()$ (mod $n$)\;
  $nd.rW_b \leftarrow nd.W_b+nd.W_s$ ( mod $n$)\;
}{
 $nd.W_b \leftarrow 0$\;
}
}
\ForEach {$nd \in flowgraph$}{
\If{$nd.stride \neq n$} {
  $nd.lp.tfp \leftarrow nd.W_b - nd.lp.W_b$ (mod $n$)\;
  $nd.rp.tfp \leftarrow nd.rW_b - nd.rp.W_b$ (mod $n$)\;
}
\If{$nd.stride = 1$} {
 $nd.tfp \leftarrow 0 - nd.W_b$ (mod $n$)\;
}
}
\caption{How to Generate a Random Member FFT Algorithm}
\label{algo:random} 
\end{algorithm}

\begin{example}
 Figure \ref{fig:randomfft} shows a random member from the family of FFTs realizable by a size-$8$ flowgraph. Consider node $nd(1,0,3)$. Since $nd.stride \neq n$, we assign a random integer (mod 8) of 3 to $nd.W_b$. Following Algorithm \ref{algo:random},
\begin{align*}
 nd.rW_b &= nd.W_b + nd.W_s \text{ (mod }n\text{)}\\
  &= 3 + 3 \text{ (mod }8\text{)}\\
  &= 6.
\end{align*}
This same process is repeated until $W_b$ and $rW_b$ have been assigned for all nodes with $stride \neq n$ in the flowgraph. Nodes with a single input, shown as dotted, always have $W_b=0$ following Definition \ref{defn:graphweightonbase}. Next, actual twiddle factors are computed from the weight assignments. Again consider node $nd(1,0,3)$ and the computation of twiddle factors for that node's parents:
\begin{align*}
 nd.lp.ltfp &= nd.W_b - nd.lp.W_b \text{ (mod }n\text{)}\\
  &= 3 - 1 \text{ (mod }8\text{)}\\
  &= 2\\
 nd.rp.ltfp &= nd.rW_b - nd.rp.W_b \text{ (mod }8\text{)}\\
  &= 6 - 2 \text{ (mod }8\text{)}\\
  &= 4.
\end{align*}
Since $nd$ has no children nodes ($nd.stride=1$) we must compute its twiddle factor as
\begin{align*}
 nd.tfp &= 0 - nd.lp.W_b \text{ (mod }n\text{)}\\
 &= 0 - 3 \text{ (mod }8\text{)}\\
 &= 5.
\end{align*}
This process is repeated until all twiddle factors are assigned. 
\end{example}

\begin{figure}[t]
\centering
\tikzstyle{fftnode}=[rectangle split, rectangle split parts=3, rounded corners, text badly centered, text width=8mm, minimum width=8.1mm, draw=black,inner sep=0.5mm]
\tikzstyle{fftterminal}=[rectangle, text centered]
\tikzstyle{arc}=[->]
\tikzstyle{fftgraph}=[>=latex]
\tikzstyle{fftmatrix}=[row sep=10mm, column sep={13mm,between origins}]
\tikzstyle{background}=[rectangle,rounded corners,fill=gray!25,inner sep=1.0mm]
\tikzstyle{polytext}=[anchor=120,inner sep=0mm]
\tikzstyle{costtext}=[anchor=220,inner sep=0mm]
\tikzstyle{stride}=[inner sep=1.0mm]
\input{random_size8}
\caption{Size-8 Random FFT Flowgraph}
\label{fig:randomfft}
\end{figure}
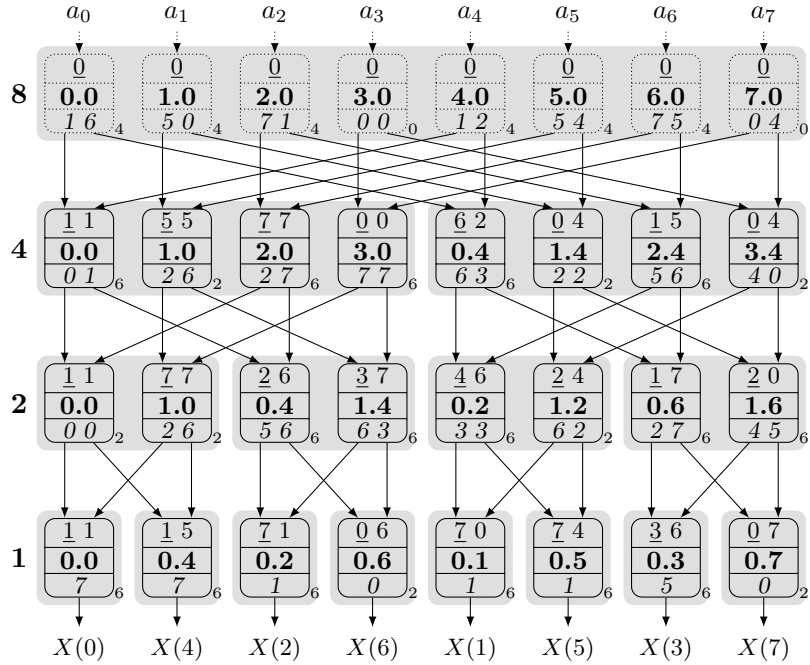

\section{Searching a Family of FFT Algorithms}
\label{sec:secondhalf}

In Section \ref{sec:firsthalf}, a flowgraph representation of common power-of-two FFTs is developed that defines an invariant \emph{weight stride}, $W_s$, for each node in the flowgraph. Algorithm \ref{algo:random} uses the \emph{weight stride} invariants to generate a new assignment of twiddle factors and hence a unique FFT. Since Algorithm \ref{algo:random} is arbitrary, a rich solution space of valid FFTs, called a family, results. In this section, we characterize the size of this family, specify a family as a Satisfiability Modulo Theory (SMT) problem, and demonstrate SMT solver-based search of this solution space. Although this search can be directed in various ways, we use it to prove the lowest total arithmetic complexity (fewest required FLOPs) when all twiddle factors are $n^{th}$ roots of unity.

\subsection{The Size of a Family}

The solution space of valid FFTs for a given flowgraph is extremely large!

\begin{definition}
A size-$n$ flowgraph's \textbf{solution cardinality} is $2^{n \log_2 n \log_2 n}$, and is the number of valid FFTs realizable by the given flowgraph. By direct examination of Algorithm \ref{algo:random}, each node $nd$ in a size-$n$ flowgraph, where $nd.stride \neq n$, is arbitrarily assigned some integer (mod $n$) to $nd.W_b$. Thus, there are $n=2^{\log_2 n}$ possible choices for a single node's $W_b$. And, since there are $n \log_2 n$ nodes in the flowgraph where $nd.stride \neq n$, there are $(2^{\log_2 n})^{n \log_2 n}=2^{n \log_2 n \log_2 n}$ possible assignments of all $W_b$ in a size-$n$ flowgraph.  
\label{defn:solutioncardinality} 
\end{definition}

\begin{example}
For a size-256 flowgraph, there are
\begin{equation*}
 \textbf{solution cardinality} = 2^{256 \log_2 256 \log_2 256} = 2^{16384} 
\end{equation*}
valid FFTs possible. One can better appreciate the magnitude of this number when reminded that the estimated number of atoms in the universe is $2^{264}$ and the current fastest supercomputer performs $2^{144}$ FLOPs per second. Yet this number is very small when compared to all valid and invalid $n^{th}$ root of unity assignments possible for twiddle factors, $2^{34816}$. Thus, for this size-256 flowgraph, there is just a 1 in $2^{18432}$ chance of guessing correct twiddle factors.
\end{example}

Although the solution space is immense, in practice we are only interested in family members with desirable qualities, such as fewest required FLOPs, better precision\footnote{All family members are exact and do not sacrifice numerical accuracy. Imprecision arises from choice of twiddle factors with values very close to zero and consequent floating-point representation limitations.}, improved performance or ease of implementation on a specific microarchitecture. Consequently, we need a way to search this space and find these more desirable family members. 

\subsection{A First SMT Formulation}
\label{sec:smtformulation}

Because of the way concepts were developed in Section \ref{sec:firsthalf}, it is straight-forward to model Algorithm \ref{algo:random} as an SMT problem. This is best illustrated by considering Listing \ref{code:smt}, which shows portions of the SMT model in SMT-LIB 1.2 format\cite{ranise2006smt} that is created to find a lowest arithmetic complexity instance of a size-16 FFT. After a standard preamble, lines 4 and 5 declare the external inputs $nd(2,1,6).W_b$ and $nd(2,1,14).W_b$, which are both 4-bit vectors. Although not shown in the listing, inputs for all undetermined $W_b$ are included. It is for these variables that the SMT solver attempts to find a satisfying assignment.  For nodes where $nd.stride=n$, the value $nd.W_b$ is predetermined to be 0 and is declared as a constant. An example of this is shown in line 9 and corresponds to Algorithm \ref{algo:random} line 6. Next, $rW_b$ for all nodes is computed via addition of $W_b$ and $nd.stride$. An instance of this is seen in line 11 and corresponds to Algorithm \ref{algo:random} line 4. Note that all addition and subtraction is naturally (mod $n$) given the fixed-size bitvectors in the SMT formulation. Twiddle factors for all nodes are computed as illustrated in lines 13 and 14. This corresponds to lines 9-12 of Algorithm \ref{algo:random}. 

Unlike Algorithm \ref{algo:random}, the objective of the SMT model is to find the lowest arithmetic cost. For this, we must compute the cost implied by every twiddle factor. Lines 16-18 show the computation of cost predicates $c0,c4,c6$, (0, 4 or 6 FLOPs for multiplication, respectively), for the left twiddle factor of $nd(4,1,12)$. Not shown in this listing are any necessary predicates $c0,c2,c4,c6$ for multiplication cost incurred by the right twiddle factor.  Line 19 shows cost predicates used in an if--then--else (ITE) tree to compute the multiplication FLOPs required by a node. We compute predicates first and then a numeric cost as the predicates are useful in defining pruning constraints later. In line 21, a total cost is computed by simply adding up all node multiplication costs. Finally, line 22 constrains the total cost to be less than or equal to some constant, and line 23 specifies that this multiplication FLOP constraint is satisfied. Note that the FLOP count due to a node's addition is constant for the flowgraphs under consideration and is not explicitly included in the SMT models.

\lstset{basicstyle=\scriptsize, numbers=left, stepnumber=1, caption=Sample SMT Code, label=code:smt, xleftmargin=20pt, xrightmargin=5pt, captionpos=b,}
\begin{lstlisting}
(benchmark example1
:logic QF_BV
...
:extrafuns ((Wb_2_1_6 BitVec[4]))
:extrafuns ((Wb_2_1_14 BitVec[4]))
...
:formula
...
(let (?Wb_16_14_0 bv0[4])
...
(let (?rWb_2_1_6 (bvadd Wb_2_1_6 bv6[4]))
...
(let (?ltfp_4_1_12 (bvsub Wb_2_1_6 ?Wb_4_1_12))
(let (?ltfp_4_3_12 (bvsub ?rWb_2_1_6 ?Wb_4_3_12))
...
(flet ($c0_4_1_12 (= (extract[1:0] ?ltfp_4_1_12) bv0[2]))
(flet ($c4_4_1_12 (and (= (extract[0:0] ?ltfp_4_1_12) bv0[1]) (not $c0_4_1_12)))
(flet ($c6_4_1_12 (not (= (extract[0:0] ?ltfp_4_1_12) bv0[1])))
(let (?cost_4_1_12 (ite $c6_4_1_12 bv6[4] (ite $c4_4_1_12 bv4[4] bv0[4])))
...
(let (?totalcost (bvadd ?cost_2_2_1 (bvadd ?cost_4_1_12 ?cost_4_3_12)) ...
(flet ($maxcost (bvule ?totalcost bv22[4]))
$maxcost
)...)
\end{lstlisting}

In practice, more care is given to the total cost addition seen in line 21. A balanced adder tree is constructed, where each add uses only as many bits as required for the worst case. Furthermore, following the recursive structure in the FFT apparent from the polynomial view, cost for smaller FFTs are computed first and then combined to compute the cost for larger parent FFTs.    This total cost computation is effectively a pseudo-Boolean constraint, and we have tried implementing it as an if--then--else (ITE) tree similar to the ROBDD-techniques described in \cite{een2006translating}. Our experience is that the adder tree is 2-3 times better in terms of SMT computation time for this particular problem with the Boolector SMT solver\cite{brummayer2009boolector}. We did not implement the sorter-based technique described in \cite{een2006translating}.  

To find the lowest arithmetic complexity, the SMT model is repeatedly solved, each time with a lower value for the constant seen in line 22. At some point the model becomes unsatisfiable and the lowest possible arithmetic complexity is known. Unfortunately, this straight-forward implementation does not scale up. For flowgraphs of size-32, the time for computing the unsatisfiability of a 455 FLOP solution requires 30 seconds using the Boolector solver \cite{brummayer2009boolector} on a 64-bit Intel Core i7 Linux machine. At size-64 and for unsatisfiable cost of 1159 FLOPs, we reach our timeout of 24 hours without determining unsatisfiability.

\subsection{Cost Symmetries}
\label{sec:costsymmetries}

As formulated so far, the SMT model supports the full range of possible values for each twiddle factor since each twiddle factor is modeled as a size-$m$ bitvector. This much information is not necessary for finding the lowest possible arithmetic complexity, and only adds to the complexity of the model. Instead, it is possible to express every twiddle factor as
\begin{equation*}
 \omega_n^{tfp} = \omega_n^{tfp - (tfp \text{ (mod }\nicefrac{n}{4}\text{)})}\omega_n^{tfp \text{ (mod }\nicefrac{n}{4}\text{)}}.
\end{equation*}
In this expression, $\omega_n^{tfp - (tfp \text{ (mod }\nicefrac{n}{4}\text{)})}$ specifies the quadrant in which $\omega_n^{tfp}$ lies and is always a free multiplication by $1$,$-1$,$i$ or $-i$. Consequently, the portion of $\omega_n^{tfp}$ that solely contributes to multiplication cost can be represented by just a quarter of the $n^{th}$ roots of unity and is defined as $\psi_n^{tfp} = \omega_n^{tfp \text{ (mod }\nicefrac{n}{4}\text{)}}$. To simplify the SMT model and upcoming partitioning, we suppress the quadrant rotations, $\omega_n^{tfp - (tfp \text{ (mod }\nicefrac{n}{4}\text{)})}$, and only reason with $\psi_n$.

Multiplication of two $\psi_n$ is well defined and can be expressed as modular arithmetic. Consider the multiplication
\begin{equation*}
 \omega_n^{a+b \text{ (mod }n\text{)}} = \omega_n^{a}\omega_n^{b},
\end{equation*}
which can be re-expressed as
\begin{equation*}
 \omega_n^{a + b - (a+b \text{ (mod }\nicefrac{n}{4}\text{)})\text{ (mod }n\text{)}}  \psi_n^{a + b \text{ (mod }\nicefrac{n}{4}\text{)}}  = \omega_n^{a - (a \text{ (mod }\nicefrac{n}{4}\text{)})}\psi_n^{a} \omega_n^{b - (b \text{ (mod }\nicefrac{n}{4}\text{)})}\psi_n^{b}.
\end{equation*}
If all $\omega_n$ specifying quadrant rotations are ignored, multiplication of two $\psi_n$ is just
\begin{equation*}
 \psi_n^{a+b \text{ (mod }\nicefrac{n}{4}\text{)}} = \psi_n^{a}\psi_n^{b},
\end{equation*}
which can be expressed easily using modular arithmetic in the SMT model.

From the bitvector perspective, suppressing $\omega_n$ quadrant rotations means that the two most significant bits of every weight on base, $W_b$ or $rW_b$, need not be included in the SMT model. The SMT solver finds a satisfying assignment for all but the two most significant bits of every weight on base. The two most significant bits are then picked at random as done in Algorithm \ref{algo:random} without altering cost. In the end, all bits must be assigned to realize a correct FFT.

Eliminating these cost symmetries in the SMT model reduces a size-$n$ flowgraph's solution cardinality to $2^{n \log_2 n ((\log_2 n)-2)}$. For a size-256 flowgraph, this is a substantial reduction in the size of the solution space from $2^{16384}$ to $2^{12288}$. Computation time for proving that a size-32 flowgraph has no solution with total cost equal to or less than 455 FLOPs is now 27 seconds. The timeout of 24 hours is still reached for a size-64 flowgraph constrained to 1159 FLOPs. It is possible that the SMT computation time improves only modestly since the SMT solver is detecting these cost symmetries without explicit help.

\subsection{Butterflies}

 The next three techniques to simplify and partition the SMT model require reasoning with FFT butterflies. Although FFT butterflies are a well established idea, we define and review concepts relevant to our SMT model.  
 
\begin{definition}
A size-$q$ \textbf{butterfly} is any subgraph of a size-$n$ FFT flowgraph that is graph isomorphic to a size-$q$ FFT flowgraph where $q \leq n$. A butterfly's \textbf{canonical label} is $bf(nd.stride,nd.base,nd.W_s,q)$ for the single node $nd \in bf$ such that $nd.stride,nd.base$ and $nd.W_s$ are less than or equal to the $stride,base$ and $W_s$ of any other node in the butterfly. As with all FFT flowgraphs considered in this paper, $q$ must be some power of two.
\label{defn:butterfly} 
\end{definition}

\begin{example}
In Figure \ref{fig:conjugatesplitradixfft}, the butterfly $bf(1,0,0,2)$ contains the four nodes $nd(1,0,0)$, $nd(1,0,8)$, $nd(2,0,0)$ and $nd(2,1,0)$. The expected traditional butterfly structure is clearly seen with the node used for identification, $nd(1,0,0)$, at the bottom left. This same node, $nd(1,0,0)$, is also used to identify the size-$4$ butterfly $bf(1,0,0,4)$ which contains 12 nodes and is also clearly visible. Less obvious are small butterflies that appear toward the top of the flowgraph such as $bf(4,2,0,2)$ which contains the four nodes $nd(4,2,0)$, $nd(4,2,8)$, $nd(8,2,0)$ and $nd(8,6,0)$. The larger butterfly $bf(4,2,0,4)$ which contains 12 nodes can also be traced with $nd(4,2,0)$ anchoring the bottom left corner. Finally, the entire flowgraph in Figure \ref{fig:conjugatesplitradixfft} can be denoted as $bf(1,0,0,16)$.
\end{example}

It is also useful to refer to nodes in an arbitrary butterfly not by canonical node label but by relative position. To facilitate this, one can view the nodes of a butterfly as forming a matrix and use matrix row,column indexing to refer to a specific node, $nd_{r,c}$. For example, for any size-$2$ butterfly, the top-left corner node is $nd_{0,0}$, the top-right corner node is $nd_{0,1}$, the bottom-left corner node is $nd_{1,0}$, and the bottom-right corner node is $nd_{1,1}$.

\begin{prop}
For size-$2$ and size-$4$ butterflies in the flowgraph, all $W_s$ for nodes in the same row are congruent modulo $\nicefrac{n}{4}$. The value to which they are all congruent modulo $\nicefrac{n}{4}$ is referred to as $nd_{r,*}.W_s$. This property arises from the correspondence of $W_s$ in a flowgraph to the polynomial view as described in Section \ref{sec:polyview}.
\label{prop:congruentweightstrides}
\end{prop}

\begin{example}
Consider the size-$4$ butterfly $bf(1,0,3,4)$ from Figure \ref{fig:conjugatesplitradixfft}. By inspection, 
\begin{align*}
  nd_{0,*}.W_s &= \{12,12,12,12\} & &\equiv 0 \text{ (mod }\nicefrac{n}{4}\text{)}\\
  nd_{1,*}.W_s &= \{6,6,14,14\} & &\equiv 2 \text{ (mod }\nicefrac{n}{4}\text{)}\\
  nd_{2,*}.W_s &= \{3,11,7,15\} & &\equiv 3 \text{ (mod }\nicefrac{n}{4}\text{)}.\\
\end{align*}
\end{example}

\subsection{Shared Twiddle Factors}

Our formulation permits two multiplications, by $\omega^{ltfp}_n$ and $\omega^{rtfp}_n$, per node in the FFT flowgraph. Although this generality may be useful for some algorithms, we show here that it is not needed when minimizing the total FLOP count is the objective. In fact, it only increases the complexity of the SMT model.

\begin{theorem}
 For any size-$2$ butterfly, $bf_1$, such that left and right twiddle factors are unshared per node in row 0 ($nd_{0,c}.ltfp \neq nd_{0,c}.rtfp$) but shared per node in row 1 ($nd_{1,c}.ltfp = nd_{1,c}.rtfp$) there exists another size-$2$ butterfly, $bf_2$, such that left and right twiddle factors are shared per node for all nodes, that realizes all final weighted sums $X(0)$ and $X(1)$ possible by $bf_1$. Furthermore, no $bf_1$ exists with lower FLOP count than some $bf_2$.
\label{thm:sharedtwiddlefactors}
\end{theorem}

\begin{proof}
The proof is in two parts. First, we prove the existence of three different $bf_2$. Consider the computation performed by $bf_1$,
\begin{subequations}\label{eq:arb2}
\begin{align}
 X(0) &\equiv \wnf^{nd_{1,0}.tfp}(a_0\wnf^{nd_{0,0}.ltfp} + a_1\wnf^{nd_{0,1}.ltfp}) \mdf,\label{eq:X0} \\
 X(1) &\equiv \wnf^{nd_{1,1}.tfp}(a_0\wnf^{nd_{0,0}.rtfp} + a_1\wnf^{nd_{0,1}.rtfp}) \mdf.\label{eq:X1}
\end{align}
\end{subequations}
Because of Property \ref{prop:congruentweightstrides} and Theorem \ref{thm:weightstride}, $nd_{0,*}$ left and right twiddle factors are related by a common weight stride,
\begin{equation*}
 \wnf^{nd_{1,*}.W_s} \equiv \frac{\wnf^{nd_{0,1}.W_b}\wnf^{nd_{0,1}.ltfp}}{\wnf^{nd_{0,0}.W_b}\wnf^{nd_{0,0}.ltfp}} 
		     \equiv \frac{\wnf^{nd_{0,1}.W_b}\wnf^{nd_{0,1}.rtfp}}{\wnf^{nd_{0,0}.W_b}\wnf^{nd_{0,0}.rtfp}} \mdf.
\end{equation*}
This simplifies to
\begin{equation}
 \frac{\wnf^{nd_{0,1}.ltfp}}{\wnf^{nd_{0,0}.ltfp}} \equiv \frac{\wnf^{nd_{0,1}.rtfp}}{\wnf^{nd_{0,0}.rtfp}} \mdf
\label{eq:ltfprtfprelation}, 
\end{equation}
which directly relates left and right twiddle factors for nodes in row 0 of $bf_1$.

Equations \ref{eq:arb2} and \ref{eq:ltfprtfprelation} can be used to derive three different $bf_2$, labeled $bf_{2A}$, $bf_{2B}$ and $bf_{2C}$. Butterfly $bf_{2A}$, with $nd_{0,0}.lftp=nd_{0,0}.rtfp=0$, is created by multiplying Equation \ref{eq:X0} by $1=\frac{\wnf^{nd_{0,0}.ltfp}}{\wnf^{nd_{0,0}.ltfp}}$ and Equation \ref{eq:X1} by $1=\frac{\wnf^{nd_{0,0}.rtfp}}{\wnf^{nd_{0,0}.rtfp}}$,
\begin{align*}
 X(0) &\equiv (\wnf^{nd_{1,0}.tfp}\wnf^{nd_{0,0}.ltfp})(a_0\frac{\wnf^{nd_{0,0}.ltfp}}{\wnf^{nd_{0,0}.ltfp}} + a_1\frac{\wnf^{nd_{0,1}.ltfp}}{\wnf^{nd_{0,0}.ltfp}}) \mdf \\
 X(1) &\equiv (\wnf^{nd_{1,1}.tfp}\wnf^{nd_{0,0}.rtfp})(a_0\frac{\wnf^{nd_{0,0}.rtfp}}{\wnf^{nd_{0,0}.rtfp}} + a_1\frac{\wnf^{nd_{0,1}.rtfp}}{\wnf^{nd_{0,0}.rtfp}}) \mdf.
\end{align*}
After simplification of twiddle factors, the two twiddle factors applied to $a_0$ are now shared ($\wnf^0$) and the two twiddle factors applied to $a_1$ are also shared due to Equation \ref{eq:ltfprtfprelation}.

Butterfly $bf_{2B}$, with $nd_{0,0}.rtfp$ made equal to $nd_{0,0}.ltfp$, is created by multiplying Equation \ref{eq:X1} by $1=\frac{\wnf^{nd_{0,0}.ltfp}\wnf^{nd_{0,0}.rtfp}}{\wnf^{nd_{0,0}.ltfp}\wnf^{nd_{0,0}.rtfp}}$,
\begin{align*}
 X(0) &\equiv \wnf^{nd_{1,0}.tfp}(a_0\wnf^{nd_{0,0}.ltfp} + a_1\wnf^{nd_{0,1}.ltfp}) \mdf\\ 
 X(1) &\equiv (\wnf^{nd_{1,1}.tfp}\frac{\wnf^{nd_{0,0}.rtfp}}{\wnf^{nd_{0,0}.ltfp}})(a_0\frac{\wnf^{nd_{0,0}.ltfp}\wnf^{nd_{0,0}.rtfp}}{\wnf^{nd_{0,0}.rtfp}} + a_1\frac{\wnf^{nd_{0,0}.ltfp}\wnf^{nd_{0,1}.rtfp}}{\wnf^{nd_{0,0}.rtfp}}) \mdf.
\end{align*}
Again from direct inspection and application of Equation \ref{eq:ltfprtfprelation}, every node shares left and right twiddle factors.

Butterfly $bf_{3C}$, with $nd_{0,0}.ltfp$ made equal to $nd_{0,0}.rtfp$, is created by multiplying Equation \ref{eq:X0} by $1=\frac{\wnf^{nd_{0,0}.ltfp}\wnf^{nd_{0,0}.rtfp}}{\wnf^{nd_{0,0}.ltfp}\wnf^{nd_{0,0}.rtfp}}$,
\begin{align*}
 X(0) &\equiv (\wnf^{nd_{1,0}.tfp}\frac{\wnf^{nd_{0,0}.ltfp}}{\wnf^{nd_{0,0}.rtfp}})(a_0\frac{\wnf^{nd_{0,0}.ltfp}\wnf^{nd_{0,0}.rtfp}}{\wnf^{nd_{0,0}.ltfp}} + a_1\frac{\wnf^{nd_{0,1}.ltfp}\wnf^{nd_{0,0}.rtfp}}{\wnf^{nd_{0,0}.ltfp}}) \mdf \\ 
 X(1) &\equiv \wnf^{nd_{1,1}.tfp}(a_0\wnf^{nd_{0,0}.rtfp} + a_1\wnf^{nd_{0,1}.rtfp}) \mdf.
\end{align*}
Here, too, every node shares left and right twiddle factors.

Second, exhaustive search with SMT is used to prove that no $bf_1$ exists with lower FLOP count than some $bf_2$. The SMT-based proof is a miter between $bf_1$ and $bf_{2A}$, $bf_{2B}$ and $bf_{2C}$. The $bf_1$ side of the miter is a size-$2$ FFT modeled in SMT as described in Section \ref{sec:smtformulation}. Additional constraints that $nd_{1,0}.ltfp = nd_{1,0}.rtfp$ and $nd_{1,1}.ltfp = nd_{1,1}.rtfp$ are added for $bf_1$. The $bf_2$ side of the miter includes models for all three cases A, B and C. These are also modeled in SMT as described in Section \ref{sec:smtformulation} but with the additional constraint that each node has just one twiddle factor, $tfp$. Furthermore, the constraints $\nobreak{bf_{2A}.nd_{0,0}=0}$, $bf_{2B}.nd_{0,0}.tfp = bf_1.nd_{0,0}.ltfp$, and $bf_{2C}.nd_{0,0}.tfp = nd_{0,0}.rtfp$ are included with the respective $bf_2$ models. Input values $a_j$ with arbitrary initial weights on base $nd_{0,c}.W_b$ and row weight strides $nd_{1,*}.W_s$ are common to all $bf_1$ and $bf_2$.
The free variables decided by the SMT solver include these common initial weights on base and row weight strides as well as weights on base per node for all row 1 nodes in $bf_1$ and $bf_2$. FLOP counts for $bf_1$, $bf_{2A}$, $bf_{2B}$ and $bf_{2C}$, are individually and explicitly tallied within the SMT model. The question posed to the SMT solver is to find a $bf_1$ with lower FLOP count than $bf_{2A}$, $bf_{2B}$ or $bf_{2C}$. The theorem is proved for some $n$ if the SMT solver returns unsatisfiable. The proof can be run once for every size-$n$ FFT flowgraph under consideration, or induction can be used to establish the result for $n+1$ and higher. 
\end{proof}

\begin{example}
Consider the concrete computation performed by some $bf_1$ from a size-$16$ FFT flowgraph expressed as,
\begin{align*}
 X(0) &\equiv \psi_{16}^0(a_0\psi_{16}^1 + a_1\psi_{16}^3) \mdf \\
 X(1) &\equiv \psi_{16}^0(a_0\psi_{16}^3 + a_1\psi_{16}^1) \mdf.
\end{align*}
By substituting into Equation \ref{eq:ltfprtfprelation}, we see that this is a valid butterfly with weight stride adhering to Property \ref{prop:congruentweightstrides},
\begin{equation*}
 \frac{\psi_{16}^3}{\psi_{16}^1} \equiv \frac{\psi_{16}^1}{\psi_{16}^3} \equiv \psi_{16}^2 \mdf.
\end{equation*}
Some sharing can occur during the complex multiplication of $a_0$ and $a_1$ with these left and right twiddle factors since $\Re{(\psi_{16}^1)}=\Im{(\psi_{16}^3)}$ and $\Re{(\psi_{16}^3)}=\Im{(\psi_{16}^1)}$. Hence, the multiplication FLOP count for $bf_1$ is only $16=8+8+0+0$.

Butterfly $bf_{2A}$ has only a $\psi_{16}^0$ twiddle factor applied to $a_0$,
\begin{align*}
 X(0) &\equiv \psi_{16}^1(a_0\psi_{16}^0 + a_1\psi_{16}^2) \mdf \\
 X(1) &\equiv \psi_{16}^3(a_0\psi_{16}^0 + a_1\psi_{16}^2) \mdf.
\end{align*}
Note that the twiddle factors applied to $a_0$ and $a_1$, $\psi_{16}^0$, are shared for $X(0)$ and $X(1)$. The results for $X(0)$ and $X(1)$ are still equivalent to $bf_1$ as the final twiddle factors, $nd_{1,c}.tfp$, are now adjusted by factoring out $\psi_{16}^1$ and $\psi_{16}^3$ respectively. The total multiplication cost for $bf_{2A}$ is $16=0+4+6+6$, which is the same as $bf_1$.

Butterfly $bf_{2B}$ has only a $\psi_{16}^1$ twiddle factor applied to $a_0$,
\begin{align*}
 X(0) &\equiv \psi_{16}^0(a_0\psi_{16}^1 + a_1\psi_{16}^3) \mdf \\
 X(1) &\equiv \psi_{16}^2(a_2\psi_{16}^1 + a_1\psi_{16}^3) \mdf.
\end{align*}
The $\psi_{16}^2$ is factored out of the sum in $X(1)$ to maintain equivalence with $bf_1$. The total cost for $bf_{2B}$ is also $16=6+6+0+4$.

Butterfly $bf_{2C}$ has only a $\psi_{16}^3$ twiddle factor applied to $a_0$,
\begin{align*}
 X(0) &\equiv \psi_{16}^2(a_0\psi_{16}^3 + a_1\psi_{16}^1) \mdf \\
 X(1) &\equiv \psi_{16}^0(a_2\psi_{16}^3 + a_1\psi_{16}^1) \mdf.
\end{align*}
The $\psi_{16}^2$ is factored out of the sum in $X(0)$ to maintain equivalence with $bf_1$. The total cost for $bf_{2C}$ is also $16=6+6+4+0$.
Although all butterflies in this example have the same multiplication cost, this is not always the case in general. The SMT portion of the proof of Theorem \ref{thm:sharedtwiddlefactors} shows that at least one case of $bf_2$ will have FLOP count less than or equal to $bf_1$.
\end{example}

The definition of $bf_1$ in Theorem \ref{thm:sharedtwiddlefactors} requires that twiddle factors be shared per node in row 1, $nd_{1,c}.ltfp = nd_{1,c}.rtfp$. Butterflies meeting this constraint only occur at the bottom of the FFT flowgraph, where a single weight may be applied to some $X(k)$. But after applying Theorem \ref{thm:sharedtwiddlefactors} to all terminal size-$2$ butterflies in the bottom row, we now have shared twiddle factors in the next to the bottom row of the FFT flowgraph. Therefore, Theorem \ref{thm:sharedtwiddlefactors} can be applied iteratively to the entire flowgraph, starting from the bottom and proceeding to the top, so that all nodes have a single twiddle factor, $tfp$, without any FLOP count penalty.

\begin{prop}
For any size-$2$ butterfly from a FFT flowgraph, if all nodes have a single shared twiddle factor $tfp$, then $nd_{1,0}.W_b \equiv nd_{1,1}.W_b \mdf$. This is because $nd_{1,0}$ and $nd_{1,1}$ both have the same left parent with the same $tfp \mdf$ applied.
\label{prop:congruentweightonbase}
\end{prop}

Given Property \ref{prop:congruentweightonbase}, Algorithm \ref{algo:random} can now be updated so that the SMT formulation assigns a weight on base per size-$2$ butterfly and not per node. Instead of assigning a random $W_{b}$ per node as seen in line 3 of the algorithm, a random $W_{b}$ is assigned per bottom two nodes of every size-$2$ butterfly. This reduces the number of free $W_{b}$ variables by half and substantially speeds up the SMT-based search. Shared twiddle factors in the SMT model reduce a size-$n$ flowgraph's solution cardinality to $2^{\frac{n}{2} \log_2 n ((\log_2 n)-2)}$. For a size-256 flowgraph, this further reduces the solution space to $2^{6144}$. Computation time for proving that a size-32 flowgraph has no solution with total cost less than or equal to 455 FLOPs is now 3.5 seconds. The timeout of 24 hours is still reached for a size-64 flowgraph constrained to 1159 FLOPs.

\subsection{Partitioning}

Every SMT model formulated so far has been monolithic, and it has been computationally difficult to prove the lowest arithmetic complexity for any FFT larger than size-32. In this section, we show that analysis of butterflies at the top and bottom of the flowgraph can be used to partition larger FFTs into several smaller SMT models that can be solved. This analysis is facilitated by explicitly writing out the final weight on base computations, with all operations congruent$\mdf$, for an arbitrary size-$4$ butterfly:
%\begin{figure}[H]
%\begin{center}
\begin{align}
 X(0).W_b &\equiv nd_{2,0}.tfp + nd_{1,0}.tfp + nd_{0,0}.tfp + nd_{0,0}.W_b \notag \\
 X(0).W_b &\equiv nd_{2,0}.tfp + nd_{1,0}.tfp + nd_{0,2}.tfp + nd_{0,2}.W_b - nd_{1,*}.W_s \notag \\
 X(0).W_b &\equiv nd_{2,0}.tfp + nd_{1,1}.tfp + nd_{0,1}.tfp + nd_{0,1}.W_b - nd_{2,*}.W_s \notag \\
 X(0).W_b &\equiv nd_{2,0}.tfp + nd_{1,1}.tfp + nd_{0,3}.tfp + nd_{0,3}.W_b - nd_{2,*}.W_s - nd_{1,*}.W_s \notag \displaybreak[0]\\
 X(2).W_b &\equiv nd_{2,1}.tfp + nd_{1,0}.tfp + nd_{0,0}.tfp + nd_{0,0}.W_b \notag \\
 X(2).W_b &\equiv nd_{2,1}.tfp + nd_{1,0}.tfp + nd_{0,2}.tfp + nd_{0,2}.W_b - nd_{1,*}.W_s \notag \\
 X(2).W_b &\equiv nd_{2,1}.tfp + nd_{1,1}.tfp + nd_{0,1}.tfp + nd_{0,1}.W_b - nd_{2,*}.W_s \notag \\
 X(2).W_b &\equiv nd_{2,1}.tfp + nd_{1,1}.tfp + nd_{0,3}.tfp + nd_{0,3}.W_b - nd_{2,*}.W_s - nd_{1,*}.W_s \label{eq:arbitrarysize4butterfly} \displaybreak[0]\\
 X(1).W_b &\equiv nd_{2,3}.tfp + nd_{1,3}.tfp + nd_{0,0}.tfp + nd_{0,0}.W_b \notag \\
 X(1).W_b &\equiv nd_{2,3}.tfp + nd_{1,3}.tfp + nd_{0,2}.tfp + nd_{0,2}.W_b - nd_{1,*}.W_s \notag \\
 X(1).W_b &\equiv nd_{2,3}.tfp + nd_{1,4}.tfp + nd_{0,1}.tfp + nd_{0,1}.W_b - nd_{2,*}.W_s \notag \\
 X(1).W_b &\equiv nd_{2,3}.tfp + nd_{1,4}.tfp + nd_{0,3}.tfp + nd_{0,3}.W_b - nd_{2,*}.W_s - nd_{1,*}.W_s \notag \displaybreak[0]\\
 X(4).W_b &\equiv nd_{2,4}.tfp + nd_{1,3}.tfp + nd_{0,0}.tfp + nd_{0,0}.W_b \notag \\
 X(4).W_b &\equiv nd_{2,4}.tfp + nd_{1,3}.tfp + nd_{0,2}.tfp + nd_{0,2}.W_b - nd_{1,*}.W_s \notag \\
 X(4).W_b &\equiv nd_{2,4}.tfp + nd_{1,4}.tfp + nd_{0,1}.tfp + nd_{0,1}.W_b - nd_{2,*}.W_s \notag \\
 X(4).W_b &\equiv nd_{2,4}.tfp + nd_{1,4}.tfp + nd_{0,3}.tfp + nd_{0,3}.W_b - nd_{2,*}.W_s - nd_{1,*}.W_s. \notag 
\end{align}
%\caption{Arbitrary Size-4 Butterfly $W_b$ Computations}
%\captionof{figure}{Arbitrary Size-4 Butterfly $W_b$ Computations}
%\label{fig:arbitrarysize4butterfly}
%\end{figure}
%\end{center} 
All weights on base internal to the butterfly have been eliminated by repeated substitution. All weight strides for nodes in the same row are congruent due to Property \ref{prop:congruentweightstrides}. It is instructive to trace all 16 paths from an input operand to an output value for a size-$4$ butterfly and verify that the weight on base computation for that path is included in Equation \ref{eq:arbitrarysize4butterfly}. 

\subsubsection{Partitioning Using Original Butterflies}

At the top of a flowgraph, all $a_j$ have a weight of 1, $\omega^0_n$. Butterflies that have input values which are some of these original $a_j$ are called {\bf original butterflies}. Analysis of original butterflies can exploit this known weight on $a_j$ to partition the FFT flowgraph and hence the SMT model.

\begin{prop}
The weight stride for all nodes in any butterfly that includes only nodes belonging to $f$ mod $x^*-1$, $x^*+1$, $x^*-i$ and $x^*+i$ from the polynomial factor tree is congruent to 0 (mod $\nicefrac{n}{4}$). This follows from the weight stride relationship to the polynomial view established in Section \ref{sec:polyview}.
\label{prop:originalbutterflies}
\end{prop}

\begin{example}
Consider the size-$4$ original butterfly $bf(4,2,0,4)$ from Figure \ref{fig:conjugatesplitradixfft}. By inspection, $W_s$ for all nodes in this butterfly \{0,4,8,12\} is congruent to 0 (mod 4). All size-$4$ original butterflies, and some larger, exhibit Property \ref{prop:originalbutterflies}.
\end{example}

\begin{theorem}
 For any arbitrary size-$4$ original butterfly, $bf_1$, there exists another size-$4$ butterfly, $bf_2$, which has zero-cost twiddle factors for nodes in rows 0 and 1, such that all realizable final weighted sums $X(k)$ of $bf_1$ can be realized by $bf_2$. Furthermore, no $bf_1$ exists with lower FLOP count than this $bf_2$.
\label{thm:topbutterflies}
\end{theorem}

\begin{proof}
The proof is in two parts. First, to prove all realizable final weighted sums of $bf_1$ can be achieved by $bf_2$, we substitute 0 for all initial weights ($nd_{0,*}.W_b = 0$), for all twiddle factors in rows 0 and 1 ($nd_{0,*}.tfp=nd_{1,*}.tfp=0$) and for all weight strides ($nd_{1,*}.W_s=nd_{2,*}.W_s=0$), into the expressions from Equation \ref{eq:arbitrarysize4butterfly}:
\begin{align*}
 X(0).W_b &\equiv nd_{2,0}.tfp + 0 + 0 + 0 \\
 X(0).W_b &\equiv nd_{2,0}.tfp + 0 + 0 + 0 - 0 \\
 X(0).W_b &\equiv nd_{2,0}.tfp + 0 + 0 + 0 - 0 \\
 X(0).W_b &\equiv nd_{2,0}.tfp + 0 + 0 + 0 - 0 - 0 \displaybreak[0]\\
 X(2).W_b &\equiv nd_{2,1}.tfp + 0 + 0 + 0 \\
 X(2).W_b &\equiv nd_{2,1}.tfp + 0 + 0 + 0 - 0 \\
 X(2).W_b &\equiv nd_{2,1}.tfp + 0 + 0 + 0 - 0 \\
 X(2).W_b &\equiv nd_{2,1}.tfp + 0 + 0 + 0 - 0 - 0 \displaybreak[0]\\
 X(1).W_b &\equiv nd_{2,3}.tfp + 0 + 0 + 0 \\
 X(1).W_b &\equiv nd_{2,3}.tfp + 0 + 0 + 0 - 0 \\
 X(1).W_b &\equiv nd_{2,3}.tfp + 0 + 0 + 0 - 0 \\
 X(1).W_b &\equiv nd_{2,3}.tfp + 0 + 0 + 0 - 0 - 0 \displaybreak[0]\\
 X(4).W_b &\equiv nd_{2,4}.tfp + 0 + 0 + 0 \\
 X(4).W_b &\equiv nd_{2,4}.tfp + 0 + 0 + 0 - 0 \\
 X(4).W_b &\equiv nd_{2,4}.tfp + 0 + 0 + 0 - 0 \\
 X(4).W_b &\equiv nd_{2,4}.tfp + 0 + 0 + 0 - 0 - 0 
\end{align*}
By direct inspection we establish that any final weight on base can be realized by twiddle factors of nodes only in the last row. 

Second, exhaustive search with SMT is used to prove that no $bf_1$ exists with lower FLOP count than its $bf_2$. The SMT proof is a miter that includes $bf_1$ and $bf_2$. The $bf_2$ side of the miter is a direct translation to SMT of the final weight on base computations just seen. The $bf_1$ side of the miter is created by substituting 0 for all initial weights ($nd_{0,*}.W_b = 0$) and for all weight strides ($nd_{1,*}.W_s=nd_{2,*}.W_s=0$) in the expressions from Equation \ref{eq:arbitrarysize4butterfly}. Final weights $bf_1.X(k).W_b$ are required to be equivalent to corresponding final weights $bf_2.X(k).W_b$. FLOP counts for $bf_1$ and $bf_2$ are individually and explicitly tallied within the SMT model. The question posed to the SMT solver is to find a $bf_1$ with lower FLOP count than $bf_2$. The theorem is proved for some $n$ if the SMT solver returns unsatisfiable. The proof can be run once for every size-$n$ FFT under consideration, or induction can be used to establish the result for $n+1$ and higher. 
\end{proof}

This theorem appears to conflict with decimation-in-time FFTs, such as shown in Figure \ref{fig:conjugatesplitradixfft}, where costly twiddle factors appear in the first two rows of the FFT. Consider the size-$4$ butterfly $bf(4,2,0,4)$ from Figure \ref{fig:conjugatesplitradixfft}. There is multiplication cost at internal nodes $nd(8,2,8)$ and $nd(8,6,8)$. But the twiddle factor, $nd(8,2,8).\omega^2_{16}$ can be factored out and pushed down to the children nodes $nd(4,2,4)$ and $nd(4,2,12)$. Likewise, an $\omega^2_{16}$ must also be factored out of $nd(8,6,8)$ to maintain algebraic correctness. After factoring out the $\omega^2_{16}$, all multiplication cost occurs on the bottom row of $bf(4,2,0,4)$ and the total cost remains 24 FLOPs. Globally, there is now no size-$4$ original butterfly with cost in the first two rows.

Because of Theorem \ref{thm:topbutterflies} and the recursive structure of the FFT, we can now partition the FFT flowgraph when solving for minimum total arithmetic complexity. In general, we must solve for all FFTs corresponding to $f$ mod $x^{*}-i$ and $x^{*}+i$ branches in the factor tree. For a size-$n$ FFT, this requires solving SMT models for pairs of size-$p$ butterflies, for all $p$ from 1 up to $\nicefrac{n}{4}$. In practice, for values of $p=8$ the problem becomes trivial and is used as the terminal case of partitioning.  The most difficult partition of a size-$n$ FFT flowgraph, a size-$\frac{n}{4}$ butterfly, will have a solution space of $2^{\frac{n}{8} \log_2 \frac{n}{4} ((\log_2 n)-2)}$. In more concrete terms, the largest SMT models required to solve a size-256 flowgraph are for two size-64 butterflies corresponding to the $f$ mod $x^{64}-i$ and $x^{64}+i$ branches of the factor tree. One of these size-64 butterflies has a solution space of $2^{1152}$. 

Computation time for proving that a partitioned size-64 FFT flowgraph has no solution with total cost equal to or less than 1159 FLOPs is now 2.8 seconds. Our timeout of 24 hours is reached when attempting to prove that a size-128 FFT flowgraph has no solution with total cost equal to or less than 2824 FLOPs.  

\subsubsection{Partitioning Using Terminal Butterflies}

At the bottom of a flowgraph, the weight on base required for each final result $X(k)$ is known. This enables analysis of {\bf terminal butterflies}, or butterflies producing some final values of $X(k)$, so that the model may be further partitioned.

\begin{theorem}
 For any arbitrary size-$4$ terminal butterfly, $bf_1$, there exists another size-$4$ butterfly, $bf_2$, which has zero-cost twiddle factors for nodes in rows 1 and 2, such that all realizable final weighted sums $X(k)$ of $bf_1$ can be realized by $bf_2$. Furthermore, no $bf_1$ exists with lower FLOP count than this $bf_2$.
\label{thm:bottombutterflies}
\end{theorem}
\begin{proof}
The proof is in two parts. First, to prove all realizable final weighted sums of $bf_1$ can be achieved by $bf_2$, we substitute 0 for all final weights ($X(k).W_b = 0$) and for all twiddle factors in rows 1 and 2 ($nd_{1,*}.tfp=nd_{2,*}.tfp=0$) into the expressions from Equation \ref{eq:arbitrarysize4butterfly}: 
\begin{align*}
 X(0).W_b = 0 &\equiv 0 + 0 + nd_{0,0}.tfp + nd_{0,0}.W_b \\
 X(2).W_b = 0 &\equiv 0 + 0 + nd_{0,0}.tfp + nd_{0,0}.W_b \\
 X(1).W_b = 0 &\equiv 0 + 0 + nd_{0,0}.tfp + nd_{0,0}.W_b \\
 X(4).W_b = 0 &\equiv 0 + 0 + nd_{0,0}.tfp + nd_{0,0}.W_b \displaybreak[0]\\
 X(0).W_b = 0 &\equiv 0 + 0 + nd_{0,2}.tfp + nd_{0,2}.W_b - nd_{1,*}.W_s \\
 X(2).W_b = 0 &\equiv 0 + 0 + nd_{0,2}.tfp + nd_{0,2}.W_b - nd_{1,*}.W_s \\
 X(1).W_b = 0 &\equiv 0 + 0 + nd_{0,2}.tfp + nd_{0,2}.W_b - nd_{1,*}.W_s \\
 X(4).W_b = 0 &\equiv 0 + 0 + nd_{0,2}.tfp + nd_{0,2}.W_b - nd_{1,*}.W_s \displaybreak[0]\\
 X(0).W_b = 0 &\equiv 0 + 0 + nd_{0,1}.tfp + nd_{0,1}.W_b - nd_{2,*}.W_s \\
 X(2).W_b = 0 &\equiv 0 + 0 + nd_{0,1}.tfp + nd_{0,1}.W_b - nd_{2,*}.W_s \\
 X(1).W_b = 0 &\equiv 0 + 0 + nd_{0,1}.tfp + nd_{0,1}.W_b - nd_{2,*}.W_s \\
 X(4).W_b = 0 &\equiv 0 + 0 + nd_{0,1}.tfp + nd_{0,1}.W_b - nd_{2,*}.W_s \displaybreak[0]\\
 X(0).W_b = 0 &\equiv 0 + 0 + nd_{0,3}.tfp + nd_{0,3}.W_b - nd_{2,*}.W_s - nd_{1,*}.W_s \\
 X(2).W_b = 0 &\equiv 0 + 0 + nd_{0,3}.tfp + nd_{0,3}.W_b - nd_{2,*}.W_s - nd_{1,*}.W_s \\
 X(1).W_b = 0 &\equiv 0 + 0 + nd_{0,3}.tfp + nd_{0,3}.W_b - nd_{2,*}.W_s - nd_{1,*}.W_s \\
 X(4).W_b = 0 &\equiv 0 + 0 + nd_{0,3}.tfp + nd_{0,3}.W_b - nd_{2,*}.W_s - nd_{1,*}.W_s 
\end{align*}
Rows have been reordered to group common twiddle factors. By direct inspection we establish that a final weight on base of 0 for all $X(k)$ can be realized by twiddle factors of nodes only in the first row. 

Second, exhaustive search with SMT is used to prove that no $bf_1$ exists with lower FLOP count that its $bf_2$. The SMT proof is a miter that includes $bf_1$ and $bf_2$. The $bf_2$ side of the miter is a direct translation to SMT of the final weight on base computations just seen. The $bf_1$ side of the miter is created by substituting 0 for all final weights ($X(k).W_b = 0$) in the expressions from Equation \ref{eq:arbitrarysize4butterfly}. Input values $nd_{0,*}.W_b$ and row weight strides, $nd_{1,*}.W_s,nd_{2,*}.W_s$, are common to $bf_1$ and $bf_2$. FLOP counts for $bf_1$ and $bf_2$ are individually and explicitly tallied within the SMT model. The question posed to the SMT solver is to find a $bf_1$ with lower FLOP count than $bf_2$. The theorem is proved for some $n$ if the SMT solver returns unsatisfiable. The proof can be run once for every size-$n$ FFT under consideration, or induction can be used to establish the result for $n+1$ and higher. 
\end{proof}
This theorem appears to conflict with decimation-in-frequency FFT algorithms, such as shown in Figure \ref{fig:radix2fft}, where costly twiddle factors appear in the last two rows of the FFT flowgraph. Consider the size-$4$ butterfly $bf(1,0,1,4)$ from Figure \ref{fig:radix2fft}. There is multiplication cost at internal nodes $nd(2,1,2)$ and $nd(2,1,6)$. But the twiddle factor, $nd(2,1,2).\omega^1_{8}$ can be distributed and pushed up to the parent nodes $nd(4,1,4)$ and $nd(4,3,4)$. Likewise, an $\omega^1_{8}$ must also be factored out of $nd(2,1,6)$ to maintain algebraic correctness. Now all multiplication cost occurs in the top row of $bf(1,0,1,4)$ and the total cost remains the same. Globally, there is now no size-$4$ terminal butterfly with cost in the last two rows. Note that for this small size-$8$ FFT, this new configuration of twiddle factors now fails conditions for partitioning by original butterflies as costly twiddle factors now occur in the top two rows. For this reason, combined original and terminal partitioning is only applicable to size-$16$ and larger FFT flowgraphs.

By Theorem \ref{thm:bottombutterflies} and the recursive structure of the FFT, we can now further partition the FFT flowgraph when solving for minimum FLOP count. In general, we must solve for all FFTs corresponding to $f$ mod $x^{*}-i$ and $x^{*}+i$ branches in the factor tree, but now each branch can be partitioned into four smaller equally sized FFTs. For a size-$n$ FFT, this requires solving SMT models for groups of 8 size-$p$ butterflies for all $p$ from 1 up to $\frac{n}{16}$. In practice, for values of $p=8$ the problem becomes trivial and that is used as the terminal case of partitioning.  The most difficult partition of a size-$n$ FFT flowgraph, a size-$\frac{n}{16}$ butterfly, will have a solution space of $2^{\frac{n}{32} \log_2 \frac{n}{16} ((\log_2 n)-2)}$. In concrete terms, the largest SMT models required to solve a size-256 flowgraph are eight size-16 butterflies corresponding to the $f$ mod $x^{64}-i$ and $x^{64}+i$ branches of the factor tree. One of these size-16 butterflies has a solution space of $2^{192}$. 

We can now prove the surprising result that size-256 FFTs exists which require only 6616 FLOPs, rather than the 6664 FLOPs required by the traditional split-radix, even when twiddle factors are of modulus one. Finding a 6616 FLOP algorithm requires 22 seconds to compute when the lowest cost constraint is used for each partition. Just over 5 seconds is required for the toughest size-$16$ partition. Of course, searching for the lowest cost in a partition requires repeated SMT runs and consequently the total search time is higher. To prove that no solution exists with FLOP count lower than 6616 requires 160 seconds total, with the toughest partition requiring just over 50 seconds.  

\subsection{Symmetry Reductions}

We find that there are many FFTs with equivalent final FLOP count yet with different twiddle factor values. Prior work in twisting\cite{bernstein2007tangent}\cite{2008-mateer} indicates that this should be expected. In this section, we highlight two types of symmetry reduction that reduce SMT run times. Many local symmetry reduction constraints are possible and we experimented with dozens but found only these two to be of any significance. 

\subsubsection{3-Node Symmetries}

A size-2 butterfly is {\bf3-node symmetric} if 3 of its 4 nodes require 6 FLOPs for multiplication. Symmetries are eliminated by forcing  $nd_{1,0}.tfp$ to have no multiplication cost.

\begin{example}
Consider a concrete computation performed by a size-$2$ butterfly from the size-$32$ FFT flowgraph expressed as
\begin{align*}
 X(0) &\equiv \psi_{32}^0(a_0\psi_{32}^1 + a_1\psi_{32}^3) \mdf\\
 X(1) &\equiv \psi_{32}^7(a_0\psi_{32}^1 + a_1\psi_{32}^3) \mdf.
\end{align*}
This butterfly requires $18=6+6+6$ FLOPs for multiplication. If the $\psi_{32}^3$ is factored out to ``zero'' the weight on $a_1$ and shared twiddle factors are preserved, these equations can be expressed as
\begin{align*}
 X(0) &\equiv \psi_{32}^3(a_0\psi_{32}^6 + a_1\psi_{32}^0) \mdf \\
 X(1) &\equiv \psi_{32}^2(a_0\psi_{32}^6 + a_1\psi_{32}^0) \mdf,
\end{align*}
with total multiplication cost of $18=6+6+6$ FLOPs again. Alternatively, if the $\psi_{32}^1$ is factored out to ``zero'' the weight on $a_0$, these equations can be expressed as 
\begin{align*}
 X(0) &\equiv \psi_{32}^1(a_0\psi_{32}^0 + a_1\psi_{32}^2) \mdf \\
 X(1) &\equiv \psi_{32}^0(a_0\psi_{32}^0 + a_1\psi_{32}^2) \mdf,
\end{align*}
with total multiplication cost of $12=6+6$ FLOPs. For the values in this example we find a cost benefit from factoring out the $\psi_{32}^1$. We must be pessimistic and assume the worse case, $18=6+6+6$ FLOPs, since only one weight is guaranteed zero-cost. It is also possible to ``zero'' the weight on the $X(1)$ sum by distributing the $\psi_{32}^7$ and achieve the same 12 FLOP configuration.
\label{exmp:3-node symmetric} 
\end{example}

In the SMT model, 3-node symmetric size-2 butterflies are detected and only those with zero multiplication cost for $nd_{1,0}$ are allowed. This is built by defining the following illegal condition,
\begin{equation*}
 nd_{1,0}.c6 \land ((nd_{0,0}.c6 \land nd_{0,1}.c6) \lor (nd_{0,0}.c6 \land nd_{1,1}.c6) \lor (nd_{0,1}.c6 \land nd_{1,1}.c6)), 
\end{equation*}
for each size-2 butterfly and then requiring the inverse be satisfied in the SMT model. 

We have verified with SMT-based proofs like those seen previously that this constraint doesn't increase the butterfly's FLOP count. As in the example, it may lead to a lower FLOP count if some node other than $nd_{1,0}$ has an applied weight of zero. We have formulated more complex constraints to detect these better cases early but found negligible speed-up in SMT runs. Instead, we rely on the cost-constraint described in Section \ref{sec:smtformulation} to eventually eliminate bad choices. Finally, if the SMT solver happens to choose the better placement of zero applied weight to begin with, the node is not 3-node symmetric (multiplication cost is less than 16 FLOPs) and no 3-node symmetric constraint will apply.
  
\subsubsection{Bottom Equal-Pair Symmetries}

A size-2 butterfly has {\bf equal-pair symmetries} if nodes $nd_{1,0}$ and $nd_{1,1}$ have multiplication cost and equal twiddle factors, $nd_{1,0}.tfp=nd_{1,1}.tfp$. This symmetry is eliminated by requiring that these identical twiddle factors in row 1 be distributed to row 0 nodes of the butterfly.

\begin{example}
Consider a concrete computation performed by a size-$2$ butterfly from the size-$32$ FFT flowgraph expressed as
\begin{align*}
 X(0) &\equiv \psi_{32}^3(a_0\psi_{32}^0 + a_1\psi_{32}^0) \mdf \\
 X(1) &\equiv \psi_{32}^3(a_0\psi_{32}^0 + a_1\psi_{32}^0) \mdf.
\end{align*}
This butterfly requires $12=6+6$ FLOPs for multiplication. If the $\psi_{32}^3$ is distributed, these equations can be expressed as
\begin{align*}
 X(0) &\equiv \psi_{32}^0(a_0\psi_{32}^3 + a_1\psi_{32}^3) \mdf \\
 X(1) &\equiv \psi_{32}^0(a_0\psi_{32}^3 + a_1\psi_{32}^3) \mdf,
\end{align*}
with total multiplication cost of $12=6+6$ FLOPs again. Another example with initial multiplication cost in row 0 is
\begin{align*}
 X(0) &\equiv \psi_{32}^3(a_0\psi_{32}^2 + a_1\psi_{32}^0) \mdf \\
 X(1) &\equiv \psi_{32}^3(a_0\psi_{32}^2 + a_1\psi_{32}^0) \mdf,
\end{align*}
with total multiplication cost of $18=6+6+6$ FLOPs. After distributing the $\psi_{32}^3$, this becomes
\begin{align*}
 X(0) &\equiv \psi_{32}^0(a_0\psi_{32}^5 + a_1\psi_{32}^3) \\
 X(1) &\equiv \psi_{32}^0(a_0\psi_{32}^5 + a_1\psi_{32}^3),
\end{align*}
with lower multiplication cost of $12=6+6$ FLOPs. For the values in this example we find a benefit but note that the final FLOP count is never worse than the initial as proved with SMT.
\label{exmp:bottom equal-pair symmetric} 
\end{example}

In the SMT model, bottom equal-pair symmetric butterflies are not allowed. This is built by defining the following illegal condition,
\begin{equation*}
 (\lnot nd_{1,0}.c0) \land (nd_{1,0}.tfp = nd_{1,1}.tfp),
\end{equation*}
for each size-2 butterfly and then requiring the inverse be satisfied in the SMT model. 

We have verified with SMT-based proofs like those seen previously that this symmetry reduction doesn't increase the butterfly's FLOP count. A similar constraint for top equal-pair symmetric butterflies can be formulated, and even applied in combination with the bottom equal-pair symmetric constraint with care, but we found negligible speed-up in SMT runs when doing so. 

These two symmetry reduction constraints now bring the total time for finding a 6616 FLOP count solution for a size-256 FFT down to 8 seconds. To prove that no solution exists with less than 6616 FLOPs now requires 50 seconds. It is now possible to find a 15128 FLOP count solution for a size-512 FFT in about 11 hours. We gave up on attempts to find solutions better than 15128 FLOPs after spending more than 14 days. There were four partitions for which we could not prove unsatisfiable when applying a FLOP count constraint of the ``best found less one.'' From experience, we suspect that a 15127 FLOP solution is most likely unsatisfiable given the dramatic increase in SMT solver run times.  

\section{Results and Experiments}
\label{sec:results}

Table \ref{tab:results} summarizes our results for SMT-based search of various size FFT flowgraphs. For size-$256$ FFTs and larger, we see that algorithms with FLOP count lower than the traditional split-radix do exist even when all twiddle factors have modulus one. We also show FLOP counts for the traditional spit-radix and for the tangent FFT\cite{johnson2007modified}\cite{bernstein2007tangent}, where twiddle factors are scaled and hence not modulus one. As expected, the required SMT time quickly becomes intractable as larger FFTs are considered. Yet it is still instructive to consider FFTs of relatively small size as such FFTs appear in larger FFTs. Finally, we do not know the number of FFT algorithms meeting these minimum FLOP count constraints but do know that there are many. We did search for multiple solutions of a size-$256$ FFT flowgraph partition and found hundreds before terminating. These solutions have both different values and placement patterns for costly twiddle factors.

\begin{table}[h]
\centering
\begin{tabular}{| c | c | c | c | l | c | l |}
\cline{2-7}
\multicolumn{1}{c|}{}&
{\bf Tangent} &
{\bf Split-Radix} &
\multicolumn{4}{c|}{\bf SMT Search}\\

\multicolumn{1}{c|}{}&
{\footnotesize$|\omega_n^*|=*$} &
{\footnotesize$|\omega_n^*|=1$} &
\multicolumn{4}{c|}{\footnotesize$|\omega_n^*|=1$}\\

\cline{4-7}
\multicolumn{1}{c|}{}&
&
&
\multicolumn{2}{c|}{Satisfiable} &
\multicolumn{2}{c|}{Unsatisfiable}\\

\hline
{\bf FFT Size} &
{\small FLOPs} &
{\small FLOPs} &
{\small FLOPs} &
{\small time({\it s})} &
{\small FLOPs} &
{\small time({\it s})} \\

\hline

{\bf 32} & 456 & 456 & 456 & $1.4 \times 10^{-1}$ & 455 & $1.5 \times 10^{-1}$ \\
\hline
{\bf 64} & 1152 & 1160 & 1160 & $3.1 \times 10^{-1}$ & 1159 & $3.3 \times 10^{-1}$ \\
\hline
{\bf 128} & 2792 & 2824 & 2824 & $9.3 \times 10^{-1}$ & 2823 & $1.1 \times 10^{0}$ \\
\hline
{\bf 256} & 6552 & 6664 & 6616 & $8.3 \times 10^0$ & 6615 & $5.0 \times 10^1$ \\
\hline
{\bf 512} & 15048 & 15368 & 15128 & $3.9 \times 10^4$ & 15127? & >$1 \times 10^{6}$ \\
\hline
\end{tabular}

\caption{Lowest FLOP Counts Found by SMT Search}
\label{tab:results}

\end{table}

The times reported in Table \ref{tab:results} are for the FLOP bounds at the boundary between satisfiable and unsatisfiable. We search for this boundary using binary search akin to Newton's method. We start with the best known FLOP bound for that size and class of FFT found in the literature, and divide that by 2. If that is satisfiable, we consider that the best known FLOP count and repeat. But if it is unsatisfiable, we choose a new FLOP bound half way between the unsatisfiable FLOP bound and the last known satisfiable bound and repeat. A complete search does require more time than seen in Table \ref{tab:results}, but we find that FLOP counts far away from the boundary are solved relatively fast, whether they are satisfiable or unsatisfiable. Only when the boundary is approached do times increase dramatically. Furthermore, by imposing a timeout, we can skew the search to approach the boundary from the satisfiable side, where FLOP counts are successively becoming lower. This improves overall search performance as proving satisfiable cases is generally less costly than proving unsatisfiable cases.

The reduction in FLOP count of FFTs found by SMT search appears to accelerate for larger $n$ when compared to the tangent FFT. Our size-$256$ solution has an advantage of 48 FLOPs when compared to the traditional split-radix FFT, whereas the tangent FFT has an advantage of 112 FLOPs. At this size, our FFT provides $\nicefrac{48}{112}=0.429$ of the advantage of the tangent FFT. At size-$512$, this advantage is $\nicefrac{240}{320}=0.75$. It is unclear if this approaches the tangent FFT advantage asymptotically, eventually surpasses it, or degrades. We suspect that the opportunities for optimization may be increasingly richer as partition sizes and the number of costly twiddle factors that they contain grow.

\subsection{SMT QF\_BV Solver Experiments}

The results reported so far have all been generated using the SMT solver Boolector\cite{brummayer2009boolector}. In this section, we present results for various SMT solvers, and identify some SMT solver characteristics best suited for our problem.

We use four representative benchmarks for our experiments. The first, \texttt{Sz256\_6616}, is the hardest partition from a size-$256$ flowgraph with a 6616 FLOP bound and is known to be satisfiable. The second, \texttt{Sz256\_6615}, is also the hardest partition from a size-$256$ flowgraph but with a 6615 FLOP bound and is known to be unsatisfiable. Likewise, the third and fourth benchmarks, \texttt{Sz512\_15128} and \texttt{Sz512\_15127}, are the hardest partitions from a size-$512$ flowgraph with 6616 and 6615 FLOP bounds respectively. Only \texttt{Sz512\_15128} is known to be satisfiable. Whether benchmark \texttt{Sz512\_15127} is satisfiable is unknown, but we suspect it is unsatisfiable.

For state-of-the-art SMT solvers, we use the top four SMT solvers in the QF\_BV category, closed quantifier-free formulas over the theory of fixed-size bitvectors, from the SMT-2011 competition\cite{SMTCompetitions}: Z3\cite{Z3}, STP2\cite{STP}, Boolector\cite{brummayer2009boolector} as well as MathSat5\cite{mathsat4} main and application configurations. We include two additional QF\_BV solvers that performed well in earlier competitions: Beaver\cite{jha2009beaver} and Yices\cite{yices}. For the SMT-2011 competition solvers, we used the binary executables and unmodified run scripts from the SMT-2011 competition web site\cite{SMTCompetitions}. For Beaver and Yices, we downloaded the latest available version from the web: Beaver 1.2.0.780 and Yices 2.0, build date of July 29, 2010, for x86\_64-unknown-linux-gnu. Both Beaver and Yices were executed without any additional command line options. We updated our pretty printer to support SMT-LIB 2.0\cite{ranise2006smt} for the four SMT-2011 competition solvers. Beaver and Yices were given SMT-LIB 1.2 input. All experiments were run on a 64-bit Intel Core i7 Linux machine.

Results for our SMT solver experiments are shown in Table \ref{tab:smtsolvers}. We have ordered the results from best to worst performance on benchmark \texttt{Sz256\_6615}, which we consider the most representative as all lowest FLOP searches must end with a proven unsatisfiable case. For this unsatisfiable case, all solvers perform in the same order of magnitude, with the worst performer requiring $2.5 \times$ the amount of time as the best performer. For the satisfiable cases, we see a larger variation in performance due to the rich set of solutions that exist and the chances that a particular solver's search strategy will find one first. All SMT solvers reached the timeout of 24 hours without solving \texttt{Sz512\_15127}.

\begin{table}[h]
\centering
\begin{tabular}{| l | c | c | c | c |}
\cline{2-5}

\multicolumn{1}{c|}{}&
{\bf Sz256\_6616}&
{\bf Sz256\_6615}&
{\bf Sz512\_15128}&
{\bf Sz512\_15127}\\

\cline{1-1}
{\bf Solver}&
{SAT \small time({\it s})}&
{UNSAT \small time({\it s})}&
{SAT \small time({\it s})}&
{Unknown}\\

\hline{Beaver}        & $2.0 \times 10^{0}$ & $7.6 \times 10^{0}$ & $6.2 \times 10^{2}$ & Timeout \\
\hline{STP2}          & $0.5 \times 10^{0}$ & $1.2 \times 10^{1}$ & $8.4 \times 10^{3}$ & Timeout \\
\hline{Boolector}     & $3.4 \times 10^{0}$ & $1.2 \times 10^{1}$ & $2.9 \times 10^{4}$ & Timeout \\
\hline{MathSAT5 app}  & $4.0 \times 10^{0}$ & $1.5 \times 10^{1}$ & $5.7 \times 10^{4}$ & Timeout \\
\hline{MathSAT5 main} & $5.2 \times 10^{0}$ & $1.8 \times 10^{1}$ & $1.1 \times 10^{4}$ & Timeout \\
\hline{Z3}            & $2.8 \times 10^{0}$ & $1.9 \times 10^{1}$ & $2.1 \times 10^{4}$ & Timeout \\
\hline{Yices}         & $4.3 \times 10^{0}$ & $1.9 \times 10^{1}$ & $6.8 \times 10^{4}$ & Timeout \\

\hline
\end{tabular}

\caption{SMT Solver Performance}
\label{tab:smtsolvers}

\end{table}

For Beaver, the best performing SMT solver on \texttt{Sz256\_6615}, we also varied the command line options to test their effect. The most noticeable differences, although minor, came from disabling optimizations. Table \ref{tab:beaver} summarizes our findings. Constant propagation appears to be the most effective for our application. It should prove beneficial to incorporate constant propagation at the high-level when we generate our initial SMT models.

\begin{table}[h]
\centering
\begin{tabular}{| l | c |}
\cline{2-2}

\multicolumn{1}{c|}{}&
{\bf Sz256\_6615}\\

\cline{1-1}
{\bf Solver}&
{UNSAT \small time({\it s})}\\

\hline{Beaver --disable-const} & 9.9 \\
\hline{Beaver (disable all)} & 9.5 \\
\hline{Beaver --disable-commute} & 8.1 \\
\hline{Beaver --disable-assoc} & 7.8 \\
\hline{Beaver --disable-non-linear} & 7.7 \\
\hline{Beaver (enable all)}& 7.6 \\

\hline
\end{tabular}

\caption{Beaver Optimization Options}
\label{tab:beaver}

\end{table}

\subsection{Bit-Blasting and SAT Solver Experiments}

A common trait of the three best SMT solvers for our problem, Beaver, STP2 and Boolector, is that they focus on bitvector problems. All three perform bit-blasting and then use a SAT solver back-end to solve a traditional SAT problem. Furthermore, they pay close attention to the circuit structure when optimizing and bit-blasting. All three incorporate AIGs, And-Invert Graphs\cite{kuehlmann2001circuit}, and rewriting of AIGs. Beaver and STP2 use the ABC\cite{abc} library to facilitate this. Furthermore, Beaver employs the SAT solver nflsat\cite{jain2009efficient}, which operates on AIGs natively, as its default back-end. Since this circuit-centric approach works well for our problems, this section presents experimental data on bit-blasting flows and SAT solver performance when considered separately.

For each of the three best SMT solvers for our problem we implemented four experimental bit-blasting flows. At a high-level, these four flows are:
\begin{enumerate}
 \item SMT solver circuit representation to CNF with ABC
 \item SMT solver circuit representation to CNF with ABC after ABC optimization for SAT
 \item SMT solver circuit representation to CNF with AIGER
 \item SMT solver circuit representation to CNF with AIGER after ABC optimization for SAT
\end{enumerate}

Since each SMT solver's native circuit representation is slightly different, we first standardized all circuit representations to AIGs. Beaver incorporates ABC and can generate AIGs natively. We generated an AIG using the command line options \texttt{beaver --no-solve --aig --aig-file=<file.aig> <file.smt>}. STP2 also incorporates ABC but has no working command line option to generate an AIG file. Since STP2 is distributed as source, we were able to add an option to generate an AIG output file of it's internal circuit representation. Boolector has an option to dump expressions in BTOR format. We generated BTOR with that option and converted it to AIG using \texttt{synthebtor -m <file.btor> <file.aig>}, which is a tool provided with the Boolector package. Thus, the bit-blasted representation from all three solvers are standardized as AIGs.

In flows 1 and 2, CNF is generated from AIG by ABC using \texttt{write\_cnf}. In flows 3 and 4, CNF is generated from AIG by using the tool \texttt{aigtocnf}, which is part of the AIGER package\cite{aiger}. In flows 2 and 4, the ABC optimization for SAT command \texttt{drwsat} is executed 3 times to generate a simplified AIG. 

We selected 7 SAT solvers that are readily available and either have performed well in recent SAT competitions\cite{SATCompetitions}\cite{SATRace} or are used in the back-end for Beaver, STP2 or Boolector already.
 
\begin{itemize}
 \item glueminisat 2.2.5\cite{glueminisat}
 \item simplifying minisat 2.2.0\cite{minisat}
 \item cryptominisat 2.9.0\cite{cryptominisat}
 \item precosat 570\cite{precosat}
 \item lingeling 276\cite{precosat}
 \item clasp 2.0.2\cite{clasp}
 \item nflsat 05102009\cite{jain2009efficient}
\end{itemize}

\begin{figure}[t!]
\centering
\includegraphics{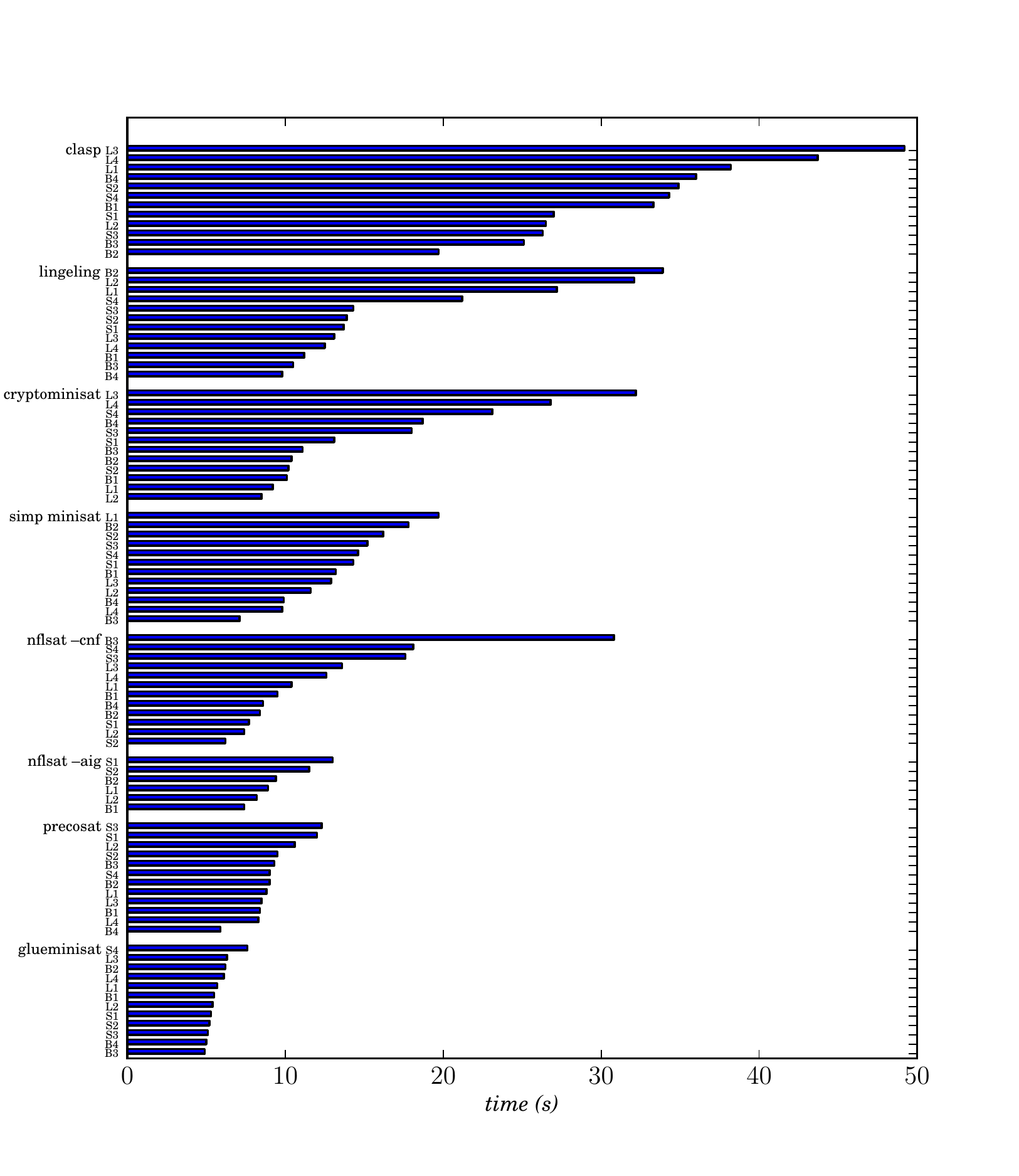}
\caption{SAT Solver Performance for Various Bit-Blasting Flows}
\label{fig:satsolverperformance}
\end{figure}

Figure \ref{fig:satsolverperformance} shows the results for each of the seven SAT solvers with all 12 bit-blasting flows. The bit-blasting flows are keyed to the SMT front-end, (B=Beaver, L=Boolector, S=STP2), along with the route to CNF, 1-4. Results are grouped by SAT solver, and ordered from worst (top) to best (bottom) average performance. Within a SAT solver grouping, results are ordered again from worst to best performance. The reported times include format conversion and ABC optimization, if applicable. 

From these results, we see that the SAT solver glueminisat, the winner in the SAT 2011 competition for the UNSAT application track, consistently performs the best. Overall, the back-end SAT solver choice is more significant than the SMT bit-blasting tool and/or path to CNF. Still, Beaver bit-blasting appears to provide a slight second-order advantage.  It is also interesting to note that the SAT solver clasp, the winner in the SAT 2011 competition for the UNSAT crafted track, performed the worst. Furthermore, Beaver's default choice of a back-end SAT solver, nflsat with native AIG input, does not distinguish itself. 

For our problem, the data suggests that bit-blasting by Beaver with conversion to CNF by AIGER (flow B3) for input to the SAT solver glueminisat is the best choice. We applied this flow to the unknown problem \texttt{Sz512\_15127} but were still unable to produce a result after days of compute time. From this we conclude that the greatest advances in solving our particular problem will come from high-level insight, such as additional problem partitioning and identification of new problem-specific constraints. Next, the choice of the underlying SAT technology will have some beneficial effect as SAT technology continues to improve. And finally, how we cast our problem as SAT, including initial specification, constant propagation, choice of bit-blasting and conversion to CNF, can be adjusted for further second-order improvements.

\subsection{SMT QF\_LIA Solver Experiments}

\label{sec:liasolver}

If is unclear whether modeling our problem as QF\_BV is the best choice. It is possible to model bitvector problems with other logics\cite{somenzi2011selective}\cite{zeng2001lpsat}\cite{brinkmann2002rtl}\cite{bozzano2006encoding}. In this section, we present a first attempt to model our problem as QF\_LIA, following the techniques of Kim and Somenzi\cite{somenzi2011selective}.

In their recent paper\cite{somenzi2011selective}, Kim and Somenzi showed that some QF\_BV problems could be cast as QF\_LIA for improved SMT solver performance. The main idea of their casting is to detect overflow and underflow of integer operations and use ITE operators to enforce the modular arithmetic of QF\_BV within QF\_LIA. We have implemented this casting by adding underflow/overflow detection and correction for all (mod $n$) operations in Algorithm \ref{algo:random}. Table \ref{tab:smtliasolvers} summarizes our results for all QF\_LIA solvers we could find that accept SMT-LIB 2.0. The benchmarks \texttt{Sz64\_1160}, \texttt{Sz64\_1159}, \texttt{Sz128\_2824}, \texttt{Sz128\_2823}, are of similar character to the set used in QF\_BV experiments shown in Table \ref{tab:smtsolvers} except that they are from considerably smaller (size-$64$ and size-$128$) FFTs. When attempting the same benchmarks as used in the QF\_BV experiments, the timeout of 24 hours was reached in all cases. Clearly, we must improve out initial QF\_LIA specification and/or the underlying QF\_LIA solver technology to make QF\_LIA solvers competitive with QF\_BV solvers on our problem.

\begin{table}[h]
\centering
\begin{tabular}{| l | c | c | c | c |}
\cline{2-5}

\multicolumn{1}{c|}{}&
{\bf Sz64\_1160}&
{\bf Sz64\_1159}&
{\bf Sz128\_2824}&
{\bf Sz128\_2823}\\

\cline{1-1}
{\bf Solver}&
{SAT \small time({\it s})}&
{UNSAT \small time({\it s})}&
{SAT \small time({\it s})}&
{UNSAT \small time({\it s})}\\

\hline{Z3}            & $ 3.8 \times 10^{-2}$ & $ 5.2 \times 10^{-2}$ & $ 1.0 \times 10^{0}$ & $ 2.0 \times 10^{3} $ \\
\hline{MathSAT5 app}  & $ 7.6 \times 10^{-2}$ & $ 1.9 \times 10^{-1}$ & $ 1.5 \times 10^{3}$ & Timeout \\
\hline{MathSAT5 main} & $ 5.5 \times 10^{-2}$ & $ 1.0 \times 10^{-1}$ & $ 1.0 \times 10^{2}$ & $ 1.4 \times 10^{3} $ \\

\hline
\end{tabular}

\caption{SMT QF\_LIA Solver Performance}
\label{tab:smtliasolvers}

\end{table}

\subsection{Algorithm Design}
\label{sec:algorithmdesign}

The FFTs found by SMT-based search and posted on our web site\cite{fftexamples} are witnesses that FFTs with lower total FLOP count than the split-radix exist even when all twiddle factors have modulus one, but are not practical algorithms in their current state. FFTs in widespread use usually can be defined succinctly in mathematical terms which leads to very regular patterns of twiddle factors in the FFT flowgraph. It is possible to formulate SMT constraints that require various forms of regularity in any satisfying solution. For example, additional constraints can be formulated and added to the model which allow costly twiddle factors only at specified nodes in the graph. A tighter constraint might force specific nodes to have prespecified twiddle factor values. A more relaxed constraint might just impose a relationship, such as a stride, between pairs of twiddle factors. In this way, the techniques described in this paper can be extended to do practical FFT algorithm design at the expense of proven optimality. Although this is a topic for further research, we highlight a few early experiments here.

The split-radix created by delayed twisting as described by Bernstein\cite{bernstein2007tangent} and Mateer\cite{2008-mateer} is very succinct yet can be used to generate a rich family of highly regular split-radix algorithms simply by choosing different legal twisting coefficients, $\zeta$. By examining the twiddle factor patterns generated by this algorithm, we determine that twiddle factors applied to ordered coefficients of a polynomial in the factor tree must have a constant stride (twisted), match constant values as seen in the classic decomposition (delayed twisting), or combine these two cases (twisting to something other than $x^*-1$). With constraints formulated and applied to the SMT model that require this pattern of twiddle factors, we no longer find solutions with total FLOP count less than the split-radix for size-256 FFT flowgraphs. We do find solutions with FLOP count equal to the split-radix as expected. This confirms the theorem by Mateer\cite{2008-mateer} that combinations of twisting, though very rich, will never lead to an FFT with FLOP count lower than the split-radix. Although the regularity imposed by twisting doesn't support our solutions, other types of regularity might.

The tangent FFT\cite{johnson2007modified}\cite{bernstein2007tangent} starts with a version of the conjugate split-radix FFT\cite{kamarconjugate}. In this algorithm, twiddle factors occur as conjugate pairs, where the conjugate pair is either at the top or bottom of a size-$2$ butterfly. The complex twiddle factors for a conjugate pair can be factored as
\begin{equation*}
 \cos{\alpha} (1 + i\tan{\alpha}), \cos{\gamma} (1 + i\tan{\gamma}).
\end{equation*}
Since $\alpha$ and $\gamma$ are conjugate angles, we know that $\cos{\alpha} = \cos{\gamma}$. Van Buskirk's trick\cite{lundy2007new}, which is exploited in the tangent FFT, moves these real scaling factors so that their cost is absorbed by other multiplications. With constraints formulated and applied to the SMT model that require twiddle factors to occur globally as conjugate pairs, we no longer find solutions with total FLOP count less than the split-radix for size-256 FFT flowgraphs. We do still find instances with FLOP count equal to the split-radix. It still may be possible to find solutions where conjugate pairs occur locally in specific places such that optimizations similar to Van Buskirk's can be beneficially applied.

An objective to minimize FLOP count is primarily academic given the capabilities of modern computing hardware. Other more practical objectives include enhancing precision or easing implementation. For example, avoiding twiddle factors where the real or imaginary part is a number very close to zero may enhance the precision of the final result. Alternatively, restricting all twiddle factors to some limited set may ease implementation, and we can formulate a SMT model that does just that. There are size-$32$ FFTs that use just two non-trivial costly twiddle factors, plus the free multiplications by $1$, $-1$, $i$ or $-i$. The minimum FLOP count for these algorithms is high at 616 compared to 456 for the split-radix but there may be benefits of having to multiply by just a few constants, especially in hardware implementations. If we increase the set of allowed non-trivial twiddle factors for a size-$32$ FFT to three, the minimum FLOP count is 536. For a size-$64$ FFT, we find a 2112 FLOP count solution that uses only non-trivial twiddle factor powers from the set $\{7,8,9\}$. Note that these twiddle factor powers include conjugates so that only three transcendental function computations or table look-ups are required. We have posted some examples of these FFTs on our web site\cite{fftexamples}.

\section{Conclusions and Future Work}
\label{sec:conclusions}

This paper presented a Boolean Satisfiability-based proof of the lowest FLOP count required by FFT algorithms up to size-$512$ with flowgraphs isomorphic to those generated by common power-of-two FFTs, and where all twiddle factors are $n^{th}$ roots of unity. Even with these constraints, we find FFTs requiring fewer FLOPs than the split-radix starting at size-$256$. At the core of this proof is a novel way to enumerate all FFTs realizable by a given flowgraph. Partitioning and symmetry reduction techniques are developed to make it possible to prove FLOP count bounds for larger size-$512$ FFTs. Finally, because the SAT-based formulation and search techniques are general, the paper introduced additional search objectives that mimic twiddle factor patterns from twisting, require conjugate twiddle factor pairs, and minimize the allowed values of twiddle factors.  

As seen from our experimental results, our biggest advances came from applying a high-level understanding of this problem to partition and detect symmetries in order to simplify the input for SMT and SAT solvers. We believe that more effort along these lines is a good direction for future work. In particular, work in symmetry detection and breaking to simplify SAT\cite{DBLP:series/faia/Sakallah09}\cite{katebi2010symmetry} is of interest. Just as this work uses computer automation to search for graph isomorphisms in the CNF structure, we can do the same at the more abstract FFT flowgraph level. Although symmetry breaking at the CNF level can benefit our problem, we believe that more progress can be made by exploiting higher-level symmetries in our specific problem. The challenge for us is to find useful isomorphisms with regard to twiddle factors, as the FFT flowgraph is very regular and rich in self-similarity. All our effort to partition and detect symmetry has been through human observation, and assistance from computer search may lead to better techniques.  

We have cast finding FFT algorithms as a bitvector problem and have used SAT and SMT solvers in a stand-alone manner to find solutions. This raises two questions for future work. First, is QF\_BV the best logic for this problem? Although we present preliminary results when cast as QF\_LIA in Section \ref{sec:liasolver}, there are still other casting techniques and logics to try\cite{zeng2001lpsat}\cite{brinkmann2002rtl}\cite{bozzano2006encoding}\cite{PBCompetitions}. Of particular interest to us is casting our problem to integer linear programming with the techniques presented by Brinkmann\cite{brinkmann2002rtl}. This may allow us to optimize larger problems, especially when optimality need not be proven. Second, can larger and more interesting instances of our problem be solved through tighter integration with SAT and/or SMT solvers? There are ideas for integrating optimization with SAT and SMT solvers\cite{nieuwenhuis2006sat}\cite{larrosa2011framework}\cite{cimatti2010satisfiability}. Solvers such as STP2\cite{STP} are providing APIs for tighter integration of user's applications. These directions remain unexplored by us but may yield significant improvements. 

Besides the future work just described, we plan additional work in three more directions. First, Section \ref{sec:algorithmdesign} highlights FFT algorithm design possible with techniques described in this paper. We will study the applicability of our techniques to practical FFT algorithm design, with cost objectives ranging from improved precision to implementation on specific hardware\cite{nordin2005automatic}. Second, we seek to impose regularity on our lowest FLOP count solutions to determine if they can be described more traditionally as succinct algorithms. This should also help us better characterize the FLOP savings as the the size of the FFT increases. Finally, we hope to ease the current constraint that all twiddle factors are $n^{th}$ roots of unity, and thus incorporate optimizations similar to those in Van Buskirk's\cite{lundy2007new} algorithm and the tangent FFT\cite{johnson2007modified}\cite{bernstein2007tangent} directly into our search.

\bibliographystyle{plainbv}
\bibliography{fft}

\end{document}

%% file: definitions.tex
\begin{tikzpicture}[>=latex]

\node(nd) [rectangle] {\large{$(a_1\omega_8^1 + a_3\omega_8^3 + a_5\omega_8^5 + a_7\omega_8^7)x$}};

\node(lta) [rectangle,above=6mm,left=1.7cm] at (nd.center) {};
\node(ltae) [rectangle,above=3mm,left=1.7cm] at (nd.center) {};
\draw[->] (lta.center) -- (ltae.center);

\node(rta) [rectangle,above=6mm,left=3.1mm] at (nd.center) {};
\node(rtae) [rectangle,above=3mm,left=3.1mm] at (nd.center) {};
\draw[->] (rta.center) -- (rtae.center);

\draw[-] (lta.center) -- (rta.center);

\node(wstride) [rectangle,anchor=south west,inner sep=-1mm] at (lta.north east) {\small{$W_{s}$=$2$}};

\node(lba) [rectangle,below=6mm,left=2.1cm] at (nd.center) {};
\node(lbae) [rectangle,below=3mm,left=2.1cm] at (nd.center) {};
\draw[->] (lba.center) -- (lbae.center);

\node(rba) [rectangle,below=6mm,left=8.0mm] at (nd.center) {};
\node(rbae) [rectangle,below=3mm,left=8.0mm] at (nd.center) {};
\draw[->] (rba.center) -- (rbae.center);

\draw[-] (lba.center) -- (rba.center);
\node(stride) [rectangle,anchor=north west,inner sep=-1mm] at (lba.south east) {\small{$stride$=$2$}};

\node(lbb) [rectangle,below=6mm,left=3cm] at (nd.center) {};
\node(lbbe) [rectangle,below=3mm,left=2.3cm] at (nd.center) {};
\draw[->] (lbb.center) -| (lbbe.center);
\node(base) [rectangle,anchor=east,inner sep=-1mm] at (lbb.west) {\small{$base$=$a_1$}};

\node(tbb) [rectangle,above=6mm,left=3cm] at (nd.center) {};
\node(tbbe) [rectangle,above=3mm,left=1.8cm] at (nd.center) {};
\draw[->] (tbb.center) -- (tbbe.center);
\node(wbase) [rectangle,anchor=east,inner sep=-1mm] at (tbb.west) {\small{$W_{b}$=$1$}};

\end{tikzpicture}

%% file: nodeinternals.tex
\begin{tikzpicture}[>=latex]

\node(nd) [rectangle split, rectangle split parts=3, rectangle split empty part height=1cm, rounded corners, text badly centered, minimum width=3.8cm, draw=black,inner sep=0.5mm]{ \nodepart{second} \nodepart{third} };

\node(add) [circle,draw,very thick,inner sep=2mm,above=-2mm] at (nd.center) {};
\draw[-,very thick] (add.west) -- (add.east);
\draw[-,very thick] (add.north) -- (add.south);

\node(addid) [rectangle,below=2.5mm] at (nd.center) {};
\node(addidlbl) [rectangle,right] at (addid) {\footnotesize\emph{id}};

\node(lmul) [circle,draw,very thick,above=3mm,inner sep=2mm] at (nd.255) {};
\draw[-,very thick] (lmul.north west) -- (lmul.south east);
\draw[-,very thick] (lmul.north east) -- (lmul.south west);

\node(rmul) [circle,draw,very thick,above=3mm,inner sep=2mm] at (nd.285) {};
\draw[-,very thick] (rmul.north west) -- (rmul.south east);
\draw[-,very thick] (rmul.north east) -- (rmul.south west);

\draw[->] (lmul.south) -- (nd.250);
\draw[->]  (rmul.south) -- (nd.290);

\node(ltf) [rectangle,left=2mm,anchor=east,inner sep=0mm] at (lmul.west) {$\omega^{ltfp}_n$};
\node(rtf) [rectangle,right=2mm,anchor=west,inner sep=0mm] at (rmul.east) {$\omega^{rtfp}_n$};

\draw[->] (ltf.east) -- (lmul.west);
\draw[->] (rtf.west) -- (rmul.east);

\draw[-] (add.south) -- (addid.center);
\draw[->] (addid.center) -- (lmul.north);
\draw[->] (addid.center) -- (rmul.north);

\draw[->] (nd.110) -- (add.110);
\draw[->] (nd.70) -- (add.70);

\node(linp) [rectangle,anchor=south,above=5mm] at (nd.north west) {};
\node(rinp) [rectangle,anchor=south,above=5mm] at (nd.north east) {};

\node(lout) [rectangle,anchor=north,below=5mm] at (nd.south west) {};
\node(rout) [rectangle,anchor=north,below=5mm] at (nd.south east) {};

\draw[->] (linp) -- (nd.110);
\draw[->] (rinp) -- (nd.70);

\draw[->] (nd.250) -- (lout);
\draw[->] (nd.290) -- (rout);

\node(inplbl) [rectangle,above=3mm] at (nd.north) {\footnotesize\sl{Input}};
\node(outlbl) [rectangle,below=3mm] at (nd.south) {\footnotesize\sl{Output}};

\end{tikzpicture}

%% file: classic_size8.tex
\begin{tikzpicture}[fftgraph]
	\matrix [fftmatrix] {
		& [-5mm]
		\node(a0) [fftterminal]{\small$a_{0}$}; &
		\node(a1) [fftterminal]{\small$a_{1}$}; &
		\node(a2) [fftterminal]{\small$a_{2}$}; &
		\node(a3) [fftterminal]{\small$a_{3}$}; &
		\node(a4) [fftterminal]{\small$a_{4}$}; &
		\node(a5) [fftterminal]{\small$a_{5}$}; &
		\node(a6) [fftterminal]{\small$a_{6}$}; &
		\node(a7) [fftterminal]{\small$a_{7}$}; &
		\\[-7mm]
		\node (stride8) [stride]{\normalsize\bf8}; & [-5mm]
		\node (n8_0_0) [fftnode,densely dotted]{\footnotesize\underline{0}\nodepart{second}\small\bf0.0\nodepart{third}\footnotesize\sl0  0}; &
		\node (n8_0_1) [fftnode,densely dotted]{\footnotesize\underline{0}\nodepart{second}\small\bf1.0\nodepart{third}\footnotesize\sl0  0}; &
		\node (n8_0_2) [fftnode,densely dotted]{\footnotesize\underline{0}\nodepart{second}\small\bf2.0\nodepart{third}\footnotesize\sl0  0}; &
		\node (n8_0_3) [fftnode,densely dotted]{\footnotesize\underline{0}\nodepart{second}\small\bf3.0\nodepart{third}\footnotesize\sl0  0}; &
		\node (n8_0_4) [fftnode,densely dotted]{\footnotesize\underline{0}\nodepart{second}\small\bf4.0\nodepart{third}\footnotesize\sl0  4}; &
		\node (n8_0_5) [fftnode,densely dotted]{\footnotesize\underline{0}\nodepart{second}\small\bf5.0\nodepart{third}\footnotesize\sl0  4}; &
		\node (n8_0_6) [fftnode,densely dotted]{\footnotesize\underline{0}\nodepart{second}\small\bf6.0\nodepart{third}\footnotesize\sl0  4}; &
		\node (n8_0_7) [fftnode,densely dotted]{\footnotesize\underline{0}\nodepart{second}\small\bf7.0\nodepart{third}\footnotesize\sl0  4}; &
		\\
		\node (stride4) [stride]{\normalsize\bf4}; & [-5mm]
		\node (n4_0_0) [fftnode]{\footnotesize\underline{0}  0\nodepart{second}\small\bf0.0\nodepart{third}\footnotesize\sl0  0}; &
		\node (n4_0_1) [fftnode]{\footnotesize\underline{0}  0\nodepart{second}\small\bf1.0\nodepart{third}\footnotesize\sl0  0}; &
		\node (n4_0_2) [fftnode]{\footnotesize\underline{0}  0\nodepart{second}\small\bf2.0\nodepart{third}\footnotesize\sl0  4}; &
		\node (n4_0_3) [fftnode]{\footnotesize\underline{0}  0\nodepart{second}\small\bf3.0\nodepart{third}\footnotesize\sl0  4}; &
		\node (n4_4_0) [fftnode]{\footnotesize\underline{0}  4\nodepart{second}\small\bf0.4\nodepart{third}\footnotesize\sl0  0}; &
		\node (n4_4_1) [fftnode]{\footnotesize\underline{0}  4\nodepart{second}\small\bf1.4\nodepart{third}\footnotesize\sl0  0}; &
		\node (n4_4_2) [fftnode]{\footnotesize\underline{0}  4\nodepart{second}\small\bf2.4\nodepart{third}\footnotesize\sl2  6}; &
		\node (n4_4_3) [fftnode]{\footnotesize\underline{0}  4\nodepart{second}\small\bf3.4\nodepart{third}\footnotesize\sl2  6}; &
		\\
		\node (stride2) [stride]{\normalsize\bf2}; & [-5mm]
		\node (n2_0_0) [fftnode]{\footnotesize\underline{0}  0\nodepart{second}\small\bf0.0\nodepart{third}\footnotesize\sl0  0}; &
		\node (n2_0_1) [fftnode]{\footnotesize\underline{0}  0\nodepart{second}\small\bf1.0\nodepart{third}\footnotesize\sl0  4}; &
		\node (n2_4_0) [fftnode]{\footnotesize\underline{0}  4\nodepart{second}\small\bf0.4\nodepart{third}\footnotesize\sl0  0}; &
		\node (n2_4_1) [fftnode]{\footnotesize\underline{0}  4\nodepart{second}\small\bf1.4\nodepart{third}\footnotesize\sl2  6}; &
		\node (n2_2_0) [fftnode]{\footnotesize\underline{0}  2\nodepart{second}\small\bf0.2\nodepart{third}\footnotesize\sl0  0}; &
		\node (n2_2_1) [fftnode]{\footnotesize\underline{0}  2\nodepart{second}\small\bf1.2\nodepart{third}\footnotesize\sl1  5}; &
		\node (n2_6_0) [fftnode]{\footnotesize\underline{0}  6\nodepart{second}\small\bf0.6\nodepart{third}\footnotesize\sl0  0}; &
		\node (n2_6_1) [fftnode]{\footnotesize\underline{0}  6\nodepart{second}\small\bf1.6\nodepart{third}\footnotesize\sl3  7}; &
		\\
		\node (stride1) [stride]{\normalsize\bf1}; & [-5mm]
		\node (n1_0_0) [fftnode]{\footnotesize\underline{0}  0\nodepart{second}\small\bf0.0\nodepart{third}\footnotesize\sl0}; &
		\node (n1_4_0) [fftnode]{\footnotesize\underline{0}  4\nodepart{second}\small\bf0.4\nodepart{third}\footnotesize\sl0}; &
		\node (n1_2_0) [fftnode]{\footnotesize\underline{0}  2\nodepart{second}\small\bf0.2\nodepart{third}\footnotesize\sl0}; &
		\node (n1_6_0) [fftnode]{\footnotesize\underline{0}  6\nodepart{second}\small\bf0.6\nodepart{third}\footnotesize\sl0}; &
		\node (n1_1_0) [fftnode]{\footnotesize\underline{0}  1\nodepart{second}\small\bf0.1\nodepart{third}\footnotesize\sl0}; &
		\node (n1_5_0) [fftnode]{\footnotesize\underline{0}  5\nodepart{second}\small\bf0.5\nodepart{third}\footnotesize\sl0}; &
		\node (n1_3_0) [fftnode]{\footnotesize\underline{0}  3\nodepart{second}\small\bf0.3\nodepart{third}\footnotesize\sl0}; &
		\node (n1_7_0) [fftnode]{\footnotesize\underline{0}  7\nodepart{second}\small\bf0.7\nodepart{third}\footnotesize\sl0}; &
		\\[-6mm]
		& [-5mm]
		\node(X0) [fftterminal]{\footnotesize$X(0)$}; &
		\node(X4) [fftterminal]{\footnotesize$X(4)$}; &
		\node(X2) [fftterminal]{\footnotesize$X(2)$}; &
		\node(X6) [fftterminal]{\footnotesize$X(6)$}; &
		\node(X1) [fftterminal]{\footnotesize$X(1)$}; &
		\node(X5) [fftterminal]{\footnotesize$X(5)$}; &
		\node(X3) [fftterminal]{\footnotesize$X(3)$}; &
		\node(X7) [fftterminal]{\footnotesize$X(7)$}; &
		\\
};
	\draw[arc,densely dotted] (a0.south) -- (n8_0_0.north);
	\draw[arc,densely dotted] (a4.south) -- (n8_0_4.north);
	\draw[arc,densely dotted] (a2.south) -- (n8_0_2.north);
	\draw[arc,densely dotted] (a6.south) -- (n8_0_6.north);
	\draw[arc,densely dotted] (a1.south) -- (n8_0_1.north);
	\draw[arc,densely dotted] (a5.south) -- (n8_0_5.north);
	\draw[arc,densely dotted] (a3.south) -- (n8_0_3.north);
	\draw[arc,densely dotted] (a7.south) -- (n8_0_7.north);
	\draw[arc] (n8_0_0.250) -- (n4_0_0.110);
	\draw[arc] (n8_0_0.290) -- (n4_4_0.110);
	\draw[arc] (n8_0_4.250) -- (n4_0_0.70);
	\draw[arc] (n8_0_4.290) -- (n4_4_0.70);
	\draw[arc] (n8_0_2.250) -- (n4_0_2.110);
	\draw[arc] (n8_0_2.290) -- (n4_4_2.110);
	\draw[arc] (n8_0_6.250) -- (n4_0_2.70);
	\draw[arc] (n8_0_6.290) -- (n4_4_2.70);
	\draw[arc] (n8_0_1.250) -- (n4_0_1.110);
	\draw[arc] (n8_0_1.290) -- (n4_4_1.110);
	\draw[arc] (n8_0_5.250) -- (n4_0_1.70);
	\draw[arc] (n8_0_5.290) -- (n4_4_1.70);
	\draw[arc] (n8_0_3.250) -- (n4_0_3.110);
	\draw[arc] (n8_0_3.290) -- (n4_4_3.110);
	\draw[arc] (n8_0_7.250) -- (n4_0_3.70);
	\draw[arc] (n8_0_7.290) -- (n4_4_3.70);
	\draw[arc] (n4_0_0.250) -- (n2_0_0.110);
	\draw[arc] (n4_0_0.290) -- (n2_4_0.110);
	\draw[arc] (n4_0_2.250) -- (n2_0_0.70);
	\draw[arc] (n4_0_2.290) -- (n2_4_0.70);
	\draw[arc] (n4_0_1.250) -- (n2_0_1.110);
	\draw[arc] (n4_0_1.290) -- (n2_4_1.110);
	\draw[arc] (n4_0_3.250) -- (n2_0_1.70);
	\draw[arc] (n4_0_3.290) -- (n2_4_1.70);
	\draw[arc] (n4_4_0.250) -- (n2_2_0.110);
	\draw[arc] (n4_4_0.290) -- (n2_6_0.110);
	\draw[arc] (n4_4_2.250) -- (n2_2_0.70);
	\draw[arc] (n4_4_2.290) -- (n2_6_0.70);
	\draw[arc] (n4_4_1.250) -- (n2_2_1.110);
	\draw[arc] (n4_4_1.290) -- (n2_6_1.110);
	\draw[arc] (n4_4_3.250) -- (n2_2_1.70);
	\draw[arc] (n4_4_3.290) -- (n2_6_1.70);
	\draw[arc] (n2_0_0.250) -- (n1_0_0.110);
	\draw[arc] (n2_0_0.290) -- (n1_4_0.110);
	\draw[arc] (n2_0_1.250) -- (n1_0_0.70);
	\draw[arc] (n2_0_1.290) -- (n1_4_0.70);
	\draw[arc] (n2_4_0.250) -- (n1_2_0.110);
	\draw[arc] (n2_4_0.290) -- (n1_6_0.110);
	\draw[arc] (n2_4_1.250) -- (n1_2_0.70);
	\draw[arc] (n2_4_1.290) -- (n1_6_0.70);
	\draw[arc] (n2_2_0.250) -- (n1_1_0.110);
	\draw[arc] (n2_2_0.290) -- (n1_5_0.110);
	\draw[arc] (n2_2_1.250) -- (n1_1_0.70);
	\draw[arc] (n2_2_1.290) -- (n1_5_0.70);
	\draw[arc] (n2_6_0.250) -- (n1_3_0.110);
	\draw[arc] (n2_6_0.290) -- (n1_7_0.110);
	\draw[arc] (n2_6_1.250) -- (n1_3_0.70);
	\draw[arc] (n2_6_1.290) -- (n1_7_0.70);
	\draw[arc] (n1_0_0.south) -- (X0.north);
	\draw[arc] (n1_4_0.south) -- (X4.north);
	\draw[arc] (n1_2_0.south) -- (X2.north);
	\draw[arc] (n1_6_0.south) -- (X6.north);
	\draw[arc] (n1_1_0.south) -- (X1.north);
	\draw[arc] (n1_5_0.south) -- (X5.north);
	\draw[arc] (n1_3_0.south) -- (X3.north);
	\draw[arc] (n1_7_0.south) -- (X7.north);
	\begin{pgfonlayer}{background}
		\node (rn8_0_0) [background,fit=(n8_0_0) (n8_0_7)] {};
		\node (rn4_0_0) [background,fit=(n4_0_0) (n4_0_3)] {};
		\node (rn4_4_0) [background,fit=(n4_4_0) (n4_4_3)] {};
		\node (rn2_0_0) [background,fit=(n2_0_0) (n2_0_1)] {};
		\node (rn2_4_0) [background,fit=(n2_4_0) (n2_4_1)] {};
		\node (rn2_2_0) [background,fit=(n2_2_0) (n2_2_1)] {};
		\node (rn2_6_0) [background,fit=(n2_6_0) (n2_6_1)] {};
		\node (rn1_0_0) [background,fit=(n1_0_0) (n1_0_0)] {};
		\node (rn1_4_0) [background,fit=(n1_4_0) (n1_4_0)] {};
		\node (rn1_2_0) [background,fit=(n1_2_0) (n1_2_0)] {};
		\node (rn1_6_0) [background,fit=(n1_6_0) (n1_6_0)] {};
		\node (rn1_1_0) [background,fit=(n1_1_0) (n1_1_0)] {};
		\node (rn1_5_0) [background,fit=(n1_5_0) (n1_5_0)] {};
		\node (rn1_3_0) [background,fit=(n1_3_0) (n1_3_0)] {};
		\node (rn1_7_0) [background,fit=(n1_7_0) (n1_7_0)] {};
	\end{pgfonlayer}
	\node [polytext] at (rn8_0_0.south west) {\tiny$x^8$-$\omega^{0}_{8}$};
	\node [polytext] at (rn4_0_0.south west) {\tiny$x^4$-$\omega^{0}_{8}$};
	\node [polytext] at (rn4_4_0.south west) {\tiny$x^4$-$\omega^{4}_{8}$};
	\node [polytext] at (rn2_0_0.south west) {\tiny$x^2$-$\omega^{0}_{8}$};
	\node [polytext] at (rn2_4_0.south west) {\tiny$x^2$-$\omega^{4}_{8}$};
	\node [polytext] at (rn2_2_0.south west) {\tiny$x^2$-$\omega^{2}_{8}$};
	\node [polytext] at (rn2_6_0.south west) {\tiny$x^2$-$\omega^{6}_{8}$};
	\node [polytext] at (rn1_0_0.south west) {\tiny$x$-$\omega^{0}_{8}$};
	\node [polytext] at (rn1_4_0.south west) {\tiny$x$-$\omega^{4}_{8}$};
	\node [polytext] at (rn1_2_0.south west) {\tiny$x$-$\omega^{2}_{8}$};
	\node [polytext] at (rn1_6_0.south west) {\tiny$x$-$\omega^{6}_{8}$};
	\node [polytext] at (rn1_1_0.south west) {\tiny$x$-$\omega^{1}_{8}$};
	\node [polytext] at (rn1_5_0.south west) {\tiny$x$-$\omega^{5}_{8}$};
	\node [polytext] at (rn1_3_0.south west) {\tiny$x$-$\omega^{3}_{8}$};
	\node [polytext] at (rn1_7_0.south west) {\tiny$x$-$\omega^{7}_{8}$};
	\node [costtext] at (n8_0_0.south east) {\tiny0};
	\node [costtext] at (n8_0_4.south east) {\tiny0};
	\node [costtext] at (n8_0_2.south east) {\tiny0};
	\node [costtext] at (n8_0_6.south east) {\tiny0};
	\node [costtext] at (n8_0_1.south east) {\tiny0};
	\node [costtext] at (n8_0_5.south east) {\tiny0};
	\node [costtext] at (n8_0_3.south east) {\tiny0};
	\node [costtext] at (n8_0_7.south east) {\tiny0};
	\node [costtext] at (n4_0_0.south east) {\tiny2};
	\node [costtext] at (n4_0_2.south east) {\tiny2};
	\node [costtext] at (n4_0_1.south east) {\tiny2};
	\node [costtext] at (n4_0_3.south east) {\tiny2};
	\node [costtext] at (n4_4_0.south east) {\tiny2};
	\node [costtext] at (n4_4_2.south east) {\tiny2};
	\node [costtext] at (n4_4_1.south east) {\tiny2};
	\node [costtext] at (n4_4_3.south east) {\tiny2};
	\node [costtext] at (n2_0_0.south east) {\tiny2};
	\node [costtext] at (n2_0_1.south east) {\tiny2};
	\node [costtext] at (n2_4_0.south east) {\tiny2};
	\node [costtext] at (n2_4_1.south east) {\tiny2};
	\node [costtext] at (n2_2_0.south east) {\tiny2};
	\node [costtext] at (n2_2_1.south east) {\tiny6};
	\node [costtext] at (n2_6_0.south east) {\tiny2};
	\node [costtext] at (n2_6_1.south east) {\tiny6};
	\node [costtext] at (n1_0_0.south east) {\tiny2};
	\node [costtext] at (n1_4_0.south east) {\tiny2};
	\node [costtext] at (n1_2_0.south east) {\tiny2};
	\node [costtext] at (n1_6_0.south east) {\tiny2};
	\node [costtext] at (n1_1_0.south east) {\tiny2};
	\node [costtext] at (n1_5_0.south east) {\tiny2};
	\node [costtext] at (n1_3_0.south east) {\tiny2};
	\node [costtext] at (n1_7_0.south east) {\tiny2};
\end{tikzpicture}

%% file: nodekey.tex
\begin{tikzpicture}[>=latex]
\node(nd) [rectangle split, rectangle split parts=3, rectangle split empty part height=1cm, rounded corners, text badly centered, minimum width=4.0cm, draw=black,inner sep=0.5mm]{};

\node(astride) [rectangle,left] at (nd.west) {\emph{stride}};
\node(flopcount) [rectangle,anchor=north west,inner sep=0mm] at (nd.south east) {\footnotesize\emph{FLOPs}};

\node(base) [rectangle,right=3mm] at (nd.west) {\emph{base}};
\node(wstride) [rectangle,left=3mm] at (nd.east) {$W_{s}$};
\node(lwbp) [rectangle,right=3mm] at (nd.text west) {$W_{b}$};
\node(rwbp) [rectangle,left=3mm] at (nd.text east) {$rW_{b}$};
\node(ltfp) [rectangle,right=3mm] at (nd.third west) {\emph{ltfp}};
\node(rtfp) [rectangle,left=3mm] at (nd.third east) {\emph{rtfp}};

\end{tikzpicture}

%% file: conjugatesplitradix_size16.tex
\begin{tikzpicture}[fftgraph]
	\matrix [fftmatrix] {
		& [-2mm]
		\node(a0) [fftterminal]{\scriptsize$a_{0}$}; &
		\node(a1) [fftterminal]{\scriptsize$a_{1}$}; &
		\node(a2) [fftterminal]{\scriptsize$a_{2}$}; &
		\node(a3) [fftterminal]{\scriptsize$a_{3}$}; &
		\node(a4) [fftterminal]{\scriptsize$a_{4}$}; &
		\node(a5) [fftterminal]{\scriptsize$a_{5}$}; &
		\node(a6) [fftterminal]{\scriptsize$a_{6}$}; &
		\node(a7) [fftterminal]{\scriptsize$a_{7}$}; &
		\node(a8) [fftterminal]{\scriptsize$a_{8}$}; &
		\node(a9) [fftterminal]{\scriptsize$a_{9}$}; &
		\node(a10) [fftterminal]{\scriptsize$a_{10}$}; &
		\node(a11) [fftterminal]{\scriptsize$a_{11}$}; &
		\node(a12) [fftterminal]{\scriptsize$a_{12}$}; &
		\node(a13) [fftterminal]{\scriptsize$a_{13}$}; &
		\node(a14) [fftterminal]{\scriptsize$a_{14}$}; &
		\node(a15) [fftterminal]{\scriptsize$a_{15}$}; &
		\\[-6mm]
		\node (stride16) [stride]{\footnotesize\bf16}; & [-2mm]
		\node (n16_0_0) [fftnode,densely dotted]{\tiny\underline{0}\nodepart{second}\scriptsize\bf0.0\nodepart{third}\tiny\sl0  0}; &
		\node (n16_0_1) [fftnode,densely dotted]{\tiny\underline{0}\nodepart{second}\scriptsize\bf1.0\nodepart{third}\tiny\sl0  0}; &
		\node (n16_0_2) [fftnode,densely dotted]{\tiny\underline{0}\nodepart{second}\scriptsize\bf2.0\nodepart{third}\tiny\sl0  0}; &
		\node (n16_0_3) [fftnode,densely dotted]{\tiny\underline{0}\nodepart{second}\scriptsize\bf3.0\nodepart{third}\tiny\sl0  0}; &
		\node (n16_0_4) [fftnode,densely dotted]{\tiny\underline{0}\nodepart{second}\scriptsize\bf4.0\nodepart{third}\tiny\sl0  0}; &
		\node (n16_0_5) [fftnode,densely dotted]{\tiny\underline{0}\nodepart{second}\scriptsize\bf5.0\nodepart{third}\tiny\sl0  0}; &
		\node (n16_0_6) [fftnode,densely dotted]{\tiny\underline{0}\nodepart{second}\scriptsize\bf6.0\nodepart{third}\tiny\sl0  8}; &
		\node (n16_0_7) [fftnode,densely dotted]{\tiny\underline{0}\nodepart{second}\scriptsize\bf7.0\nodepart{third}\tiny\sl0  8}; &
		\node (n16_0_8) [fftnode,densely dotted]{\tiny\underline{0}\nodepart{second}\scriptsize\bf8.0\nodepart{third}\tiny\sl0  8}; &
		\node (n16_0_9) [fftnode,densely dotted]{\tiny\underline{0}\nodepart{second}\scriptsize\bf9.0\nodepart{third}\tiny\sl0  8}; &
		\node (n16_0_10) [fftnode,densely dotted]{\tiny\underline{0}\nodepart{second}\scriptsize\bf10.0\nodepart{third}\tiny\sl0  8}; &
		\node (n16_0_11) [fftnode,densely dotted]{\tiny\underline{0}\nodepart{second}\scriptsize\bf11.0\nodepart{third}\tiny\sl0  8}; &
		\node (n16_0_12) [fftnode,densely dotted]{\tiny\underline{0}\nodepart{second}\scriptsize\bf12.0\nodepart{third}\tiny\sl0  8}; &
		\node (n16_0_13) [fftnode,densely dotted]{\tiny\underline{0}\nodepart{second}\scriptsize\bf13.0\nodepart{third}\tiny\sl0  8}; &
		\node (n16_0_14) [fftnode,densely dotted]{\tiny\underline{0}\nodepart{second}\scriptsize\bf14.0\nodepart{third}\tiny\sl0  0}; &
		\node (n16_0_15) [fftnode,densely dotted]{\tiny\underline{0}\nodepart{second}\scriptsize\bf15.0\nodepart{third}\tiny\sl0  0}; &
		\\
		\node (stride8) [stride]{\footnotesize\bf8}; & [-2mm]
		\node (n8_0_0) [fftnode]{\tiny\underline{0}  0\nodepart{second}\scriptsize\bf0.0\nodepart{third}\tiny\sl0  0}; &
		\node (n8_0_1) [fftnode]{\tiny\underline{0}  0\nodepart{second}\scriptsize\bf1.0\nodepart{third}\tiny\sl0  0}; &
		\node (n8_0_2) [fftnode]{\tiny\underline{0}  0\nodepart{second}\scriptsize\bf2.0\nodepart{third}\tiny\sl0  0}; &
		\node (n8_0_3) [fftnode]{\tiny\underline{0}  0\nodepart{second}\scriptsize\bf3.0\nodepart{third}\tiny\sl0  8}; &
		\node (n8_0_4) [fftnode]{\tiny\underline{0}  0\nodepart{second}\scriptsize\bf4.0\nodepart{third}\tiny\sl0  8}; &
		\node (n8_0_5) [fftnode]{\tiny\underline{0}  0\nodepart{second}\scriptsize\bf5.0\nodepart{third}\tiny\sl0  8}; &
		\node (n8_0_6) [fftnode]{\tiny\underline{0}  0\nodepart{second}\scriptsize\bf6.0\nodepart{third}\tiny\sl0  8}; &
		\node (n8_0_7) [fftnode]{\tiny\underline{0}  0\nodepart{second}\scriptsize\bf7.0\nodepart{third}\tiny\sl0  0}; &
		\node (n8_8_0) [fftnode]{\tiny\underline{0}  8\nodepart{second}\scriptsize\bf0.8\nodepart{third}\tiny\sl0  0}; &
		\node (n8_8_1) [fftnode]{\tiny\underline{0}  8\nodepart{second}\scriptsize\bf1.8\nodepart{third}\tiny\sl0  0}; &
		\node (n8_8_2) [fftnode]{\tiny\underline{0}  8\nodepart{second}\scriptsize\bf2.8\nodepart{third}\tiny\sl2  2}; &
		\node (n8_8_3) [fftnode]{\tiny\underline{0}  8\nodepart{second}\scriptsize\bf3.8\nodepart{third}\tiny\sl4  12}; &
		\node (n8_8_4) [fftnode]{\tiny\underline{0}  8\nodepart{second}\scriptsize\bf4.8\nodepart{third}\tiny\sl4  12}; &
		\node (n8_8_5) [fftnode]{\tiny\underline{0}  8\nodepart{second}\scriptsize\bf5.8\nodepart{third}\tiny\sl4  12}; &
		\node (n8_8_6) [fftnode]{\tiny\underline{8}  0\nodepart{second}\scriptsize\bf6.8\nodepart{third}\tiny\sl14  6}; &
		\node (n8_8_7) [fftnode]{\tiny\underline{8}  0\nodepart{second}\scriptsize\bf7.8\nodepart{third}\tiny\sl0  0}; &
		\\
		\node (stride4) [stride]{\footnotesize\bf4}; & [-2mm]
		\node (n4_0_0) [fftnode]{\tiny\underline{0}  0\nodepart{second}\scriptsize\bf0.0\nodepart{third}\tiny\sl0  0}; &
		\node (n4_0_1) [fftnode]{\tiny\underline{0}  0\nodepart{second}\scriptsize\bf1.0\nodepart{third}\tiny\sl0  0}; &
		\node (n4_0_2) [fftnode]{\tiny\underline{0}  0\nodepart{second}\scriptsize\bf2.0\nodepart{third}\tiny\sl0  8}; &
		\node (n4_0_3) [fftnode]{\tiny\underline{0}  0\nodepart{second}\scriptsize\bf3.0\nodepart{third}\tiny\sl0  8}; &
		\node (n4_8_0) [fftnode]{\tiny\underline{0}  8\nodepart{second}\scriptsize\bf0.8\nodepart{third}\tiny\sl0  0}; &
		\node (n4_8_1) [fftnode]{\tiny\underline{0}  8\nodepart{second}\scriptsize\bf1.8\nodepart{third}\tiny\sl2  2}; &
		\node (n4_8_2) [fftnode]{\tiny\underline{0}  8\nodepart{second}\scriptsize\bf2.8\nodepart{third}\tiny\sl4  12}; &
		\node (n4_8_3) [fftnode]{\tiny\underline{8}  0\nodepart{second}\scriptsize\bf3.8\nodepart{third}\tiny\sl14  6}; &
		\node (n4_4_0) [fftnode]{\tiny\underline{0}  4\nodepart{second}\scriptsize\bf0.4\nodepart{third}\tiny\sl0  0}; &
		\node (n4_4_1) [fftnode]{\tiny\underline{0}  4\nodepart{second}\scriptsize\bf1.4\nodepart{third}\tiny\sl1  1}; &
		\node (n4_4_2) [fftnode]{\tiny\underline{2}  6\nodepart{second}\scriptsize\bf2.4\nodepart{third}\tiny\sl0  8}; &
		\node (n4_4_3) [fftnode]{\tiny\underline{4}  8\nodepart{second}\scriptsize\bf3.4\nodepart{third}\tiny\sl15  7}; &
		\node (n4_12_0) [fftnode]{\tiny\underline{0}  12\nodepart{second}\scriptsize\bf0.12\nodepart{third}\tiny\sl0  0}; &
		\node (n4_12_1) [fftnode]{\tiny\underline{0}  12\nodepart{second}\scriptsize\bf1.12\nodepart{third}\tiny\sl3  3}; &
		\node (n4_12_2) [fftnode]{\tiny\underline{2}  14\nodepart{second}\scriptsize\bf2.12\nodepart{third}\tiny\sl4  12}; &
		\node (n4_12_3) [fftnode]{\tiny\underline{12}  8\nodepart{second}\scriptsize\bf3.12\nodepart{third}\tiny\sl13  5}; &
		\\
		\node (stride2) [stride]{\footnotesize\bf2}; & [-2mm]
		\node (n2_0_0) [fftnode]{\tiny\underline{0}  0\nodepart{second}\scriptsize\bf0.0\nodepart{third}\tiny\sl0  0}; &
		\node (n2_0_1) [fftnode]{\tiny\underline{0}  0\nodepart{second}\scriptsize\bf1.0\nodepart{third}\tiny\sl0  8}; &
		\node (n2_8_0) [fftnode]{\tiny\underline{0}  8\nodepart{second}\scriptsize\bf0.8\nodepart{third}\tiny\sl0  0}; &
		\node (n2_8_1) [fftnode]{\tiny\underline{0}  8\nodepart{second}\scriptsize\bf1.8\nodepart{third}\tiny\sl4  12}; &
		\node (n2_4_0) [fftnode]{\tiny\underline{0}  4\nodepart{second}\scriptsize\bf0.4\nodepart{third}\tiny\sl0  0}; &
		\node (n2_4_1) [fftnode]{\tiny\underline{2}  6\nodepart{second}\scriptsize\bf1.4\nodepart{third}\tiny\sl0  8}; &
		\node (n2_12_0) [fftnode]{\tiny\underline{0}  12\nodepart{second}\scriptsize\bf0.12\nodepart{third}\tiny\sl0  0}; &
		\node (n2_12_1) [fftnode]{\tiny\underline{2}  14\nodepart{second}\scriptsize\bf1.12\nodepart{third}\tiny\sl4  12}; &
		\node (n2_2_0) [fftnode]{\tiny\underline{0}  2\nodepart{second}\scriptsize\bf0.2\nodepart{third}\tiny\sl0  0}; &
		\node (n2_2_1) [fftnode]{\tiny\underline{1}  3\nodepart{second}\scriptsize\bf1.2\nodepart{third}\tiny\sl0  8}; &
		\node (n2_10_0) [fftnode]{\tiny\underline{0}  10\nodepart{second}\scriptsize\bf0.10\nodepart{third}\tiny\sl0  0}; &
		\node (n2_10_1) [fftnode]{\tiny\underline{1}  11\nodepart{second}\scriptsize\bf1.10\nodepart{third}\tiny\sl4  12}; &
		\node (n2_6_0) [fftnode]{\tiny\underline{0}  6\nodepart{second}\scriptsize\bf0.6\nodepart{third}\tiny\sl0  0}; &
		\node (n2_6_1) [fftnode]{\tiny\underline{3}  9\nodepart{second}\scriptsize\bf1.6\nodepart{third}\tiny\sl0  8}; &
		\node (n2_14_0) [fftnode]{\tiny\underline{0}  14\nodepart{second}\scriptsize\bf0.14\nodepart{third}\tiny\sl0  0}; &
		\node (n2_14_1) [fftnode]{\tiny\underline{3}  1\nodepart{second}\scriptsize\bf1.14\nodepart{third}\tiny\sl4  12}; &
		\\
		\node (stride1) [stride]{\footnotesize\bf1}; & [-2mm]
		\node (n1_0_0) [fftnode]{\tiny\underline{0}  0\nodepart{second}\scriptsize\bf0.0\nodepart{third}\tiny\sl0}; &
		\node (n1_8_0) [fftnode]{\tiny\underline{0}  8\nodepart{second}\scriptsize\bf0.8\nodepart{third}\tiny\sl0}; &
		\node (n1_4_0) [fftnode]{\tiny\underline{0}  4\nodepart{second}\scriptsize\bf0.4\nodepart{third}\tiny\sl0}; &
		\node (n1_12_0) [fftnode]{\tiny\underline{0}  12\nodepart{second}\scriptsize\bf0.12\nodepart{third}\tiny\sl0}; &
		\node (n1_2_0) [fftnode]{\tiny\underline{0}  2\nodepart{second}\scriptsize\bf0.2\nodepart{third}\tiny\sl0}; &
		\node (n1_10_0) [fftnode]{\tiny\underline{0}  10\nodepart{second}\scriptsize\bf0.10\nodepart{third}\tiny\sl0}; &
		\node (n1_6_0) [fftnode]{\tiny\underline{0}  6\nodepart{second}\scriptsize\bf0.6\nodepart{third}\tiny\sl0}; &
		\node (n1_14_0) [fftnode]{\tiny\underline{0}  14\nodepart{second}\scriptsize\bf0.14\nodepart{third}\tiny\sl0}; &
		\node (n1_1_0) [fftnode]{\tiny\underline{0}  1\nodepart{second}\scriptsize\bf0.1\nodepart{third}\tiny\sl0}; &
		\node (n1_9_0) [fftnode]{\tiny\underline{0}  9\nodepart{second}\scriptsize\bf0.9\nodepart{third}\tiny\sl0}; &
		\node (n1_5_0) [fftnode]{\tiny\underline{0}  5\nodepart{second}\scriptsize\bf0.5\nodepart{third}\tiny\sl0}; &
		\node (n1_13_0) [fftnode]{\tiny\underline{0}  13\nodepart{second}\scriptsize\bf0.13\nodepart{third}\tiny\sl0}; &
		\node (n1_3_0) [fftnode]{\tiny\underline{0}  3\nodepart{second}\scriptsize\bf0.3\nodepart{third}\tiny\sl0}; &
		\node (n1_11_0) [fftnode]{\tiny\underline{0}  11\nodepart{second}\scriptsize\bf0.11\nodepart{third}\tiny\sl0}; &
		\node (n1_7_0) [fftnode]{\tiny\underline{0}  7\nodepart{second}\scriptsize\bf0.7\nodepart{third}\tiny\sl0}; &
		\node (n1_15_0) [fftnode]{\tiny\underline{0}  15\nodepart{second}\scriptsize\bf0.15\nodepart{third}\tiny\sl0}; &
		\\[-6mm]
		& [-2mm]
		\node(X0) [fftterminal]{\tiny$X(0)$}; &
		\node(X8) [fftterminal]{\tiny$X(8)$}; &
		\node(X4) [fftterminal]{\tiny$X(4)$}; &
		\node(X12) [fftterminal]{\tiny$X(12)$}; &
		\node(X2) [fftterminal]{\tiny$X(2)$}; &
		\node(X10) [fftterminal]{\tiny$X(10)$}; &
		\node(X6) [fftterminal]{\tiny$X(6)$}; &
		\node(X14) [fftterminal]{\tiny$X(14)$}; &
		\node(X1) [fftterminal]{\tiny$X(1)$}; &
		\node(X9) [fftterminal]{\tiny$X(9)$}; &
		\node(X5) [fftterminal]{\tiny$X(5)$}; &
		\node(X13) [fftterminal]{\tiny$X(13)$}; &
		\node(X3) [fftterminal]{\tiny$X(3)$}; &
		\node(X11) [fftterminal]{\tiny$X(11)$}; &
		\node(X7) [fftterminal]{\tiny$X(7)$}; &
		\node(X15) [fftterminal]{\tiny$X(15)$}; &
		\\
};
	\draw[arc,densely dotted] (a0.south) -- (n16_0_0.north);
	\draw[arc,densely dotted] (a8.south) -- (n16_0_8.north);
	\draw[arc,densely dotted] (a4.south) -- (n16_0_4.north);
	\draw[arc,densely dotted] (a12.south) -- (n16_0_12.north);
	\draw[arc,densely dotted] (a2.south) -- (n16_0_2.north);
	\draw[arc,densely dotted] (a10.south) -- (n16_0_10.north);
	\draw[arc,densely dotted] (a6.south) -- (n16_0_6.north);
	\draw[arc,densely dotted] (a14.south) -- (n16_0_14.north);
	\draw[arc,densely dotted] (a1.south) -- (n16_0_1.north);
	\draw[arc,densely dotted] (a9.south) -- (n16_0_9.north);
	\draw[arc,densely dotted] (a5.south) -- (n16_0_5.north);
	\draw[arc,densely dotted] (a13.south) -- (n16_0_13.north);
	\draw[arc,densely dotted,very thick] (a3.south) -- (n16_0_3.north);
	\draw[arc,densely dotted] (a11.south) -- (n16_0_11.north);
	\draw[arc,densely dotted] (a7.south) -- (n16_0_7.north);
	\draw[arc,densely dotted] (a15.south) -- (n16_0_15.north);
	\draw[arc] (n16_0_0.250) -- (n8_0_0.110);
	\draw[arc] (n16_0_0.290) -- (n8_8_0.110);
	\draw[arc] (n16_0_8.250) -- (n8_0_0.70);
	\draw[arc] (n16_0_8.290) -- (n8_8_0.70);
	\draw[arc] (n16_0_4.250) -- (n8_0_4.110);
	\draw[arc] (n16_0_4.290) -- (n8_8_4.110);
	\draw[arc] (n16_0_12.250) -- (n8_0_4.70);
	\draw[arc] (n16_0_12.290) -- (n8_8_4.70);
	\draw[arc] (n16_0_2.250) -- (n8_0_2.110);
	\draw[arc] (n16_0_2.290) -- (n8_8_2.110);
	\draw[arc] (n16_0_10.250) -- (n8_0_2.70);
	\draw[arc] (n16_0_10.290) -- (n8_8_2.70);
	\draw[arc] (n16_0_6.250) -- (n8_0_6.110);
	\draw[arc] (n16_0_6.290) -- (n8_8_6.110);
	\draw[arc] (n16_0_14.250) -- (n8_0_6.70);
	\draw[arc] (n16_0_14.290) -- (n8_8_6.70);
	\draw[arc] (n16_0_1.250) -- (n8_0_1.110);
	\draw[arc] (n16_0_1.290) -- (n8_8_1.110);
	\draw[arc] (n16_0_9.250) -- (n8_0_1.70);
	\draw[arc] (n16_0_9.290) -- (n8_8_1.70);
	\draw[arc] (n16_0_5.250) -- (n8_0_5.110);
	\draw[arc] (n16_0_5.290) -- (n8_8_5.110);
	\draw[arc] (n16_0_13.250) -- (n8_0_5.70);
	\draw[arc] (n16_0_13.290) -- (n8_8_5.70);
	\draw[arc] (n16_0_3.250) -- (n8_0_3.110);
	\draw[arc,very thick] (n16_0_3.290) -- (n8_8_3.110);
	\draw[arc] (n16_0_11.250) -- (n8_0_3.70);
	\draw[arc] (n16_0_11.290) -- (n8_8_3.70);
	\draw[arc] (n16_0_7.250) -- (n8_0_7.110);
	\draw[arc] (n16_0_7.290) -- (n8_8_7.110);
	\draw[arc] (n16_0_15.250) -- (n8_0_7.70);
	\draw[arc] (n16_0_15.290) -- (n8_8_7.70);
	\draw[arc] (n8_0_0.250) -- (n4_0_0.110);
	\draw[arc] (n8_0_0.290) -- (n4_8_0.110);
	\draw[arc] (n8_0_4.250) -- (n4_0_0.70);
	\draw[arc] (n8_0_4.290) -- (n4_8_0.70);
	\draw[arc] (n8_0_2.250) -- (n4_0_2.110);
	\draw[arc] (n8_0_2.290) -- (n4_8_2.110);
	\draw[arc] (n8_0_6.250) -- (n4_0_2.70);
	\draw[arc] (n8_0_6.290) -- (n4_8_2.70);
	\draw[arc] (n8_0_1.250) -- (n4_0_1.110);
	\draw[arc] (n8_0_1.290) -- (n4_8_1.110);
	\draw[arc] (n8_0_5.250) -- (n4_0_1.70);
	\draw[arc] (n8_0_5.290) -- (n4_8_1.70);
	\draw[arc] (n8_0_3.250) -- (n4_0_3.110);
	\draw[arc] (n8_0_3.290) -- (n4_8_3.110);
	\draw[arc] (n8_0_7.250) -- (n4_0_3.70);
	\draw[arc] (n8_0_7.290) -- (n4_8_3.70);
	\draw[arc] (n8_8_0.250) -- (n4_4_0.110);
	\draw[arc] (n8_8_0.290) -- (n4_12_0.110);
	\draw[arc] (n8_8_4.250) -- (n4_4_0.70);
	\draw[arc] (n8_8_4.290) -- (n4_12_0.70);
	\draw[arc] (n8_8_2.250) -- (n4_4_2.110);
	\draw[arc] (n8_8_2.290) -- (n4_12_2.110);
	\draw[arc] (n8_8_6.250) -- (n4_4_2.70);
	\draw[arc] (n8_8_6.290) -- (n4_12_2.70);
	\draw[arc] (n8_8_1.250) -- (n4_4_1.110);
	\draw[arc] (n8_8_1.290) -- (n4_12_1.110);
	\draw[arc] (n8_8_5.250) -- (n4_4_1.70);
	\draw[arc] (n8_8_5.290) -- (n4_12_1.70);
	\draw[arc] (n8_8_3.250) -- (n4_4_3.110);
	\draw[arc,very thick] (n8_8_3.290) -- (n4_12_3.110);
	\draw[arc] (n8_8_7.250) -- (n4_4_3.70);
	\draw[arc] (n8_8_7.290) -- (n4_12_3.70);
	\draw[arc] (n4_0_0.250) -- (n2_0_0.110);
	\draw[arc] (n4_0_0.290) -- (n2_8_0.110);
	\draw[arc] (n4_0_2.250) -- (n2_0_0.70);
	\draw[arc] (n4_0_2.290) -- (n2_8_0.70);
	\draw[arc] (n4_0_1.250) -- (n2_0_1.110);
	\draw[arc] (n4_0_1.290) -- (n2_8_1.110);
	\draw[arc] (n4_0_3.250) -- (n2_0_1.70);
	\draw[arc] (n4_0_3.290) -- (n2_8_1.70);
	\draw[arc] (n4_8_0.250) -- (n2_4_0.110);
	\draw[arc] (n4_8_0.290) -- (n2_12_0.110);
	\draw[arc] (n4_8_2.250) -- (n2_4_0.70);
	\draw[arc] (n4_8_2.290) -- (n2_12_0.70);
	\draw[arc] (n4_8_1.250) -- (n2_4_1.110);
	\draw[arc] (n4_8_1.290) -- (n2_12_1.110);
	\draw[arc] (n4_8_3.250) -- (n2_4_1.70);
	\draw[arc] (n4_8_3.290) -- (n2_12_1.70);
	\draw[arc] (n4_4_0.250) -- (n2_2_0.110);
	\draw[arc] (n4_4_0.290) -- (n2_10_0.110);
	\draw[arc] (n4_4_2.250) -- (n2_2_0.70);
	\draw[arc] (n4_4_2.290) -- (n2_10_0.70);
	\draw[arc] (n4_4_1.250) -- (n2_2_1.110);
	\draw[arc] (n4_4_1.290) -- (n2_10_1.110);
	\draw[arc] (n4_4_3.250) -- (n2_2_1.70);
	\draw[arc] (n4_4_3.290) -- (n2_10_1.70);
	\draw[arc] (n4_12_0.250) -- (n2_6_0.110);
	\draw[arc] (n4_12_0.290) -- (n2_14_0.110);
	\draw[arc] (n4_12_2.250) -- (n2_6_0.70);
	\draw[arc] (n4_12_2.290) -- (n2_14_0.70);
	\draw[arc] (n4_12_1.250) -- (n2_6_1.110);
	\draw[arc] (n4_12_1.290) -- (n2_14_1.110);
	\draw[arc,very thick] (n4_12_3.250) -- (n2_6_1.70);
	\draw[arc] (n4_12_3.290) -- (n2_14_1.70);
	\draw[arc] (n2_0_0.250) -- (n1_0_0.110);
	\draw[arc] (n2_0_0.290) -- (n1_8_0.110);
	\draw[arc] (n2_0_1.250) -- (n1_0_0.70);
	\draw[arc] (n2_0_1.290) -- (n1_8_0.70);
	\draw[arc] (n2_8_0.250) -- (n1_4_0.110);
	\draw[arc] (n2_8_0.290) -- (n1_12_0.110);
	\draw[arc] (n2_8_1.250) -- (n1_4_0.70);
	\draw[arc] (n2_8_1.290) -- (n1_12_0.70);
	\draw[arc] (n2_4_0.250) -- (n1_2_0.110);
	\draw[arc] (n2_4_0.290) -- (n1_10_0.110);
	\draw[arc] (n2_4_1.250) -- (n1_2_0.70);
	\draw[arc] (n2_4_1.290) -- (n1_10_0.70);
	\draw[arc] (n2_12_0.250) -- (n1_6_0.110);
	\draw[arc] (n2_12_0.290) -- (n1_14_0.110);
	\draw[arc] (n2_12_1.250) -- (n1_6_0.70);
	\draw[arc] (n2_12_1.290) -- (n1_14_0.70);
	\draw[arc] (n2_2_0.250) -- (n1_1_0.110);
	\draw[arc] (n2_2_0.290) -- (n1_9_0.110);
	\draw[arc] (n2_2_1.250) -- (n1_1_0.70);
	\draw[arc] (n2_2_1.290) -- (n1_9_0.70);
	\draw[arc] (n2_10_0.250) -- (n1_5_0.110);
	\draw[arc] (n2_10_0.290) -- (n1_13_0.110);
	\draw[arc] (n2_10_1.250) -- (n1_5_0.70);
	\draw[arc] (n2_10_1.290) -- (n1_13_0.70);
	\draw[arc] (n2_6_0.250) -- (n1_3_0.110);
	\draw[arc] (n2_6_0.290) -- (n1_11_0.110);
	\draw[arc] (n2_6_1.250) -- (n1_3_0.70);
	\draw[arc] (n2_6_1.290) -- (n1_11_0.70);
	\draw[arc] (n2_14_0.250) -- (n1_7_0.110);
	\draw[arc] (n2_14_0.290) -- (n1_15_0.110);
	\draw[arc] (n2_14_1.250) -- (n1_7_0.70);
	\draw[arc] (n2_14_1.290) -- (n1_15_0.70);
	\draw[arc] (n1_0_0.south) -- (X0.north);
	\draw[arc] (n1_8_0.south) -- (X8.north);
	\draw[arc] (n1_4_0.south) -- (X4.north);
	\draw[arc] (n1_12_0.south) -- (X12.north);
	\draw[arc] (n1_2_0.south) -- (X2.north);
	\draw[arc] (n1_10_0.south) -- (X10.north);
	\draw[arc] (n1_6_0.south) -- (X6.north);
	\draw[arc] (n1_14_0.south) -- (X14.north);
	\draw[arc] (n1_1_0.south) -- (X1.north);
	\draw[arc] (n1_9_0.south) -- (X9.north);
	\draw[arc] (n1_5_0.south) -- (X5.north);
	\draw[arc] (n1_13_0.south) -- (X13.north);
	\draw[arc] (n1_3_0.south) -- (X3.north);
	\draw[arc] (n1_11_0.south) -- (X11.north);
	\draw[arc] (n1_7_0.south) -- (X7.north);
	\draw[arc] (n1_15_0.south) -- (X15.north);
	\begin{pgfonlayer}{background}
		\node (rn16_0_0) [background,fit=(n16_0_0) (n16_0_15)] {};
		\node (rn8_0_0) [background,fit=(n8_0_0) (n8_0_7)] {};
		\node (rn8_8_0) [background,fit=(n8_8_0) (n8_8_7)] {};
		\node (rn4_0_0) [background,fit=(n4_0_0) (n4_0_3)] {};
		\node (rn4_8_0) [background,fit=(n4_8_0) (n4_8_3)] {};
		\node (rn4_4_0) [background,fit=(n4_4_0) (n4_4_3)] {};
		\node (rn4_12_0) [background,fit=(n4_12_0) (n4_12_3)] {};
		\node (rn2_0_0) [background,fit=(n2_0_0) (n2_0_1)] {};
		\node (rn2_8_0) [background,fit=(n2_8_0) (n2_8_1)] {};
		\node (rn2_4_0) [background,fit=(n2_4_0) (n2_4_1)] {};
		\node (rn2_12_0) [background,fit=(n2_12_0) (n2_12_1)] {};
		\node (rn2_2_0) [background,fit=(n2_2_0) (n2_2_1)] {};
		\node (rn2_10_0) [background,fit=(n2_10_0) (n2_10_1)] {};
		\node (rn2_6_0) [background,fit=(n2_6_0) (n2_6_1)] {};
		\node (rn2_14_0) [background,fit=(n2_14_0) (n2_14_1)] {};
		\node (rn1_0_0) [background,fit=(n1_0_0) (n1_0_0)] {};
		\node (rn1_8_0) [background,fit=(n1_8_0) (n1_8_0)] {};
		\node (rn1_4_0) [background,fit=(n1_4_0) (n1_4_0)] {};
		\node (rn1_12_0) [background,fit=(n1_12_0) (n1_12_0)] {};
		\node (rn1_2_0) [background,fit=(n1_2_0) (n1_2_0)] {};
		\node (rn1_10_0) [background,fit=(n1_10_0) (n1_10_0)] {};
		\node (rn1_6_0) [background,fit=(n1_6_0) (n1_6_0)] {};
		\node (rn1_14_0) [background,fit=(n1_14_0) (n1_14_0)] {};
		\node (rn1_1_0) [background,fit=(n1_1_0) (n1_1_0)] {};
		\node (rn1_9_0) [background,fit=(n1_9_0) (n1_9_0)] {};
		\node (rn1_5_0) [background,fit=(n1_5_0) (n1_5_0)] {};
		\node (rn1_13_0) [background,fit=(n1_13_0) (n1_13_0)] {};
		\node (rn1_3_0) [background,fit=(n1_3_0) (n1_3_0)] {};
		\node (rn1_11_0) [background,fit=(n1_11_0) (n1_11_0)] {};
		\node (rn1_7_0) [background,fit=(n1_7_0) (n1_7_0)] {};
		\node (rn1_15_0) [background,fit=(n1_15_0) (n1_15_0)] {};
	\end{pgfonlayer}
	\node [costtext] at (n16_0_0.south east) {\tiny0};
	\node [costtext] at (n16_0_8.south east) {\tiny0};
	\node [costtext] at (n16_0_4.south east) {\tiny0};
	\node [costtext] at (n16_0_12.south east) {\tiny0};
	\node [costtext] at (n16_0_2.south east) {\tiny0};
	\node [costtext] at (n16_0_10.south east) {\tiny0};
	\node [costtext] at (n16_0_6.south east) {\tiny0};
	\node [costtext] at (n16_0_14.south east) {\tiny0};
	\node [costtext] at (n16_0_1.south east) {\tiny0};
	\node [costtext] at (n16_0_9.south east) {\tiny0};
	\node [costtext] at (n16_0_5.south east) {\tiny0};
	\node [costtext] at (n16_0_13.south east) {\tiny0};
	\node [costtext] at (n16_0_3.south east) {\tiny0};
	\node [costtext] at (n16_0_11.south east) {\tiny0};
	\node [costtext] at (n16_0_7.south east) {\tiny0};
	\node [costtext] at (n16_0_15.south east) {\tiny0};
	\node [costtext] at (n8_0_0.south east) {\tiny2};
	\node [costtext] at (n8_0_4.south east) {\tiny2};
	\node [costtext] at (n8_0_2.south east) {\tiny2};
	\node [costtext] at (n8_0_6.south east) {\tiny2};
	\node [costtext] at (n8_0_1.south east) {\tiny2};
	\node [costtext] at (n8_0_5.south east) {\tiny2};
	\node [costtext] at (n8_0_3.south east) {\tiny2};
	\node [costtext] at (n8_0_7.south east) {\tiny2};
	\node [costtext] at (n8_8_0.south east) {\tiny2};
	\node [costtext] at (n8_8_4.south east) {\tiny2};
	\node [costtext] at (n8_8_2.south east) {\tiny6};
	\node [costtext] at (n8_8_6.south east) {\tiny6};
	\node [costtext] at (n8_8_1.south east) {\tiny2};
	\node [costtext] at (n8_8_5.south east) {\tiny2};
	\node [costtext] at (n8_8_3.south east) {\tiny2};
	\node [costtext] at (n8_8_7.south east) {\tiny2};
	\node [costtext] at (n4_0_0.south east) {\tiny2};
	\node [costtext] at (n4_0_2.south east) {\tiny2};
	\node [costtext] at (n4_0_1.south east) {\tiny2};
	\node [costtext] at (n4_0_3.south east) {\tiny2};
	\node [costtext] at (n4_8_0.south east) {\tiny2};
	\node [costtext] at (n4_8_2.south east) {\tiny2};
	\node [costtext] at (n4_8_1.south east) {\tiny6};
	\node [costtext] at (n4_8_3.south east) {\tiny6};
	\node [costtext] at (n4_4_0.south east) {\tiny2};
	\node [costtext] at (n4_4_2.south east) {\tiny2};
	\node [costtext] at (n4_4_1.south east) {\tiny8};
	\node [costtext] at (n4_4_3.south east) {\tiny8};
	\node [costtext] at (n4_12_0.south east) {\tiny2};
	\node [costtext] at (n4_12_2.south east) {\tiny2};
	\node [costtext] at (n4_12_1.south east) {\tiny8};
	\node [costtext] at (n4_12_3.south east) {\tiny8};
	\node [costtext] at (n2_0_0.south east) {\tiny2};
	\node [costtext] at (n2_0_1.south east) {\tiny2};
	\node [costtext] at (n2_8_0.south east) {\tiny2};
	\node [costtext] at (n2_8_1.south east) {\tiny2};
	\node [costtext] at (n2_4_0.south east) {\tiny2};
	\node [costtext] at (n2_4_1.south east) {\tiny2};
	\node [costtext] at (n2_12_0.south east) {\tiny2};
	\node [costtext] at (n2_12_1.south east) {\tiny2};
	\node [costtext] at (n2_2_0.south east) {\tiny2};
	\node [costtext] at (n2_2_1.south east) {\tiny2};
	\node [costtext] at (n2_10_0.south east) {\tiny2};
	\node [costtext] at (n2_10_1.south east) {\tiny2};
	\node [costtext] at (n2_6_0.south east) {\tiny2};
	\node [costtext] at (n2_6_1.south east) {\tiny2};
	\node [costtext] at (n2_14_0.south east) {\tiny2};
	\node [costtext] at (n2_14_1.south east) {\tiny2};
	\node [costtext] at (n1_0_0.south east) {\tiny2};
	\node [costtext] at (n1_8_0.south east) {\tiny2};
	\node [costtext] at (n1_4_0.south east) {\tiny2};
	\node [costtext] at (n1_12_0.south east) {\tiny2};
	\node [costtext] at (n1_2_0.south east) {\tiny2};
	\node [costtext] at (n1_10_0.south east) {\tiny2};
	\node [costtext] at (n1_6_0.south east) {\tiny2};
	\node [costtext] at (n1_14_0.south east) {\tiny2};
	\node [costtext] at (n1_1_0.south east) {\tiny2};
	\node [costtext] at (n1_9_0.south east) {\tiny2};
	\node [costtext] at (n1_5_0.south east) {\tiny2};
	\node [costtext] at (n1_13_0.south east) {\tiny2};
	\node [costtext] at (n1_3_0.south east) {\tiny2};
	\node [costtext] at (n1_11_0.south east) {\tiny2};
	\node [costtext] at (n1_7_0.south east) {\tiny2};
	\node [costtext] at (n1_15_0.south east) {\tiny2};
\end{tikzpicture}

%% file: random_size8.tex
\begin{tikzpicture}[fftgraph]
	\matrix [fftmatrix] {
		& [-5mm]
		\node(a0) [fftterminal]{\small$a_{0}$}; &
		\node(a1) [fftterminal]{\small$a_{1}$}; &
		\node(a2) [fftterminal]{\small$a_{2}$}; &
		\node(a3) [fftterminal]{\small$a_{3}$}; &
		\node(a4) [fftterminal]{\small$a_{4}$}; &
		\node(a5) [fftterminal]{\small$a_{5}$}; &
		\node(a6) [fftterminal]{\small$a_{6}$}; &
		\node(a7) [fftterminal]{\small$a_{7}$}; &
		\\[-7mm]
		\node (stride8) [stride]{\normalsize\bf8}; & [-5mm]
		\node (n8_0_0) [fftnode,densely dotted]{\footnotesize\underline{0}\nodepart{second}\small\bf0.0\nodepart{third}\footnotesize\sl1  6}; &
		\node (n8_0_1) [fftnode,densely dotted]{\footnotesize\underline{0}\nodepart{second}\small\bf1.0\nodepart{third}\footnotesize\sl5  0}; &
		\node (n8_0_2) [fftnode,densely dotted]{\footnotesize\underline{0}\nodepart{second}\small\bf2.0\nodepart{third}\footnotesize\sl7  1}; &
		\node (n8_0_3) [fftnode,densely dotted]{\footnotesize\underline{0}\nodepart{second}\small\bf3.0\nodepart{third}\footnotesize\sl0  0}; &
		\node (n8_0_4) [fftnode,densely dotted]{\footnotesize\underline{0}\nodepart{second}\small\bf4.0\nodepart{third}\footnotesize\sl1  2}; &
		\node (n8_0_5) [fftnode,densely dotted]{\footnotesize\underline{0}\nodepart{second}\small\bf5.0\nodepart{third}\footnotesize\sl5  4}; &
		\node (n8_0_6) [fftnode,densely dotted]{\footnotesize\underline{0}\nodepart{second}\small\bf6.0\nodepart{third}\footnotesize\sl7  5}; &
		\node (n8_0_7) [fftnode,densely dotted]{\footnotesize\underline{0}\nodepart{second}\small\bf7.0\nodepart{third}\footnotesize\sl0  4}; &
		\\
		\node (stride4) [stride]{\normalsize\bf4}; & [-5mm]
		\node (n4_0_0) [fftnode]{\footnotesize\underline{1}  1\nodepart{second}\small\bf0.0\nodepart{third}\footnotesize\sl0  1}; &
		\node (n4_0_1) [fftnode]{\footnotesize\underline{5}  5\nodepart{second}\small\bf1.0\nodepart{third}\footnotesize\sl2  6}; &
		\node (n4_0_2) [fftnode]{\footnotesize\underline{7}  7\nodepart{second}\small\bf2.0\nodepart{third}\footnotesize\sl2  7}; &
		\node (n4_0_3) [fftnode]{\footnotesize\underline{0}  0\nodepart{second}\small\bf3.0\nodepart{third}\footnotesize\sl7  7}; &
		\node (n4_4_0) [fftnode]{\footnotesize\underline{6}  2\nodepart{second}\small\bf0.4\nodepart{third}\footnotesize\sl6  3}; &
		\node (n4_4_1) [fftnode]{\footnotesize\underline{0}  4\nodepart{second}\small\bf1.4\nodepart{third}\footnotesize\sl2  2}; &
		\node (n4_4_2) [fftnode]{\footnotesize\underline{1}  5\nodepart{second}\small\bf2.4\nodepart{third}\footnotesize\sl5  6}; &
		\node (n4_4_3) [fftnode]{\footnotesize\underline{0}  4\nodepart{second}\small\bf3.4\nodepart{third}\footnotesize\sl4  0}; &
		\\
		\node (stride2) [stride]{\normalsize\bf2}; & [-5mm]
		\node (n2_0_0) [fftnode]{\footnotesize\underline{1}  1\nodepart{second}\small\bf0.0\nodepart{third}\footnotesize\sl0  0}; &
		\node (n2_0_1) [fftnode]{\footnotesize\underline{7}  7\nodepart{second}\small\bf1.0\nodepart{third}\footnotesize\sl2  6}; &
		\node (n2_4_0) [fftnode]{\footnotesize\underline{2}  6\nodepart{second}\small\bf0.4\nodepart{third}\footnotesize\sl5  6}; &
		\node (n2_4_1) [fftnode]{\footnotesize\underline{3}  7\nodepart{second}\small\bf1.4\nodepart{third}\footnotesize\sl6  3}; &
		\node (n2_2_0) [fftnode]{\footnotesize\underline{4}  6\nodepart{second}\small\bf0.2\nodepart{third}\footnotesize\sl3  3}; &
		\node (n2_2_1) [fftnode]{\footnotesize\underline{2}  4\nodepart{second}\small\bf1.2\nodepart{third}\footnotesize\sl6  2}; &
		\node (n2_6_0) [fftnode]{\footnotesize\underline{1}  7\nodepart{second}\small\bf0.6\nodepart{third}\footnotesize\sl2  7}; &
		\node (n2_6_1) [fftnode]{\footnotesize\underline{2}  0\nodepart{second}\small\bf1.6\nodepart{third}\footnotesize\sl4  5}; &
		\\
		\node (stride1) [stride]{\normalsize\bf1}; & [-5mm]
		\node (n1_0_0) [fftnode]{\footnotesize\underline{1}  1\nodepart{second}\small\bf0.0\nodepart{third}\footnotesize\sl7}; &
		\node (n1_4_0) [fftnode]{\footnotesize\underline{1}  5\nodepart{second}\small\bf0.4\nodepart{third}\footnotesize\sl7}; &
		\node (n1_2_0) [fftnode]{\footnotesize\underline{7}  1\nodepart{second}\small\bf0.2\nodepart{third}\footnotesize\sl1}; &
		\node (n1_6_0) [fftnode]{\footnotesize\underline{0}  6\nodepart{second}\small\bf0.6\nodepart{third}\footnotesize\sl0}; &
		\node (n1_1_0) [fftnode]{\footnotesize\underline{7}  0\nodepart{second}\small\bf0.1\nodepart{third}\footnotesize\sl1}; &
		\node (n1_5_0) [fftnode]{\footnotesize\underline{7}  4\nodepart{second}\small\bf0.5\nodepart{third}\footnotesize\sl1}; &
		\node (n1_3_0) [fftnode]{\footnotesize\underline{3}  6\nodepart{second}\small\bf0.3\nodepart{third}\footnotesize\sl5}; &
		\node (n1_7_0) [fftnode]{\footnotesize\underline{0}  7\nodepart{second}\small\bf0.7\nodepart{third}\footnotesize\sl0}; &
		\\[-6mm]
		& [-5mm]
		\node(X0) [fftterminal]{\footnotesize$X(0)$}; &
		\node(X4) [fftterminal]{\footnotesize$X(4)$}; &
		\node(X2) [fftterminal]{\footnotesize$X(2)$}; &
		\node(X6) [fftterminal]{\footnotesize$X(6)$}; &
		\node(X1) [fftterminal]{\footnotesize$X(1)$}; &
		\node(X5) [fftterminal]{\footnotesize$X(5)$}; &
		\node(X3) [fftterminal]{\footnotesize$X(3)$}; &
		\node(X7) [fftterminal]{\footnotesize$X(7)$}; &
		\\
};
	\draw[arc,densely dotted] (a0.south) -- (n8_0_0.north);
	\draw[arc,densely dotted] (a4.south) -- (n8_0_4.north);
	\draw[arc,densely dotted] (a2.south) -- (n8_0_2.north);
	\draw[arc,densely dotted] (a6.south) -- (n8_0_6.north);
	\draw[arc,densely dotted] (a1.south) -- (n8_0_1.north);
	\draw[arc,densely dotted] (a5.south) -- (n8_0_5.north);
	\draw[arc,densely dotted] (a3.south) -- (n8_0_3.north);
	\draw[arc,densely dotted] (a7.south) -- (n8_0_7.north);
	\draw[arc] (n8_0_0.250) -- (n4_0_0.110);
	\draw[arc] (n8_0_0.290) -- (n4_4_0.110);
	\draw[arc] (n8_0_4.250) -- (n4_0_0.70);
	\draw[arc] (n8_0_4.290) -- (n4_4_0.70);
	\draw[arc] (n8_0_2.250) -- (n4_0_2.110);
	\draw[arc] (n8_0_2.290) -- (n4_4_2.110);
	\draw[arc] (n8_0_6.250) -- (n4_0_2.70);
	\draw[arc] (n8_0_6.290) -- (n4_4_2.70);
	\draw[arc] (n8_0_1.250) -- (n4_0_1.110);
	\draw[arc] (n8_0_1.290) -- (n4_4_1.110);
	\draw[arc] (n8_0_5.250) -- (n4_0_1.70);
	\draw[arc] (n8_0_5.290) -- (n4_4_1.70);
	\draw[arc] (n8_0_3.250) -- (n4_0_3.110);
	\draw[arc] (n8_0_3.290) -- (n4_4_3.110);
	\draw[arc] (n8_0_7.250) -- (n4_0_3.70);
	\draw[arc] (n8_0_7.290) -- (n4_4_3.70);
	\draw[arc] (n4_0_0.250) -- (n2_0_0.110);
	\draw[arc] (n4_0_0.290) -- (n2_4_0.110);
	\draw[arc] (n4_0_2.250) -- (n2_0_0.70);
	\draw[arc] (n4_0_2.290) -- (n2_4_0.70);
	\draw[arc] (n4_0_1.250) -- (n2_0_1.110);
	\draw[arc] (n4_0_1.290) -- (n2_4_1.110);
	\draw[arc] (n4_0_3.250) -- (n2_0_1.70);
	\draw[arc] (n4_0_3.290) -- (n2_4_1.70);
	\draw[arc] (n4_4_0.250) -- (n2_2_0.110);
	\draw[arc] (n4_4_0.290) -- (n2_6_0.110);
	\draw[arc] (n4_4_2.250) -- (n2_2_0.70);
	\draw[arc] (n4_4_2.290) -- (n2_6_0.70);
	\draw[arc] (n4_4_1.250) -- (n2_2_1.110);
	\draw[arc] (n4_4_1.290) -- (n2_6_1.110);
	\draw[arc] (n4_4_3.250) -- (n2_2_1.70);
	\draw[arc] (n4_4_3.290) -- (n2_6_1.70);
	\draw[arc] (n2_0_0.250) -- (n1_0_0.110);
	\draw[arc] (n2_0_0.290) -- (n1_4_0.110);
	\draw[arc] (n2_0_1.250) -- (n1_0_0.70);
	\draw[arc] (n2_0_1.290) -- (n1_4_0.70);
	\draw[arc] (n2_4_0.250) -- (n1_2_0.110);
	\draw[arc] (n2_4_0.290) -- (n1_6_0.110);
	\draw[arc] (n2_4_1.250) -- (n1_2_0.70);
	\draw[arc] (n2_4_1.290) -- (n1_6_0.70);
	\draw[arc] (n2_2_0.250) -- (n1_1_0.110);
	\draw[arc] (n2_2_0.290) -- (n1_5_0.110);
	\draw[arc] (n2_2_1.250) -- (n1_1_0.70);
	\draw[arc] (n2_2_1.290) -- (n1_5_0.70);
	\draw[arc] (n2_6_0.250) -- (n1_3_0.110);
	\draw[arc] (n2_6_0.290) -- (n1_7_0.110);
	\draw[arc] (n2_6_1.250) -- (n1_3_0.70);
	\draw[arc] (n2_6_1.290) -- (n1_7_0.70);
	\draw[arc] (n1_0_0.south) -- (X0.north);
	\draw[arc] (n1_4_0.south) -- (X4.north);
	\draw[arc] (n1_2_0.south) -- (X2.north);
	\draw[arc] (n1_6_0.south) -- (X6.north);
	\draw[arc] (n1_1_0.south) -- (X1.north);
	\draw[arc] (n1_5_0.south) -- (X5.north);
	\draw[arc] (n1_3_0.south) -- (X3.north);
	\draw[arc] (n1_7_0.south) -- (X7.north);
	\begin{pgfonlayer}{background}
		\node (rn8_0_0) [background,fit=(n8_0_0) (n8_0_7)] {};
		\node (rn4_0_0) [background,fit=(n4_0_0) (n4_0_3)] {};
		\node (rn4_4_0) [background,fit=(n4_4_0) (n4_4_3)] {};
		\node (rn2_0_0) [background,fit=(n2_0_0) (n2_0_1)] {};
		\node (rn2_4_0) [background,fit=(n2_4_0) (n2_4_1)] {};
		\node (rn2_2_0) [background,fit=(n2_2_0) (n2_2_1)] {};
		\node (rn2_6_0) [background,fit=(n2_6_0) (n2_6_1)] {};
		\node (rn1_0_0) [background,fit=(n1_0_0) (n1_0_0)] {};
		\node (rn1_4_0) [background,fit=(n1_4_0) (n1_4_0)] {};
		\node (rn1_2_0) [background,fit=(n1_2_0) (n1_2_0)] {};
		\node (rn1_6_0) [background,fit=(n1_6_0) (n1_6_0)] {};
		\node (rn1_1_0) [background,fit=(n1_1_0) (n1_1_0)] {};
		\node (rn1_5_0) [background,fit=(n1_5_0) (n1_5_0)] {};
		\node (rn1_3_0) [background,fit=(n1_3_0) (n1_3_0)] {};
		\node (rn1_7_0) [background,fit=(n1_7_0) (n1_7_0)] {};
	\end{pgfonlayer}
	\node [costtext] at (n8_0_0.south east) {\tiny4};
	\node [costtext] at (n8_0_4.south east) {\tiny4};
	\node [costtext] at (n8_0_2.south east) {\tiny4};
	\node [costtext] at (n8_0_6.south east) {\tiny4};
	\node [costtext] at (n8_0_1.south east) {\tiny4};
	\node [costtext] at (n8_0_5.south east) {\tiny4};
	\node [costtext] at (n8_0_3.south east) {\tiny0};
	\node [costtext] at (n8_0_7.south east) {\tiny0};
	\node [costtext] at (n4_0_0.south east) {\tiny6};
	\node [costtext] at (n4_0_2.south east) {\tiny6};
	\node [costtext] at (n4_0_1.south east) {\tiny2};
	\node [costtext] at (n4_0_3.south east) {\tiny6};
	\node [costtext] at (n4_4_0.south east) {\tiny6};
	\node [costtext] at (n4_4_2.south east) {\tiny6};
	\node [costtext] at (n4_4_1.south east) {\tiny2};
	\node [costtext] at (n4_4_3.south east) {\tiny2};
	\node [costtext] at (n2_0_0.south east) {\tiny2};
	\node [costtext] at (n2_0_1.south east) {\tiny2};
	\node [costtext] at (n2_4_0.south east) {\tiny6};
	\node [costtext] at (n2_4_1.south east) {\tiny6};
	\node [costtext] at (n2_2_0.south east) {\tiny6};
	\node [costtext] at (n2_2_1.south east) {\tiny2};
	\node [costtext] at (n2_6_0.south east) {\tiny6};
	\node [costtext] at (n2_6_1.south east) {\tiny6};
	\node [costtext] at (n1_0_0.south east) {\tiny6};
	\node [costtext] at (n1_4_0.south east) {\tiny6};
	\node [costtext] at (n1_2_0.south east) {\tiny6};
	\node [costtext] at (n1_6_0.south east) {\tiny2};
	\node [costtext] at (n1_1_0.south east) {\tiny6};
	\node [costtext] at (n1_5_0.south east) {\tiny6};
	\node [costtext] at (n1_3_0.south east) {\tiny6};
	\node [costtext] at (n1_7_0.south east) {\tiny2};
\end{tikzpicture}